\documentclass[submission, Phys]{SciPost}

\usepackage{graphicx}
\usepackage{dcolumn}
\usepackage{bm}
\usepackage{color}
\usepackage{physics}
\usepackage{epstopdf}
\usepackage{amsmath}
\usepackage{amsfonts}
\usepackage{amssymb}
\usepackage{mathrsfs}
\usepackage{amsthm}

\usepackage[caption=false,subrefformat=parens,labelformat=parens]{subfig}
\usepackage{stackengine}
\usepackage{float}
\usepackage{upgreek}

\usepackage{enumerate}
\newtheorem{theorem}{Theorem}
\newtheorem{lemma}{Lemma}

\newtheorem{proposition}{Proposition}

\begin{document}

\begin{center}{\Large \textbf{
Universal tradeoff relation between speed, uncertainty, and dissipation  in  nonequilibrium  stationary  states 
}}\end{center}

\begin{center}
Izaak Neri
\end{center}

\begin{center}
Department of Mathematics, King’s College London, Strand, London, WC2R 2LS, UK
\\
\end{center}

\begin{center}
\today
\end{center}


\section*{Abstract}
{\bf
We derive  universal thermodynamic  inequalities that bound  from below  the moments of   first-passage times of stochastic  currents in nonequilibrium stationary states of Markov jump processes  in the limit  where the two thresholds that define the  first-passage problem are large.    These     inequalities  describe  a  tradeoff between speed,   uncertainty, and dissipation in nonequilibrium processes, which are quantified, respectively, with the  moments of the  first-passage times of stochastic currents,  the splitting probability of the first-passage problem, and the mean entropy production rate.        Near equilibrium, the inequalities imply that mean first-passage times are lower bounded by the Van't Hoff-Arrhenius law, whereas far from thermal equilibrium the bounds describe a universal speed limit for rate processes.     When the    current is proportional to the  stochastic entropy production, then the bounds are equalities, a remarkable property that follows from the fact that the exponentiated negative entropy production is a martingale.     }

\vspace{10pt}
\noindent\rule{\textwidth}{1pt}
\tableofcontents\thispagestyle{fancy}
\noindent\rule{\textwidth}{1pt}
\vspace{10pt}

\section{Introduction}

  In thermal equilibrium transitions between  metastable states are activated  by    thermal fluctuations.  The  equilibrium transition rates satisfy the   Van't Hoff-Arrhenius law~\cite{hanggi1990reaction, mccann1999thermally}
\begin{equation}
k=  \frac{1}{\langle T\rangle} = \nu e^{-\frac{E_{\rm b}}{\mathsf{T}_{\rm env}}}, \label{eq:Arrh}
\end{equation} 
where the rate $k$ is the inverse of the mean first-passage time $\langle T\rangle$, $E_{\rm b}$ is the energy barrier that separates the two metastable states,  $\mathsf{T}_{\rm env}$ is the temperature of the environment, and $\nu$ is a prefactor that has been determined, among others, by  Kramers \cite{kramers1940brownian, hanggi1990reaction}.    

To speed up a process, an external agent  can drive a system out of equilibrium.  For example,  in Fig.~\ref{fig1Mx} we illustrate how external driving can increase the reaction rate in a nonequilibrium version of   Kramers' model \cite{kramers1940brownian}.     Other examples  are the reduced travel times of self-propelled particles~\cite{angelani2014first, malakar2018steady, dhar2019run, biswas2020first, walter2021first},    the activated escape of a particle from a metastable state \cite{PhysRevE.73.061109},  enhanced relaxation rates in biomolecular diffusion processes \cite{godec2016active},  and  enhanced  reaction rates in nonequilibrium chemical reactions \cite{loverdo2008enhanced, siggia2013decisions, desponds2020mechanism, biswas2021first}.  
       Since dissipation can increase the rate of a process, one may wonder whether there exists a generic  speed limit on processes that are driven away from thermal equilibrium.     

In the present paper, building on  Ref.~\cite{roldan2015decision},  we  show that rate processes  are governed by   a universal   tradeoff between dissipation, speed, and uncertainty.      We quantify this tradeoff with generic inequalities on the moments of the   first-passage times of stochastic currents with two thresholds.   The derived inequalities are  reminiscent of the thermodynamic uncertainty relations for first-passage times \cite{gringich2017bis}, but there exist also a couple of important distinctions.   First, the tradeoff relations derived in this paper quantify the uncertainty in the outcome of the process with the splitting probability of the first-passage problem,   whereas the thermodynamic uncertainty relation quantifies uncertainty  with the variance of the first-passage time.    Second, the derived bounds are  equalities when the  current is the stochastic entropy production, and hence the derived first-passage inequalities are optimal  in this case.

 The paper is organised as follows: in Sec.~\ref{sec:main}, we state the main results of this paper.   In Sec.~\ref{sec:setup}, we discuss the  setup for which the main results are derived, viz.,  stochastic currents in Markov jump processes.   In Sec.~\ref{sec:fp}, we derive  the main results,   within the setup of Markov jump processes, by using recent results on large deviations and martingales in stochastic thermodynamics.   In Sec.~\ref{sec:seq} we provide an alternative derivation that is based on    the theory of sequential hypothesis testing and which  provides insights on  extensions of the main results beyond Markov jump processes.  
 In the following two Secs.~\ref{sec:prev} and \ref{sec:Arrh}, we relate the main results  of this paper to results previously published  in the literature and   to the  Van't Hoff-Arrhenius law, respectively.  In Sec.~\ref{sec:inf}, we illustrate with an example the tightness of the  first-passage time bounds when the stochastic current is proportional to the stochastic entropy production.       The paper ends with a discussion in Sec.~\ref{sec:discu} and after the discussion there are   several appendices that contain technical details on the mathematical derivations.

 \section{Main results} \label{sec:main}
 
 The paper contains two main results.  The first main result  is an inequality that holds  for the first-passage times of  stochastic currents in stationary Markov jump processes.  The second main  result is an equality that holds for first-passage times of stochastic currents that are proportional to the stochastic entropy production. 
 
 \subsection{Bounds on the moments of first-passage times of stochastic currents}
  
   Let $J(t)$ be a stochastic current in a nonequilibrium, stationary process $X(t)$ and let 
   \begin{equation}
T_J = {\rm inf}\left\{t>0: J(t)\notin (-\ell_-,\ell_+)\right\} \label{eq:def}
\end{equation}
be the first time when $J(t)$ leaves the open interval $(-\ell_-,\ell_+)$, where $t\geq0$ is an index that labels the time and where $\ell_-,\ell_+>0$ are  the threshold values of the first-passage problem.

\begin{figure}[t!]
\centering
\includegraphics[width=0.7\textwidth]{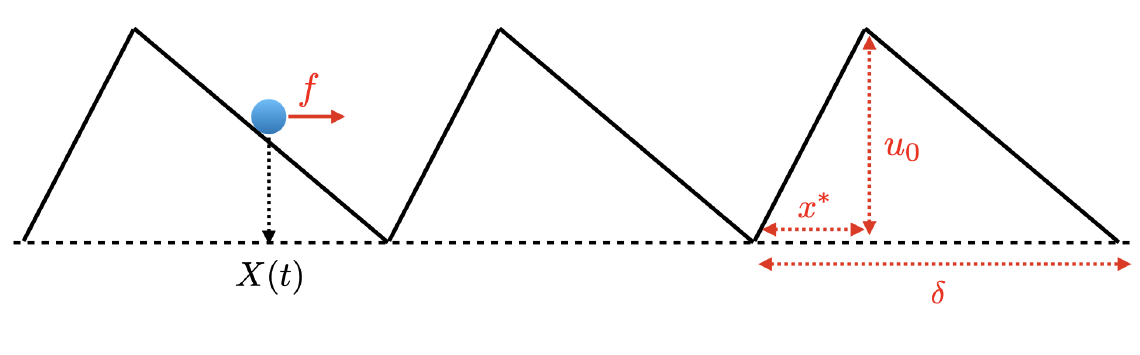}
\includegraphics[width=0.45\textwidth]{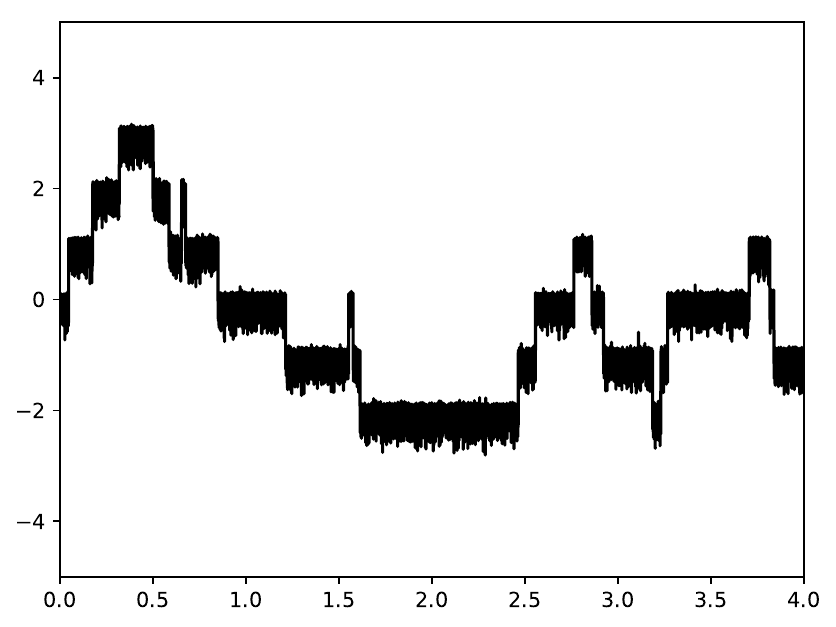}
\includegraphics[width=0.45\textwidth]{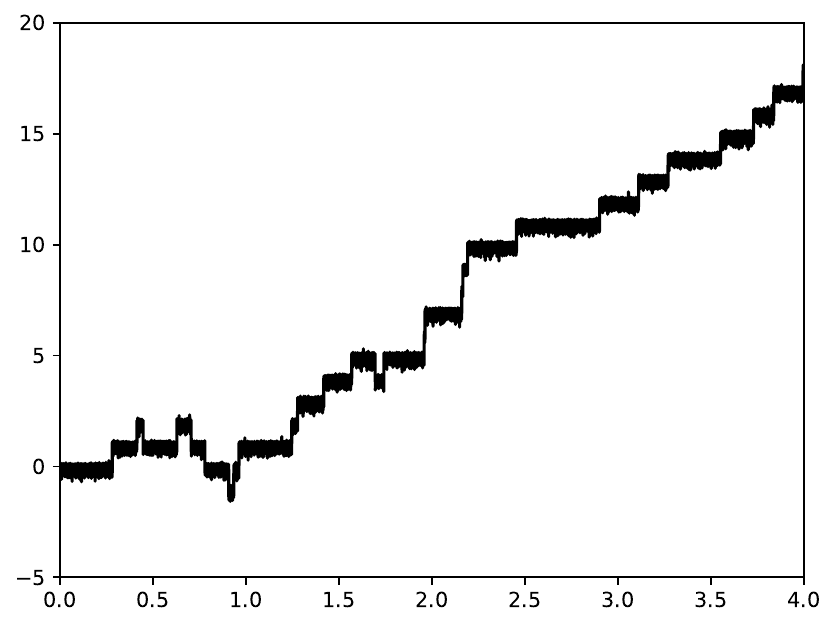}
 \put(-420,70){$\frac{X(t)}{\delta}$}
  \put(-308,-10){$t \: e^{-u_0/\mathsf{T}_{\rm env}}$} 
    \put(-120,-10){$t \: e^{-u_0/\mathsf{T}_{\rm env}}$} 
\caption{ {\it  Nonequilibrium version of Kramers' model exhibiting an increased  reaction rate due to nonequilibrium driving}.    Trajectories shown are for  a reaction coordinate $X$  that solves the Langevin equation 
$\partial_t X(t) = (f-\partial_x u(X(t)))/\gamma +  \sqrt{2  \mathsf{T}_{\rm env}/\gamma} \xi(t)$, where $\xi(t) = {\rm d}W(t)/{\rm d}t$ is a delta-correlated white Gaussian noise term, and where $u(x)$ is a triangular potential with period $\delta$, i.e. $u(x) = u(x\pm \delta)$,  $u(x) = u_0\: x/x^\ast$ if $x\in[0,x^\ast]$, and  $u(x) = u_0 (\delta-x)/(\delta-x^\ast)$ if $x\in[x^\ast,\delta]$.  Left: equilibrium trajectory with $f=0$. Right: nonequilibrium trajectory with $f \delta/\mathsf{T}_{\rm env}=1$.  
   The remaining parameters   are   set to $\delta=5$, $\gamma=1$, $x^\ast=1$,  $u_0=10$, and $\mathsf{T}_{\rm env}=1$.      } \label{fig1Mx}
\end{figure}

\begin{figure}[t!]
\centering
\includegraphics[width=0.5\textwidth]{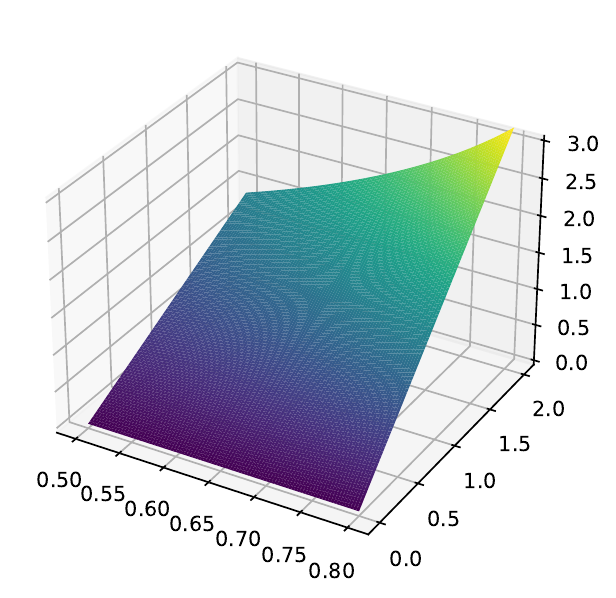}
 \put(-170,5){\Large$1-p_-$}
 \put(-35,30){\Large$1/\langle T_J\rangle$}
  \put(0,120){\Large$\dot{s}$}
     \put(-83,40){{\rotatebox{60}  {\color{red}  {\rm \bf non-permissible}}}}
\caption{ {\it  Universal tradeoff between speed, uncertainty, and dissipation in nonequilibrium processes}.   The three axes represent the speed ($1/\langle T_J\rangle$), uncertainty ($1-p_-$), and dissipation ($\dot{s}$) in a nonequilibrium process $X$.   The plotted surface is  $\dot{s} = |\log p_-|/\langle T_J\rangle$; in the present example the thresholds are symmetric, $\ell_-=\ell_+$.  Processes that are situated below the surface are physically nonpermissible as they violate the bound Eq.~(\ref{eq:2}).  } \label{figTradeoff}
\end{figure}
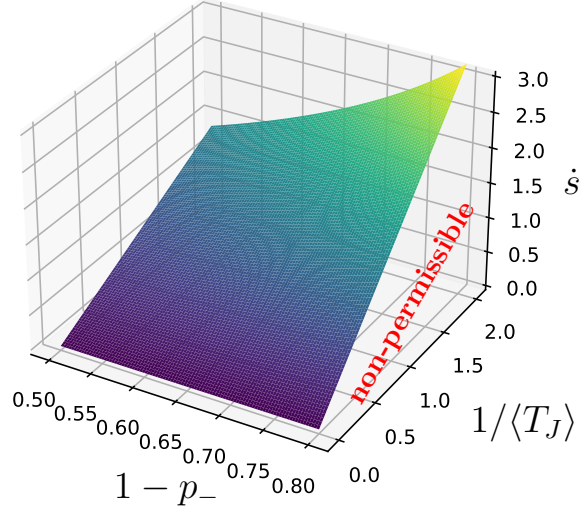  
In this paper we show that in  the limit of large thresholds $-\ell_-$ and $\ell_+$ it holds that 
 \begin{eqnarray}
\langle T^n_J\rangle  \geq  \left(\frac{\ell_+}{\ell_-}\frac{|    \log p_- |}{ \dot{s} }  \right)^n(1+o_{\ell_{\rm min}}(1)), \label{eq:2}
\end{eqnarray} 
where 
\begin{equation}
p_- = \mathbb{P}\left[J(T_J)\leq -\ell_-\right]
\end{equation}
denotes the  probability that the current $J$ goes below  the negative threshold $-\ell_-$ before exceeding for the first time the positive threshold $\ell_+$, where $\dot{s}$ is the entropy production rate, and where $n\in\mathbb{N}$.   The quantity  $p_-$ is  called the splitting probability.     The averages $\langle \cdot\rangle$ are taken over repeated realisations of the stationary process $X$.   We have used the little-o-notation    $o_{\ell_{\rm min}}(1)$ to denote a function that converges to zero when  $\ell_{\rm min} = {\rm min}\left\{\ell_-,\ell_+\right\}\rightarrow \infty$ while the ratio  $\ell_-/\ell_+$ is kept fixed.    Since we keep the ratio   $\ell_-/\ell_+$ fixed,   it holds that $o_{\ell_{\rm min}}(1) = o_{\ell_-}(1) = o_{\ell_+}(1)$.  Equation (\ref{eq:2})   holds for  $\langle J(t)\rangle>0$; if $\langle J(t)\rangle<0$, then  $p_-$ should be replaced by $p_+ = \mathbb{P}\left[J(T_J)\geq\ell_+\right]$,  $\ell_-$ with $\ell_+$, and vice versa. 

The inequality Eq.~(\ref{eq:2}) describes 
  a tradeoff between dissipation $\dot{s}$, speed $\langle T^n\rangle$, and the   uncertainty in the outcome of the process that is  quantified  by $p_-$.    
It states that processes that are fast, precise, and have a small entropy production rate are physically not permissible.   In Fig.~\ref{figTradeoff} we  illustrate this tradeoff relation graphically by plotting a surface in a three-dimensional space delimiting the  parameter regime that is physically not permissible.

    Near equilibrium $\dot{s}\sim e^{-\frac{E_{\rm b}}{\mathsf{T}_{\rm env}}}$ and $p_-\approx \ell_+/(\ell_++\ell_-)$.  Consequently, Eq.~(\ref{eq:2}) implies that  $\langle T_J\rangle$ is lower bounded by the Van't Hoff-Arrhenius law, i.e., 
\begin{equation}
\langle T_J\rangle \geq \frac{1}{\nu} e^{\frac{E_{\rm b}}{\mathsf{T}_{\rm env}}}.   \label{eq:Arrh2}
\end{equation}
On the other hand, far from thermal equilibrium   the right hand side of Eq.~(\ref{eq:2}) goes below $\frac{1}{\nu} e^{\frac{E_{\rm b}}{\mathsf{T}_{\rm env}}}$ implying that dissipation can increase the reaction  rate $k = 1/\langle T_J\rangle$, as we illustrate  in Fig.~\ref{fig1Mx} for a  nonequilibrium version of Kramers' model \cite{kramers1940brownian}.   

Taken together, the Eq.~(\ref{eq:2})  states  that we can speed up a process by driving it out of equilibrium,   but there exists a universal  speed limit that is determined by the rate of dissipation and the amount of  fluctuations  in the process.  

 \subsection{Equality for the moments of first-passage times of entropy production}

If $J(t)=S(t)$ with $S(t)$ the stochastic entropy production \cite{maes2003origin, jarzynski2011equalities, seifert2012stochastic}, then the equality sign in Eq.~(\ref{eq:2}) holds, viz.,   
 \begin{eqnarray}
\langle T^n_S\rangle  =  \left(\frac{\ell_+}{\ell_-}\frac{|    \log p_- |}{ \dot{s} }  \right)^n(1+o_{\ell_{\rm min}}(1)) .\label{eq:eq}
\end{eqnarray} 
This remarkable property  follows from the fact that  $e^{-S(t)}$ is a martingale \cite{chetrite2011two, neri2017statistics, pigolotti2017generic, neri2019integral}, which implies the formula $p_- = e^{-\ell_-}(1+o_{\ell_{\rm min}}(1))$ \cite{neri2017statistics, neri2019integral}.

Note that the definition (\ref{eq:def})    together with (\ref{eq:eq})   implies that the equality sign in Eq.~(\ref{eq:2})  also holds when $J(t)=cS(t)$, with $c$ a constant that is independent of $\ell_-$ and $\ell_+$.

The   Eq.~(\ref{eq:eq}) implies that the bound Eq.~(\ref{eq:2})  is tight when the stochastic current is proportional to the stochastic  entropy production ($J=cS$), and this is one of the main advantages of the bound (\ref{eq:2})  with respect to other  tradeoff inequalities reported in the literature, such as, the thermodynamic uncertainty relation  for first-passage times  that quantifies uncertainty in terms of the variance of the first-passage time~\cite{gringich2017bis}.

\section{System setup} \label{sec:setup}

We consider a stationary Markov jump process $X(t)$ defined on a discrete set $\mathcal{X}\ni X(t)$ and in continuous time $t\geq0$.   The dynamics of $X(t)$ consists of a sequence of jumps with  rates that are determined by a Markov transition rate matrix   $w_{x\rightarrow y}$ with $x,y\in\mathcal{X}$ \cite{schnakenberg1976network}.      We assume that $X(t)$ has a unique stationary probability distribution $p_{\rm ss}(x)$ that satisfies $p_{\rm ss}(x)>0$ for all $x\in\mathcal{X}$, and we assume that the process is reversible in the sense that $w_{x\rightarrow y}>0$ if and only if $w_{y\rightarrow x}>0$.

Stochastic   currents $J(t) = J(X^t_0)$  are real-valued  functionals defined on the set of trajectories $X^t_0$ with the following two properties: 
\begin{enumerate}[(i)]
\item  $J$ is time extensive, i.e., 
\begin{equation}
\langle J(t) \rangle = \overline{j} \: t 
\end{equation} 
where $\overline{j}$ is a nonzero current rate.    Without loss of generality we can assume that $\overline{j}>0$.

\item $J$ is odd under time-reversal, i.e., 
\begin{equation}
J(\Theta_t(X^t_0))  = -J(X^t_0),
\end{equation}
where the time-reversal operation $\Theta_t$ maps trajectories $X^t_0$ on their time-reversed trajectory $(X^\dagger)^t_0$ with entries $X^\dagger(\tau) = X(t-\tau)$.
   Note that this implies $J(0) = 0$.   
\end{enumerate}    In a Markov jump process, stochastic currents take the form 
\begin{equation}
J(t)  = \sum_{x,y\in \mathcal{X}}c_{x,y}J_{x\rightarrow y}(t), \label{eq:JMJ}
\end{equation}
with coefficients $c_{x,y}\in \mathbb{R}$ and  with $c_{x,x}=0$.   The edge currents
\begin{equation}
J_{x\rightarrow y}(t) = N_{x\rightarrow y}(t) - N_{y\rightarrow x}(t) \label{eq:edge}
\end{equation}
denote  the difference between the number of times $N_{x\rightarrow y}(t)$  the process  has jumped from the $x$-th state to the $y$-th state in the trajectory $X^t_0$ and the  number of reverse jumps $N_{y\rightarrow x}(t)$ from the $y$-th to the $x$-th state in the same trajectory.       

The stochastic entropy production $S$ is defined by the ratio~\cite{maes2003origin, seifert2012stochastic}
\begin{equation}
S(t) = \log \frac{p(X^t_0)}{p(\Theta_t(X^t_0)) } \label{eq:probRatio}
\end{equation}  
 between the probability distributions of the trajectory $X^t_0$ in the forward and backward dynamics, better known as the Radon-Nikodym derivative \cite{maes2000definition, neri2017statistics, yang2020unified}.     For a stationary process, the index $t$ in the map $\Theta_t$ of Eq.~(\ref{eq:probRatio})  is immaterial, and we can replace $\Theta_t$ by $\Theta$.    Notice that we   use natural units for which the Boltzmann constant is set equal to one. 
It is possible to write  the  stochastic entropy production in the form Eq.~(\ref{eq:JMJ}), viz., 
\begin{equation}
S(t)  = \frac{1}{2} \sum_{x,y\in \mathcal{X}} \log \frac{p_{\rm ss}(x)w_{x\rightarrow y}}{p_{\rm ss}(y)w_{y\rightarrow x}} J_{x\rightarrow y}(t),
\end{equation}
where  $p_{\rm ss}(x)$ is the probability distribution of $X(t)$ in the stationary state.     In the definition of the entropy production we require that the  process is reversible, i.e., if $w_{x\rightarrow y}>0$ then also $w_{y\rightarrow x}>0$.    A useful property that we will use repeatedly is that the  exponentiated negative entropy production $e^{-S(t)}$ is a martingale, see \cite{chetrite2011two, neri2017statistics, pigolotti2017generic, neri2019integral}.

Since the process is stationary, the entropy production rate $\dot{s}$ is given by 
\begin{equation}
\langle S(t)\rangle = \dot{s} \: t .  \label{eq:sdotDef}
\end{equation}     
For systems that are weakly coupled to an environment in thermal equilibrium, the entropy production rate equals the dissipation rate \cite{maes2003origin, seifert2012stochastic, peliti2021stochastic}, which clarifies the physical significance of the process $S(t)$.     In the literature, the latter property is often referred to as  the principle of local detailed balance~\cite{maes2020local, hartich2021violation}.

\section{First-passage time bounds from large deviation theory}\label{sec:fp}  
We derive the main results of this paper,  given by Eqs.~(\ref{eq:2}) and (\ref{eq:eq}), with large deviation theory.

  Stochastic currents $J(t)$ in Markov jump processes satisfy  a large deviation principle.     This means that  for large enough times $t$, the probability distribution of $J/t$ takes the form  \cite{touchette2009large, barato2015formal}
\begin{equation}
p_{J/t}(z) =  e^{- t\mathcal{J}(z) (1+o_t(1))},  \label{eq:pjtz}
\end{equation} 
where $o_t(1)$ is a function that converges to zero when $t$ is large enough, and where $\mathcal{J}(z)$ is the large deviation function of the current.  In Eq.~(\ref{eq:pjtz}), the normalisation constant is contained  in the  $o_{t}(1)$  term that appears in the argument of the exponential.   The large deviation function $\mathcal{J}(z)\geq0$ is a convex function that  takes its minimum value when  $J/t = \overline{j}$, i.e., $\mathcal{J}(\overline{j}) = 0$.   

An immediate consequence of Eq.~(\ref{eq:pjtz}) is that
\begin{equation}
\langle T^n_J\rangle = \left( \frac{\ell_+}{\overline{j}}\right)^n  (1+o_{\ell_{\rm min}}(1)).    \label{eq:Tna}
\end{equation}      
Indeed, since $J(t)$ satisfies the large-deviation principle Eq.~(\ref{eq:pjtz}),  $J(t)$ converges with probability one to $\overline{j}t$, viz., 
\begin{equation}
\frac{J(t)}{t} = \overline{j}(1+o_t(1)).
\end{equation}  
 Consequently,   the first-passage time  given by Eq.~(\ref{eq:def}) is deterministic for large values of $\ell_{\rm min}$, and since  $\overline{j}>0$  we obtain 
\begin{equation}
T_J = \frac{\ell_+}{\overline{j}} (1+o_{\ell_{\rm min}}(1)),  \label{eq:T}
\end{equation}
which implies  Eq.~(\ref{eq:Tna}), as long as for finite threshold values $\ell_{\rm min}$ the distribution of $T_J$ has fast enough decaying tails.

To complete the derivation of the main results, we  derive in Sec.~\ref{sec:bounds} a lower bound for the splitting probability $p_-$, in particular, we show that 
 \begin{eqnarray}
p_- &\geq & \exp\left(- \frac{\ell_- \dot{s}}{\overline{j}} (1+o_{\ell_{\rm min}}(1)) \right), \label{eq:pminusBounda}
\end{eqnarray}   
which together with  (\ref{eq:Tna}) implies Eq.~(\ref{eq:2}).   

 In Sec.~\ref{sec:equal} we show that for $J=S$ the  inequality (\ref{eq:pminusBounda}) becomes an equality, leading to~(\ref{eq:eq}). 
\subsection{Bound on the splitting probability $p_-$}  \label{sec:bounds}

We derive the bound Eq.~(\ref{eq:pminusBounda})  for the  probability $p_-$ that $J$ hits the negative boundary first, 
which together with  (\ref{eq:Tna}) readily implies  the main result Eq.~(\ref{eq:2}).

For stationary Markov jump processes, it was shown that $\mathcal{J}(z)$   is bounded from above by \cite{pietzonka2016universalx, gingrich2016dissipation, pietzonka2016affinity}
\begin{equation}
 \mathcal{J}(z) \leq  \frac{\dot{s}}{4}(z/\overline{j}-1)^2 \label{eq:JBoundx}.  
\end{equation}
In what follows, we show that the inequality (\ref{eq:pminusBounda}) follows from this fundamental bound.

The splitting probability $p_-$ can be expressed as follows, 
\begin{eqnarray} 
p_- = \mathbb{P}\left[J(T_J)\leq -\ell_-\right] = \mathbb{P}\left[J(t)\leq -\ell_-\right]  - \mathbb{P}\left[J(t)\leq -\ell_- \wedge  J(T_J)\geq \ell_+\right]
\nonumber\\
 +  \mathbb{P}\left[J(t)\geq -\ell_-\wedge J(T_J)\leq -\ell_-\right] ,
\end{eqnarray} 
where $\wedge$ is a short hand notation for the logical conjunction. 
Since probabilities are nonnegative, we obtain the bound
\begin{eqnarray} 
p_- \geq  \mathbb{P}\left[J(t)\leq -\ell_-\right] - \mathbb{P}\left[ J(t)\leq -\ell_- \wedge J(T_J)\geq \ell_+\right] .  
\end{eqnarray} 
Moreover, using that for large enough thresholds  the probability that $J(t)$  goes below  the threshold $-\ell_-$ after it went above the threshold $\ell_+$ is vanishingly small, we obtain the inequality
\begin{equation} 
p_- \geq  \mathbb{P}\left[J(t)\leq -\ell_-\right](1+o_{\ell_{\rm min}}(1)) \label{eq:pMPxx}
\end{equation} 
that holds for all $t\geq0$.

We can express the right-hand side of Eq.~(\ref{eq:pMPxx}) in terms of $p_{J/t}(z)$, i.e., 
\begin{equation}
p_- \geq  \int^{-\ell_-/t}_{-\infty} {\rm d}z \:  p_{J/t}(z) =  \int^{-\ell_-/t}_{-\infty} {\rm d}z \:  e^{- t\mathcal{J}(z) (1+o_t(1))}. \label{eq:pMM}
\end{equation}

Using the bound Eq.~(\ref{eq:JBoundx}) in Eq.~(\ref{eq:pMM}) and setting $\tau = t/\ell_-$, we obtain
\begin{equation}
p_- \geq   \int^{-1/\tau}_{-\infty} {\rm d}z \:  \exp\left(- \frac{1}{4}\ell_- \dot{s}\tau \left[\frac{z}{\overline{j}}-1\right]^2 (1+o_{\ell_{\rm min}}(1))\right),\label{eq:qM}
\end{equation}
where we have also interchanged $o_{t}(1) $ with $o_{\ell_{\rm min}}(1)$.   This is possible since the results of this paper hold for  $\ell_{\rm min}\rightarrow \infty$ while keeping the ratio $\ell_-/\ell_+$ fixed.   In Eq.~(\ref{eq:qM}) this limit corresponds with $\ell_-\rightarrow \infty$ while keeping the ratio $t/\ell_-  =\tau$ fixed.       In this limit, it holds that  $o_{t}(1) = o_{\ell_{\rm min}}(1)$, and therefore we can interchange these two symbols.

  For large values of $\ell_-$,  the expression Eq.~(\ref{eq:qM}) is a saddle point integral, and hence it is determined by   the maximum of the exponent, i.e.,
   \begin{equation}
p_-\geq  \exp\left(- \frac{1}{4}\ell_- \dot{s}\tau\left[\frac{1}{\tau \overline{j}}+1\right]^2 (1+o_{\ell_{\rm min}}(1))\right). 
   \end{equation} 
Since the above inequality holds for arbitrary $\tau$, we can  take the maximum of the right-hand side, 
      \begin{equation}
       p_- \geq {\rm max}_{\tau\geq 0} \exp\left(- \frac{1}{4}\ell_- \dot{s}\tau \left[\frac{1}{\tau \overline{j}}+1\right]^2 (1+o_{\ell_{\rm min}}(1))\right).
          \end{equation}  
       For $\tau \geq0$,       the minimum value of the function 
          $\tau\left(\frac{1}{\tau \overline{j}}+1\right)^2$ is   reached when $\tau=1/\overline{j}$, leading to  the bound Eq.~(\ref{eq:pminusBounda}) that we were meant to derive.

\subsection{An equality for $p_-$ that follows from the martingality of $e^{-S}$}\label{sec:equal} 
To derive Eq.~(\ref{eq:eq}), we show that when $J=S$, then the equality in Eq.~(\ref{eq:pminusBounda}) is satisfied, i.e., 
\begin{equation}
p_- =  e^{-\ell_-(1+o_{\ell_{\rm min}}(1))}. \label{eq:pminusEnt}
\end{equation}  
Eq.~(\ref{eq:pminusEnt}) together with  Eq.~(\ref{eq:Tna}), readily implies Eq.~(\ref{eq:eq}).

The fact that $p_-$ is universal and only depends on the threshold $\ell_-$ is a remarkable fact  that is a direct consequence of the martingale property of $e^{-S(t)}$ \cite{chetrite2011two, neri2017statistics, neri2019integral}.  Indeed, since the process  $e^{-S(t)}$ is a martingale and since $T_S$ is a first-passage time with two thresholds, the integral fluctuation relation at stopping times  \cite{neri2019integral}
\begin{equation}
\langle e^{-S(T_S)}\rangle = 1\label{eq:Int}
\end{equation} 
applies, see Corollary 2 of the Appendix of Ref.~\cite{neri2019integral}; a related, albeit not identical, relation was reported in \cite{saito2016waiting}.
Using that $\mathbb{P}[T_S<\infty]=1$, the Eq.~(\ref{eq:Int}) also reads
\begin{equation}
p_- \langle e^{-S(T_S)} \rangle_- + p_+  \langle e^{-S(T_S)}\rangle_+  = 1,
\end{equation}   
where $\langle \cdot \rangle_-$ and  $\langle \cdot \rangle_+$ denote averages over those trajectories that terminate at the negative and positive threshold values, respectively.
Using that for $\ell_-,\ell_+\gg1$ it holds that $S(T_S) = \ell_\pm (1+o_{\ell_{\rm min}}(1))$,  we  obtain 
\begin{equation}
p_- e^{\ell_-(1+o_{\ell_{\rm min}}(1))}  + p_+  e^{-\ell_+(1+o_{\ell_{\rm min}}(1))}  = 1, 
\end{equation}   
and  for $\ell_+\gg 1$ this simplifies into Eq.~(\ref{eq:pminusEnt}).

\section{First-passage time bounds  from the asymptotic optimality of sequential probability ratio tests}\label{sec:seq}  
In the previous section, we have derived the main results  Eqs.~(\ref{eq:2})  and (\ref{eq:eq})  within the setup of stationary Markov jump processes.   In the present section, we derive the main results  within the framework of   sequential hypothesis testing.  With sequential hypothesis testing theory,  we can derive   partial results in an extremely general setting.   These partial results are interesting in their own right, and they also pave the way to derive Eqs.~(\ref{eq:2})  and (\ref{eq:eq}) in a setup more general than  Markov jump processes.  
  \subsection{Review of sequential hypothesis testing}
  
As pointed out in Ref.~\cite{roldan2015decision}, first-passage problems for stochastic currents with two thresholds are sequential hypothesis tests that decide on the arrow of time, and     first-passage problems for entropy production are sequential probability ratio tests.    Therefore, we can use the theory of sequential hypothesis testing to derive bounds on the moments of first-passage times of stochastic currents.   We  provide a brief review  of the theory of sequential hypothesis testing, focusing on the asymptotic optimality of sequential probability ratio tests.

     Sequential hypothesis tests are statistical hypothesis tests that take a  decision   $D$ about the true hypothesis $H$ at a  random stopping time $T$.    The general setup goes as follows \cite{Melsa, tartakovsky2014sequential}.  There is an observation process $X(t)$ whose statistics are determined by one of two possible probability measures $\mathbb{P}_+$  or $\mathbb{P}_-$ corresponding to two hypotheses $H=+$ and $H=-$, respectively.    A sequential hypothesis test is a pair $(T,D)$, where $T$ is a stopping time relative to the process $X$, and $D\in \left\{-,+\right\}$ is a decision variable defined on the set of trajectories $X^T_0$ up to the decision time $T$.   The error probabilities  of the test are 
\begin{equation}
p_-= \mathbb{P}_+[D=-]\quad {\rm and} \quad p^\dagger_+ = \mathbb{P}_-[D=+],
\end{equation}
where $\mathbb{P}_+[D=-] =\mathbb{P}[D=-|H=+]$  and $\mathbb{P}_-[D=+] = \mathbb{P}[D=+|H=-]$.

Given certain maximally allowed error probabilities $\alpha_-$ and $\alpha_+$, we define the set 
\begin{equation}
\mathcal{C}_{\alpha_-,\alpha_+} =  \left\{(T,D): p_-\leq \alpha_+,  \  p^\dagger_+\leq \alpha_-,  \quad \langle T|H=+\rangle <\infty, \  \langle T|H=-\rangle <\infty\right\}
\end{equation} 
of all sequential hypothesis tests that meet the required constraints on the error probabilities and with finite expected decision times under both hypotheses.       We say that a sequential hypothesis test is optimal if it is an element of $\mathcal{C}_{\alpha_-,\alpha_+} $  and it minimises the mean decision times $\langle T| H=+\rangle$ and $\langle T| H=- \rangle$.

For general observation processes $X(t)$, the optimal sequential hypothesis test is not known.   However, in the asymptotic limit of small maximally allowed error probabilities $\alpha_-$ and $\alpha_+$  the optimal test is known and given  by the sequential probability ratio test    \cite{tartakovsky2014sequential}.     The sequential probability ratio test was first introduced by Wald for observation processes that are a sequence of  independent and identically distributed random variables \cite{wald1945sequential}, and      subsequently, Wald and Wolfowitz proved the optimality of the sequential probability ratio in the latter  setup   \cite{wald1948optimum}.   In a later work~\cite{lai1981asymptotic},  Lai proved the  asymptotic optimality of sequential probability ratio tests for general observation processes.     

Let
\begin{eqnarray}
\Lambda(t) = \log \frac{p^+(X^t_0)}{p^-(X^t_0) }, \label{eq:probRatioxa}
\end{eqnarray}  
be the log-likelihood ratio process, 
which should be understood as  the logarithm of the Radon-Nikodym derivative of  the probability measure $\mathbb{P}_+$   with respect to the probability measure $\mathbb{P}_-$, both constrained on the sub-$\sigma$-algebra generated by the trajectories $X^t_0$.    Loosely said,  $\Lambda(t) $ is the logarithm of the ratio of the probability densities $p^+(X^t_0)$ and $p^-(X^t_0) $ associated to the trajectories $X^t_0$, which clarifies the notation in Eq.~(\ref{eq:probRatioxa}).   The sequential probability ratio test is then the first-passage problem $T_{\Lambda}$ (see Eq.~(\ref{eq:def}) for the definition of first-passage times) with thresholds $-\ell_-$ and $\ell_+$   that determine the  error probabilities $p_-$ and $p^\dagger_+$.      When $\Lambda$ is a  continuous process, then 
\begin{equation}
\ell_- =  \log [(1-p^\dagger_+)/p_-] ,\quad \ell_+ = \log[(1-p_-)/p^\dagger_+]. 
\end{equation}

We formulate a lemma and a theorem about the  asymptotic properties of sequential hypothesis tests and the asymptotic optimality of  sequential probability ratio tests.    We first consider  Lemma 3.4.1 in \cite{tartakovsky2014sequential} that derives an asymptotic lower bound for   the moments of the decision times of sequential hypothesis tests.    

\begin{lemma}[Asymptotic lower bounds for the moments of decision times in sequential hypothesis tests] \label{Lemma1} 
Let  $\delta = (T,D)$  be a sequential hypothesis test in the set $\mathcal{C}_{\alpha_-,\alpha_+}$.  We assume that   $\Lambda(t)\in\mathbb{R}$ and $1/\Lambda(t)\in\mathbb{R}$ for all $t\geq 0$.  We assume that there exists a nonnegative increasing function $\psi(t)$ with $\psi(\infty) = \infty$  such that 
\begin{eqnarray}
\lim_{t\rightarrow \infty}\frac{\Lambda(t)}{\psi(t)} = \overline{\lambda}_+,   \quad (\mathbb{P}_+\textrm{-almost surely});\quad  \lim_{t\rightarrow \infty}\frac{\Lambda(t)}{\psi(t)} = -\overline{\lambda}_-,   \quad (\mathbb{P}_-\textrm{-almost surely}) \label{eq:40}
\end{eqnarray} 
with $\overline{\lambda}_-,\overline{\lambda}_+ \in(0,\infty)$.  Moreover, we assume that 
for all finite $\tau$ 
\begin{eqnarray}
\mathbb{P}_+\left[{\rm sup}_{t\in [0,\tau]}\Lambda(t)<\infty\right]=1,\quad \mathbb{P}_-\left[-{\rm inf}_{t\in [0,\tau]}\Lambda(t)<\infty\right]=1.
\end{eqnarray} 
Under these assumptions, it holds that for all $\epsilon>0$
 \begin{eqnarray}
 \lim_{\alpha_{\rm max}\rightarrow 0}{\rm inf}_{\delta\in\mathcal{C}\left(\alpha_-, \alpha_+\right)}\mathbb{P}_+\left[ T> (1-\epsilon)\Psi\left(|\log \alpha_-|/\overline{\lambda}_+\right) \right] = 1 \label{eq:43} \\ 
 \lim_{\alpha_{\rm max}\rightarrow 0}{\rm inf}_{\delta\in\mathcal{C}\left(\alpha_-, \alpha_+\right)} \mathbb{P}_-\left[T> (1-\epsilon)\Psi\left(|\log \alpha_+|/\overline{\lambda}_-\right) \right] = 1
 \end{eqnarray}
 where $\Psi(t)$ is the inverse of $\psi(t)$, i.e., $\Psi(\psi(t)) = t$.   Moreover, for all $n>0$    
  \begin{eqnarray}
   \lim_{\alpha_{\rm max}\rightarrow 0} {\rm inf}_{\delta\in\mathcal{C}\left(\alpha_-, \alpha_+\right)}\langle T^ n|H=+\rangle \geq \left(\Psi\left(|\log \alpha_-|/\overline{\lambda}_+\right)\right)^n(1+o_{\alpha_{\rm max}}(1))  \label{eq:ineqLemma1}\\ 
      \lim_{\alpha_{\rm max}\rightarrow 0} {\rm inf}_{\delta\in\mathcal{C}\left(\alpha_-, \alpha_+\right)}\langle T^n|H=-\rangle\geq  \left(\Psi\left(|\log \alpha_+|/\overline{\lambda}_-\right)\right)^n(1+o_{\alpha_{\rm max}}(1)).  \label{eq:ineqLemma2}
 \end{eqnarray}

  \end{lemma}   
  
  Second, we consider   Theorem 3.4.2 in \cite{tartakovsky2014sequential}  for  the asymptotic optimality of  the sequential probability ratio test.   In contrast to Lemma~\ref{Lemma1},  this theorem provides an equality for the moments of first-passage times and for this reason we will need to replace the almost sure convergence conditions Eqs.~(\ref{eq:40}) by the stronger  $r$-quick convergence condition.   Let 
  \begin{equation}
  L_{\epsilon}(Y(t)) = {\rm sup}\left\{t>0 : |Y(t)|>\epsilon\right\},
  \end{equation}
be the  last entry time of a real-valued  stochastic process $Y(t)\in\mathbb{R}$ into an interval $[-\epsilon,\epsilon]$, with ${\rm sup}\left\{\phi\right\}=0$.    We say that $Y(t)$  converges $r$-quickly to $0$ in  $\mathbb{P}_+$ if $ \langle L^r_{\epsilon} |H=+\rangle <\infty$ for every $\epsilon>0$.

\begin{theorem}[Asymptotic optimality of sequential probability ratio tests] \label{Theorem2}  
We assume that 
\begin{eqnarray}
\lim_{t\rightarrow \infty}\frac{\Lambda(t)}{\psi(t)} = \overline{\lambda}_+,   \quad (r{\rm -quickly}  \  {\rm in}  \  \mathbb{P}_+);\quad  \lim_{t\rightarrow \infty}\frac{\Lambda(t)}{\psi(t)} = -\overline{\lambda}_-,   \quad (r{\rm-quickly}  \  {\rm in}  \  \mathbb{P}_-)\label{eq:42}
\end{eqnarray} 
where $r$ is a natural number.
 It holds then that 
\begin{itemize}
  \item    for any finite threshold values $\ell_-$ and $\ell_+$,  
  \begin{equation}
  \langle T^r_{\Lambda }|H=\pm \rangle <\infty;
  \end{equation}
  \item  for all $m\in (0,r]$,  
  \begin{equation}
    \langle T^m_{\Lambda} | H=\pm \rangle = \left(\Psi\left(\ell_{\pm}/\overline{\lambda}_{\pm}\right)\right)^m \left(1+o_{\ell_{\rm min}}(1)\right);
  \end{equation}
  \item if $\ell_- = |\log p_-| (1+o_{\ell_{\rm min}}(1))$ and $\ell_+ = |\log p^\dagger_+|(1+o_{\ell_{\rm min}}(1))$, then  for all $m\in (0,r]$
  \begin{equation}
    \langle T^m_{\Lambda} | H=+ \rangle = \left(\Psi\left(|\log p^\dagger_{+}|/\overline{\lambda}_{+}\right)\right)^m \left(1+o_{\ell_{\rm min}}(1)\right) \label{eq:THPA}
  \end{equation} 
  and 
   \begin{equation}
    \langle T^m_{\Lambda} | H=- \rangle = \left(\Psi\left(|\log p_{-}|/\overline{\lambda}_{-}\right)\right)^m \left(1+o_{\ell_{\rm min}}(1)\right). 
  \end{equation} 
\end{itemize}
\end{theorem}

  \subsection{Derivation of the  first-passage bound Eq.~(\ref{eq:2})   based on  Lemma~\ref{Lemma1} } 
  
   We  use Lemma~\ref{Lemma1} to derive Eq.~(\ref{eq:2}).     However, as will become  soon evident, Lemma~\ref{Lemma1} is not equivalent to  Eq.~(\ref{eq:2}), as to derive  Eq.~(\ref{eq:2})  we  also need to relate $p^\dagger_+$  to $p_-$. 
   
  Let $\mathbb{P}$ denote the probability measure of events in the forward dynamics and  let $\mathbb{P}\circ\Theta$ be the probability measure of events in the time-reversed dynamics.   
   Setting $\mathbb{P}_+ = \mathbb{P}$,  $\mathbb{P}_- = \mathbb{P}\circ \Theta$, and $\psi(t)=t$, we obtain according to definition (\ref{eq:probRatio}) that  $\Lambda(t) = S(t)$ and $\overline{\lambda}_+ = \dot{s}$.     Since $J$ is a stochastic current it  changes sign under time-reversal and therefore the pair $(T_J,D_J)$, with $T_J$ as defined in Eq.~(\ref{eq:def}) and $D_J = {\rm sign}(J(T_J))$, is a sequential hypothesis test corresponding to the two probability measures $\mathbb{P}$ and $\mathbb{P}\circ\Theta$ \cite{roldan2015decision}.  Replacing in Eq.~(\ref{eq:ineqLemma1}) the $\alpha_-$ by $p^\dagger_+$ and the $o_{\alpha_{\rm max}}(1)$  by  $o_{\ell_{\rm min}}(1)$, we obtain \cite{roldan2015decision}
 \begin{eqnarray}
\langle T^n_J\rangle  \geq  \left(\frac{|    \log p^\dagger_+ |}{ \dot{s} }  \right)^n(1+o_{\ell_{\rm min}}(1)). \label{eq:2xxxx}
\end{eqnarray}  
In Appendix~\ref{derivP}, we derive using heuristic   mathematical arguments the equality  
\begin{equation}
\frac{\ell_-}{\ell_+}=\frac{|\log p_-|}{|\log p^\dagger_+ |}(1+o_{\ell_{\rm min}}(1)) \label{eq:pp}
\end{equation}  
 for currents  $J$ in stationary Markov jump processes $X$ taking values in a finite set $\mathcal{X}$.
Multiplying the right-hand side of Eq.~(\ref{eq:2xxxx})  with 
\begin{equation}
1 = \left(\frac{\ell_+}{\ell_-} \frac{|\log p_-|}{|\log p^\dagger_+|}\right)^n, \label{eq:conVer}
\end{equation} 
we obtain Eq.~(\ref{eq:2}), which concludes the derivation.    

The partial result Eq.~(\ref{eq:2xxxx}) is interesting in its own right as it is an extremely, general relation that has been derived with full mathematical rigour.  Indeed, Lemma~\ref{Lemma1} holds for  processes $X$ that are reversible,  in the sense  that the stochastic entropy production $S(t)$ is well defined, and obey a  weak stationary condition, in the sense that $S(t)/t$ converges almost surely to a deterministic limit.    Remarkably, we do not require a large deviation principle for $J$, and we do not even require a large deviation principle for  $S$.   
  
To obtain Eq.~(\ref{eq:2})  from Eq.~(\ref{eq:2xxxx}), we have used Eq.~(\ref{eq:pp}).   Note that (\ref{eq:pp}) has not been derived with the same mathematical rigour as (\ref{eq:2xxxx}), and  it is not clear whether  Eq.~(\ref{eq:pp}) is valid beyond the setup of stationary Markov jump  processes.  However,  Eq.~(\ref{eq:2}) can be interpreted as a tradeoff relation between dissipation, speed, and uncertainty, whereas  the interpretation of  Eq.~(\ref{eq:2xxxx}) as a tradeoff relation is less clear,    as  $p^\dagger_+$ is the splitting probability in  the time-reversed process.    
  \subsection{Derivation of the  asymptotic equality Eq.~(\ref{eq:eq})  based on  Theorem~\ref{Theorem2} }
   We set again $\mathbb{P}_+ = \mathbb{P}$,  $\mathbb{P}_- = \mathbb{P}\circ \Theta$, and $\psi(t)=t$,  obtaining
 $\overline{\lambda}_+ = \dot{s}$ and $\Psi(t) = t$.  Therefore,  Eq.~(\ref{eq:THPA}) reads 
 \begin{eqnarray}
\langle T^n_S\rangle  =  \left(\frac{|    \log p^\dagger_+ |}{ \dot{s} }  \right)^n(1+o_{\ell_{\rm min}}(1)). \label{eq:eqxxxx}
\end{eqnarray}   
In Sec.~\ref{sec:equal} we have shown that    
\begin{equation}
\ell_- = |\log p_- |(1+o_{\ell_{\rm min}}(1)), \label{eq:Spp}
\end{equation}
which follows readily from the martingale property of $e^{-S(t)}$.    Analogously, one can show that  \cite{neri2017statistics}
\begin{equation}
\ell_+=|\log p^\dagger_+ |(1+o_{\ell_{\rm min}}(1)).  \label{eq:Spp2}
\end{equation}
Multiplying the right-hand side of Eq.~(\ref{eq:eqxxxx})  with 
\begin{equation}
1 = \left(\frac{\ell_+}{\ell_-} \frac{|\log p_-|}{|\log p^\dagger_+|}\right)^n, \label{eq:ratioxx}
\end{equation} 
we obtain Eq.~(\ref{eq:eq}), which completes the derivation.     Note that because of the martingale property of $e^{-S}$ the Eq.~(\ref{eq:eq}) can be derived with full mathematical rigour in a very general setup.

          \section{Connections between Eq.~(\ref{eq:2}) and  other thermodynamic tradeoff relations }
          \label{sec:prev}
          We point out connections between Eq.~(\ref{eq:2}) and  thermodynamic tradeoff inequalities that appeared before in the literature.

\subsection{Decision making in the arrow of time} 
Equation (9) in Ref.~\cite{roldan2015decision} implies  Eq.~(\ref{eq:2xxxx}).   Indeed, Equation (9) in Ref.~\cite{roldan2015decision} implies that in the limit $\ell_{\rm min}\rightarrow \infty$,
\begin{equation}
\dot{s} \geq   \frac{ D\left[\rightarrow || \leftarrow \right]}{\langle T_J\rangle} = \frac{p_+}{\langle T_J\rangle} \log \frac{p_+}{p^\dagger_+} +\frac{ p_-}{\langle T_J\rangle} \log \frac{p_-}{p^\dagger_-} = \frac{|\log p^\dagger_+|}{\langle T_J\rangle} (1+o_{\ell_{\rm min}}(1)),
\end{equation}
which  is equivalent to Eq.~(\ref{eq:2xxxx}).

 The main distinction between the Eq.~(9) in  Ref.~\cite{roldan2015decision} and  Eq.~(\ref{eq:2}) in the present paper is that Eq.~(\ref{eq:2})   involves  $p_-$, while  Eq.~(9) of  \cite{roldan2015decision} involves  $p^\dagger_+$.   This distinction is relevant as $p^\dagger_+$ involves fluctuations of the process in a time-reversed dynamics that is not always accessible.
\subsection{Dissipation-time uncertainty relation}
Eq.~(\ref{eq:2}) is  related to the  so-called dissipation-time uncertainty relation that states 
\begin{equation}
\langle T_J \rangle  \geq \frac{1}{\dot{s}}  \label{eq:timeDis}
\end{equation} 
 in the limit $|\log p^\dagger_+| \gg 1$ \cite{PhysRevLett.125.120604, PhysRevLett.128.050603}.      
 
The dissipation-time uncertainty relation  is a loose bound when compared to the bounds Eqs.~(\ref{eq:2}) and Eq.~(\ref{eq:2xxxx}).     Indeed, comparing Eq.~(\ref{eq:timeDis}) with (\ref{eq:2}), or better  Eq.~(\ref{eq:timeDis}) with (\ref{eq:2xxxx}), we conclude that 
 \begin{equation}
 \langle T_J \rangle  \geq \frac{c}{\dot{s}}(1+o_{\ell_{\rm min}}(1))
 \end{equation}
 holds for any prefactor $c\geq0$.   This is because the prefactor in Eq.~(\ref{eq:2xxxx})  is $c=|\log p^\dagger_+|$ and thus  diverges when $p^\dagger_+$ is small.      
\subsection{Thermodynamic uncertainty relations}
The  bound (\ref{eq:2}) follows from the   bound Eq.~(\ref{eq:JBoundx}) on the large deviation function of a stochastic current.      Since also the   thermodynamic uncertainty relations have been derived using the bound (\ref{eq:JBoundx}), see Refs.~\cite{barato2015thermodynamic,gingrich2016dissipation, pietzonka2016universalx}, we discuss here how the  bound Eq.~(\ref{eq:2})  is related to thermodynamic uncertainty relations.   

The thermodynamic uncertainty relation  bounds from below the   Fano factor of  stochastic currents, i.e., \cite{barato2015thermodynamic,gingrich2016dissipation}
\begin{equation}
 \frac{\sigma^2_{J}}{2 \overline{j}^2}   \geq   \frac{1}{\dot{s}}, \label{eq:1x}
\end{equation} 
where $\overline{j}$ is the current rate and 
\begin{equation}
\sigma^2_{J} = \lim_{t\rightarrow \infty} \frac{1}{t} \left(\langle J^2(t) \rangle-  \langle J(t)\rangle^2\right).
\end{equation} 
A first-passage time thermodynamic uncertainty relation was derived in
Ref.~\cite{gringich2017bis}, viz., 
\begin{equation}
\frac{\langle T^2_J\rangle - \langle T_J\rangle^2}{2\langle T_J\rangle}   \geq   \frac{1}{ \dot{s}} (1+o_{\ell_{\rm min}}(1)).\label{eq:T2}
\end{equation}    

The bounds Eqs.~(\ref{eq:2}), (\ref{eq:1x}) and (\ref{eq:T2})  all express  a nonequilibrium tradeoff between dissipation, speed, and uncertainty.  The differences between these  bounds are in how they quantify speed and uncertainty.  The thermodynamic uncertainty relation Eq.~(\ref{eq:1x})  quantifies speed with $\overline{j}$ and uncertainty with $\sigma^2_{J}$, the first-passage time uncertainty relation  Eq.~(\ref{eq:T2}) quantifies speed with $\langle T_J\rangle$ and uncertainty with $\langle T^2_J\rangle - \langle T_J\rangle^2$, and the bound Eq.~(\ref{eq:2}) quantifies speed with $\langle T_J\rangle$ and uncertainty with $p_-$.

An important distinction between the thermodynamic uncertainty relations, Eqs.~(\ref{eq:1x}) and Eq.~(\ref{eq:T2}), and the bound Eq.~(\ref{eq:2})  on the moments of first-passage times, is that the latter is tight  when $J=S$ while the former is loose.   Indeed, if $J(t)= S(t)$, then Eq.~(\ref{eq:2}) becomes the equality  Eq.~(\ref{eq:eq}), whereas the Eqs.~(\ref{eq:1x}) and Eq.~(\ref{eq:T2}) are in general not equalities, even not when  $J(t)= S(t)$ \cite{pigolotti2017generic, busiello2019hyperaccurate}.    How is this possible, given that the relations (\ref{eq:2}), (\ref{eq:1x}), and 
(\ref{eq:T2})  are all derived from the same bound, viz.,  Eq.~(\ref{eq:JBoundx}) on the large deviation function?    We can understand this as follows.   Eq.~(\ref{eq:2}) is obtained from evaluating the bound    (\ref{eq:JBoundx}) at the value $z = -\overline{j}$, while  Eqs.~(\ref{eq:1x}) and (\ref{eq:T2})  rely on the properties of the large deviation function  in the vicinity of the point $z= \overline{j}$, in particular, the derivatives of the large deviation function at this point.   As   observed in Ref.~\cite{pietzonka2016universalx},  the large deviation function bound  Eq.~(\ref{eq:JBoundx}) is tight when $J=S$ and $z=-\dot{s}$, while this is not the case for the  curvature of the bound at $z=\dot{s}$, as the large deviation function of $S$ is in general not a parabola.   

The tightness of the bound   (\ref{eq:JBoundx})  for $J=S$ at $z=-\dot{s}$ can also be understood from the 
Gallavotti-Cohen fluctuation relation \cite{lebowitz1999gallavotti}
\begin{equation}
 \mathcal{J}(z) - \mathcal{J}(-z) = -z \label{eq:Jz}.  
\end{equation}  
For $z = -\dot{s}$,  the Gallavotti-Cohen relation implies that $\mathcal{J}(-\dot{s})=\dot{s}$ as $\mathcal{J}(\dot{s}) = 0$.    One verifies readily that the right hand side of Eq.~(\ref{eq:JBoundx}) is equal to $\dot{s}$ when $z = -\dot{s}$ and $\overline{j} = \dot{s}$, and hence the bound  Eq.~(\ref{eq:JBoundx}) is tight when $J=S$ and $z=-\dot{s}$.       
The Gallavotti-Cohen fluctuation  relation  Eq.~(\ref{eq:Jz}) also applies for currents $J$ that are proportional to the entropy production \cite{barato2012gallavotti, barato2012symmetry}, and hence the bound Eq.~(\ref{eq:2}) is also tight for those currents.    Importantly, the fluctuation relation  does not apply generically  for  currents in multicyclic networks that are not proportional to $S$ \cite{gaspard2013multivariate, barato2012gallavotti, barato2012symmetry, polettini2019effective}, and hence the inequality (\ref{eq:2}) is not tight for generic currents.

\section{Recovering the Van't Hoff-Arrhenius law in the near equilibrium limit}  \label{sec:Arrh}  
We show that near equilibrium Eq.~(\ref{eq:2})  implies that $1/\langle T_J\rangle$ is smaller than or equal to  the Van't Hoff-Arrhenius law Eq.~(\ref{eq:Arrh}).   To this aim, we consider a nonequilibrium version of Kramers' model \cite{kramers1940brownian, hanggi1990reaction}.   Details of the calculations can be found in the Appendices~\ref{eq:Period} and \ref{eq:Period2}.

\begin{figure}[t!]
\centering
\includegraphics[width=0.5\textwidth]{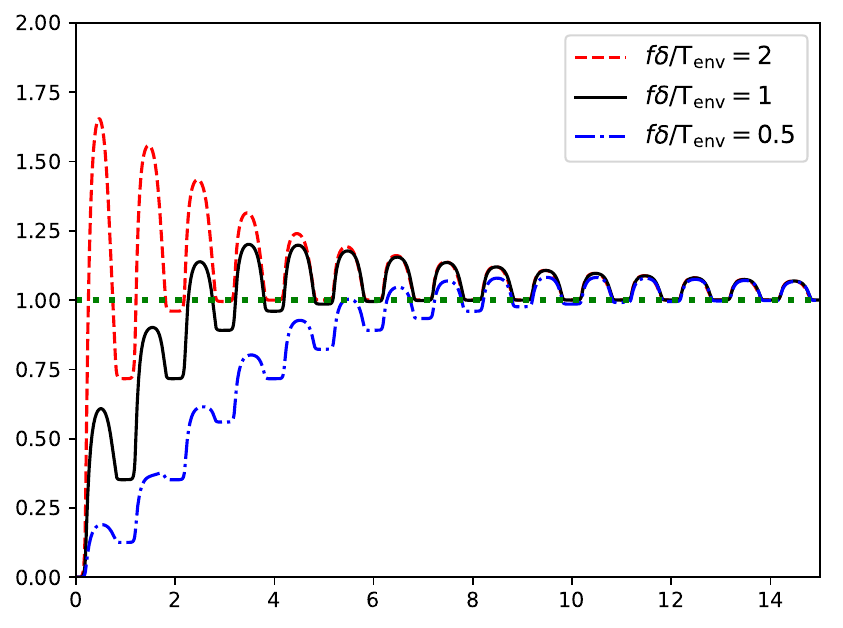}
 \put(-260,80){\Large$\frac{\langle T_X\rangle \dot{s}}{|\log p_-|}$}
  \put(-115,-10){\Large$\ell/\delta$} 
\caption{ {\it Asymptotic lower bound on the mean first-passage time.}  The ratio  $\langle T_X\rangle \dot{s}/|\log p_-|$ is plotted as a function of $\ell/\delta$, where $T_X$ is the first-passage time Eq.~(\ref{eq:def2}) of the nonequilibrium   Kramers process $X$ described by Eq.~(\ref{eq:nonequilb}) with triangular potential  $u$ given by Eq.~(\ref{eq:modelU}).  Curves shown are for the parameters  $\delta=5$, $x^\ast=1$, $u_0 = 10$, $\mathsf{T}_{\rm env}=1$, and $\gamma=1$, and the values of $f$ are given in the figure legend.   } \label{fig:Meantime}
\end{figure}

\begin{figure}[t!]
\centering
\includegraphics[width=0.335\textwidth]{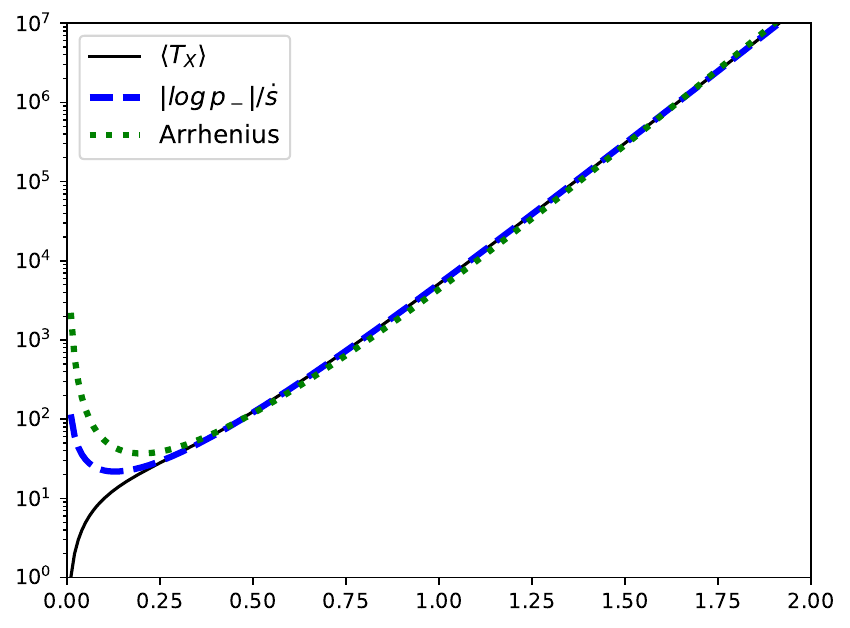}
\hspace{-3mm}
\includegraphics[width=0.335\textwidth]{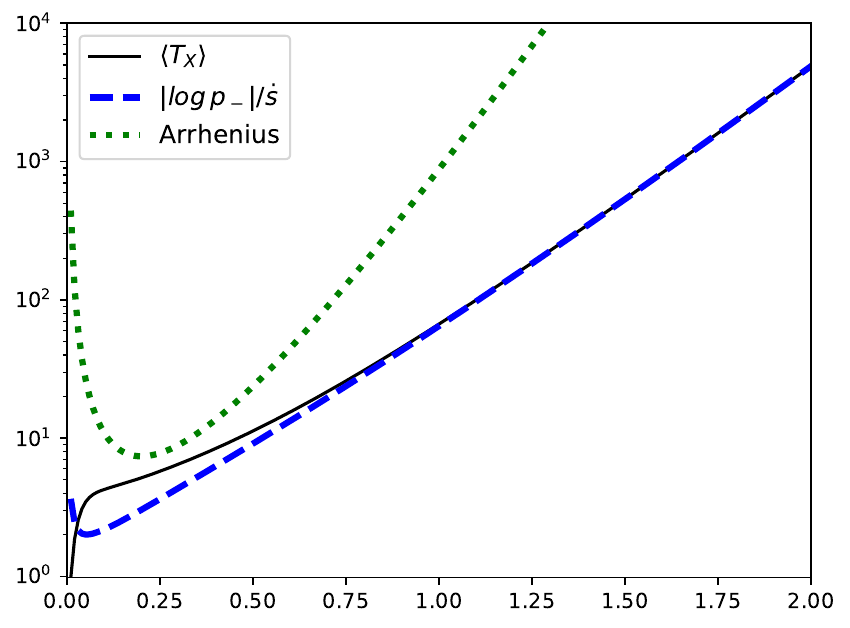}
\hspace{-3mm}
\includegraphics[width=0.335\textwidth]{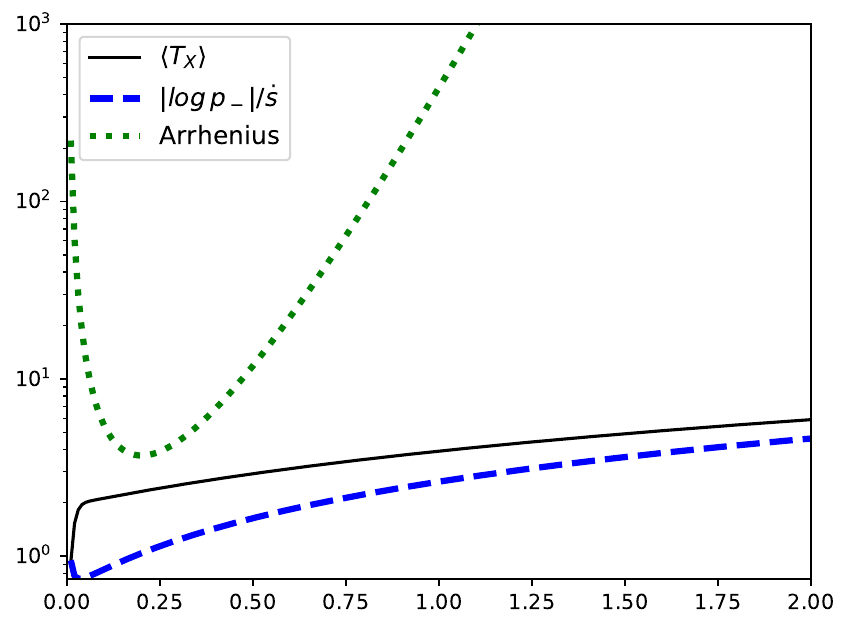}
 \put(-85,-10){$1/\mathsf{T}_{\rm env}$} 
 \put(-230,-10){$1/\mathsf{T}_{\rm env}$} 
\put(-370,-10){$1/\mathsf{T}_{\rm env}$} 
\caption{ {\it Extension of  the Van't Hoff-Arrhenius law to nonequilibrium stationary states.}  The mean first-passage time $\langle T_X\rangle$ (solid black line) of the  reaction coordinate $X$, described by Eq.~(\ref{eq:nonequilb}) with triangular potential  $u$ given by Eq.~(\ref{eq:modelU}), is plotted as a function of the inverse temperature $1/\mathsf{T}_{\rm env}$, and $\langle T_X\rangle$ is also compared with its asymptotic value $|\log p_-|/\dot{s}$ for large thresholds $\ell$ (blue dashed line) and with the Van't Hoff-Arrhenius law Eq.~(\ref{eq:Arrhx}) (green dotted line).    The model parameters are    $\delta=5$, $x^\ast=1$, $u_0 = 10$, $\mathsf{T}_{\rm env}=1$ and $\gamma=2$ and the values of $f$ are  $f=1$, $f=5$ and $f=10$ from left to right, respectively.   The threshold for the first-passage time $T_X$, which is defined in Eq.~(\ref{eq:def2}), is $\ell=10$. } \label{fig:Meantime2}
\end{figure}

We consider  a reaction coordinate $X\in\mathbb{R}$ that is described by the overdamped Langevin equation 
\begin{eqnarray}
{\rm d}X(t) = \frac{f-\partial_x u(X(t))}{\gamma} \:{\rm d}t +  \sqrt{2  \mathsf{T}_{\rm env}/\gamma}\:{\rm d}W(t), \label{eq:nonequilb}
 \end{eqnarray} 
 where $u(x)$ is a periodic potential with period $\delta$, i.e., $u(x+\delta) = u(x) = u(x-\delta)$,  $f$ is a nonconservative force,  $\gamma$ is a friction coefficient,  $W(t)$ is a  standard Wiener process that models the thermal noise, and $\mathsf{T}_{\rm env}$ is the temperature of the environment.  We assume that  at time $t=0$, $X(0) = 0$ and  $W(0) = 0$.       Note that this example goes beyond the paradigm of a Markov jump process, but the theory will still apply.

 The  variable $X$  models, e.g.,  a reaction coordinate that tracks the progress of a chemical reaction.  In this scenario, $E_{\rm b} = {\rm max}_{x}u(x)-\min_{x}u(x)$ is  the Gibbs free energy barrier that separates two  chemical states and the ratio $ \lfloor X/\delta \rfloor$ is  the number of cycles of the reaction that have  been completed;  $\lfloor a\rfloor$ denotes the largest integer smaller than or equal to $a$. 
 
 Figure~\ref{fig1Mx} presents two trajectories generated by Eq.~(\ref{eq:nonequilb}) for the special case where $u(x)$ is the triangular potential
 \begin{equation}
u(x) = \left\{\begin{array}{ccc} u_0 \frac{x}{x^\ast}  &{\rm if}&   x\in [0,x^\ast),\\  u_0 \frac{\delta-x}{\delta-x^\ast}  &{\rm if}& x\in  [x^\ast,\delta).\end{array}\right. \label{eq:modelU}
\end{equation}   
    From Fig.~\ref{fig1Mx} we observe that the dynamics consists of a sequence of jumps between metastable states that are centred at the positions $n\delta$ with $n\in\mathbb{Z}$.     In  the equilibrium case with $f=0$ the jumps are activated by thermal fluctuations and the Van't Hoff-Arrhenius law Eq.~(\ref{eq:Arrh}) applies.    On the other hand, when $f>0$, then jumps in one direction over the energy barrier $E_{\rm b}$ are facilitated by the external driving $f$, while in the reverse direction jumps are less likely.  In this case,   although the Van't Hoff-Arrhenius law Eq.~(\ref{eq:Arrh})  does not apply, the Eqs.~(\ref{eq:2}) and (\ref{eq:eq}) apply and can thus be considered nonequilibrium versions of the Van't Hoff-Arrhenius law.

For values  $f\delta/E_{\rm b}> 0$  the chemical reaction settles into a  nonequilibrium stationary state with an entropy production rate (see Appendix~\ref{appC2}) 
\begin{equation}
\dot{s} = \frac{f\delta}{\mathsf{T}_{\rm env}}j_{\rm ss}, \label{eq:sdot}
\end{equation}
where $j_{\rm ss}$ is the stationary current (see Appendix~\ref{appC1})
 \begin{equation}
j_{\rm ss} = \frac{\mathsf{T}_{\rm env}}{\gamma}\frac{1- e^{\frac{-f\delta}{\mathsf{T}_{\rm env}}}}{\int^\delta_0 {\rm d}y\: w(y) \left(\int^{y+\delta}_{y}{\rm d}x' \frac{1}{w(x')}\right)}, \label{eq:jssx}
 \end{equation} 
  and where $w(x)   = \exp(-(u(x)-fx)/\mathsf{T}_{\rm env})$.

Consider the first time
\begin{equation}
T_X = {\rm inf}\left\{t>0: X(t)\notin (-\ell, \ell )\right\} \label{eq:def2}
\end{equation} 
when  the reaction has  completed a net number  $\lfloor \ell/\delta \rfloor$ of cycles in either the forward or backward direction.   Since, (see Appendix~\ref{appC2})
\begin{equation}
S(t) =  \frac{f X(t)}{\mathsf{T}_{\rm env}} + o(t) \label{eq:SBrown}
\end{equation}
the equality (\ref{eq:eq}) applies to  $T_X$.    In Appendices~\ref{App:C3} and \ref{App:C4}, we derive  explicit analytical expressions for  the splitting probability $p_-$ and the mean first-passage time $\langle T_X\rangle$, respectively, which we omit here as the expressions  are involved.    However, as shown in Appendix~\ref{App:C5}, in   the limit of   large $\ell$ we obtain the formula   
\begin{equation}
 \frac{|\log p_-|}{  \langle T_X\rangle } = \dot{s} + O\left(\frac{1}{\ell}\right),   \label{eq:asymptT}
\end{equation}
in correspondence with Eq.~(\ref{eq:eq}), where $O$ denotes the big-O notation. 
The big-O notation $O(f(\ell))$ denotes an arbitrary  function $g(\ell)$ for which it holds that   there exists a constant $c$  such that $g(\ell)<cf(\ell)$ for $\ell$ large enough.   Hence, in this case,    the correction term in  Eq.~(\ref{eq:eq}) is of order $1/\ell$.    

In Fig.~\ref{fig:Meantime} we  plot $ \dot{s} \langle T_X\rangle/|\log p_-|$  as a function of $\ell/\delta$.   The figure demonstrates the convergence of   $\dot{s}  \langle T_X\rangle/|\log p_-|  $ to its universal limit   for different values of the nonequilibrium driving $f\delta/\mathsf{T}_{\rm env}$.    Observe the oscillations of $\dot{s} \langle T_X\rangle/|\log p_-| $.   These oscillations appear  because for the selected parameters  it holds that $E_{\rm b} \gg  \mathsf{T}_{\rm env}$, and therefore the process consists of discrete-like hops over the energy barrier $E_{\rm b}$ that  represent the subsequent completion cycles of the chemical reaction.

In the limits $\mathsf{T}_{\rm env}\rightarrow 0$  and $f\delta/\mathsf{T}_{\rm env}\rightarrow 0$, the Eq.~(\ref{eq:eq})  leads to a Van't Hoff-Arrhenius law for $1/\langle T_X\rangle$.   Indeed, as shown in Appendix~\ref{app:C6}, taking the limits $\mathsf{T}_{\rm env}\rightarrow 0$  and $f\delta/\mathsf{T}_{\rm env}\rightarrow 0$ in the expression of the stationary current   Eq.~(\ref{eq:jssx}),  we obtain
\begin{equation}
 j_{\rm ss} =\kappa \frac{f\delta }{\gamma }  e^{\frac{-E_{\rm b}}{\mathsf{T}_{\rm env}}}, \label{eq:jsso}
  \end{equation} 
  where the prefactor  is given by 
  \begin{equation}
  \kappa = \frac{\sqrt{-u''_{\rm min} u''_{\rm max}}}{2\pi\mathsf{T}_{\rm env}}
  \end{equation}
   if the second derivatives $u''_{\rm min}$ and $u''_{\rm max}$ evaluated at the minimum and maximum of $u(x)$, respectively, exist.   In the special case of the triangular potential, given by Eq.~(\ref{eq:modelU}), the second derivatives $u''_{\rm min}$ and $u''_{\rm max}$ do not exist, and therefore
     \begin{equation}
\frac{1}{\kappa} =\mathsf{T}^2_{\rm env} \Bigg| \left(\frac{1}{u^+_{\rm max}}-\frac{1}{u^-_{\rm max}}\right) \left(\frac{1}{u^+_{\rm min}}-\frac{1}{u^-_{\rm min}}\right) \Bigg| 
  \end{equation}
  where  $u^{+}_{\rm max}$ and $u^{-}_{\rm max}$ denote the left and right derivatives evaluated at the maximum of $u(x)$. 
In addition, as shown in Appendix~\ref{app:C6},  in the limit of $\mathsf{T}_{\rm env}\rightarrow 0$  and $f\delta/\mathsf{T}_{\rm env}\rightarrow 0$ the logarithm of the splitting probability is inversely proportional to the temperature, viz.,
\begin{equation}
\log p_-= -\frac{f\ell}{\mathsf{T}_{\rm env}} + O_\ell(1). \label{eq:pMf}
\end{equation} 
Combining Eqs.~(\ref{eq:eq}),  (\ref{eq:sdot}), (\ref{eq:jsso}), and (\ref{eq:pMf})  we obtain the Van't Hoff-Arrhenius  law 
\begin{equation}
 \langle T_X\rangle =  \frac{\ell}{\delta} \frac{\gamma }{f\delta}  \frac{1}{\kappa} e^{\frac{E_{\rm b}}{\mathsf{T}_{\rm env}}}. \label{eq:Arrhx}
\end{equation}

  In Fig.~\ref{fig:Meantime2} we compare $\langle T_X\rangle$   with its asymptotic value $|\log p_-|/\dot{s}$, given by Eq.~(\ref{eq:eq}), and with the Van't Hoff-Arrhenius law, given by Eq.~(\ref{eq:Arrhx}), for three values of the driving force $f$.      We make a few interesting observations: (i) the Van't Hoff-Arrhenius law approximates well $\langle T_X\rangle$ up to moderately large values of  $f\delta/\mathsf{T}_{\rm env}<5$; (ii)  for $f\delta/\mathsf{T}_{\rm env}>25$, $\langle T_X\rangle$ is significantly smaller than what is predicted by the  Van't Hoff-Arrhenius law, implying that the nonequilibrium driving speeds up the process.   Nevertheless,  $\langle T_X\rangle$ is larger than $|\log p_-|/\dot{s}$, which is a consequence of the  tradeoff between speed, uncertainty, and dissipation as expressed by Eq.~(\ref{eq:2}); (iii) the asymptotic expression  $|\log p_-|/\dot{s}$ given by Eq.~(\ref{eq:eq})   approximates  $\langle T_X\rangle$ already well for relatively small values of the threshold, viz., $\ell/\delta=2$.  
  
  Taken together, we conclude that the Eqs.~(\ref{eq:2}) and (\ref{eq:eq}) recover the Van't Hoff-Arrhenius law near equilibrium   because   $\dot{s}\sim \exp(-E_{\rm b}/\mathsf{T}_{\rm env})$  in the limit of small temperatures $\mathsf{T}_{\rm env}\approx 0$ and small driving force $f\delta/\mathsf{T}_{\rm env}\approx 0$.   On the other hand, one  can significantly increase the reaction rate $1/\langle T_X\rangle$ by driving a system out of equilibrium, even though the reaction rates are still bounded from above  by the inequality   Eq.~(\ref{eq:2}) that expresses a tradeoff between speed, uncertainty, and dissipation.

\section{Illustration of the tightness of the first-passage time bounds with a biased random walker}\label{sec:inf} 

As stated before, the bound Eq.~(\ref{eq:2}) is tight for $J=S$, whereas the  thermodynamic uncertainty relation Eq.~(\ref{eq:T2}) is loose when $J=S$.   In this section we compute the moments $\langle T^n_J\rangle$ on an example of a nonequilibrium process to better understand the origin of the tightness of the bound Eq.~(\ref{eq:2}).

We consider a hopping process   $X\in\mathbb{Z}$ described by 
\begin{equation}
{\rm d}X(t) =  {\rm d}N_+(t) -  {\rm d}N_-(t), \label{eq:randomwalkX}
\end{equation} 
where $N_+$ and $N_-$ are two counting processes with rates $k_+$ and $k_-$, respectively.    The bias of the process is defined by  the ratio 
\begin{equation}
b := \frac{k_-}{k_+} = \exp\left(-\frac{a}{\mathsf{T}_{\rm env}}\right)
\end{equation}  
where  $a$ is the thermodynamic affinity and $\mathsf{T}_{\rm env}$ the temperature of the environment.    We  assume, without loss of generality, that $k_-<k_+$ so that $b<1$.

  The coordinate $X$  may represent the number of  times a chemical reaction has been completed or the position of a molecular motor on a biofilament.   In the former,  $a=\Delta \mu$ is the difference between the sum of the chemical potentials of the   reactants and the products of the chemical reaction, and in the latter      $a=f \delta$  is the work performed by the system on the motor when it moves forwards.     Hence,   the stochastic entropy production $S$ obeys 
 \begin{equation}
 {\rm d}S(t) = \frac{a}{\mathsf{T}_{\rm env}} {\rm d}X(t) 
\end{equation}
and 
\begin{equation}
\dot{s} =  \Big\langle  \frac{{\rm d}S}{{\rm d}t}\Big\rangle = \frac{a}{\mathsf{T}_{\rm env}}(k_+-k_-) \label{eq:sigmaxM}
\end{equation}
is the entropy production rate. 
 
We consider the first passage time 
\begin{equation}
T_X = {\rm inf}\left\{t>0: X(t)-X(0)\notin (-\ell_-,\ell_+)\right\} ,\label{eq:def3}
\end{equation} 
which is also the first-passage time    $T_S$ for the stochastic entropy production with thresholds $s_- = a\ell_-/\mathsf{T}_{\rm env}$ and $s_+ = a\ell_+/\mathsf{T}_{\rm env}$.    
 
 The splitting probabilities $p_-$ and $p_+$  are given by  (see Appendix~\ref{App:SplitE}) 
 \begin{eqnarray}
p_+ = \frac{1 -b^{\lceil \ell_-\rceil} }{ 1-  b^{\lceil \ell_-\rceil+\lceil \ell_+\rceil}} \  \ {\rm and }\  \  p_- = b^{\lceil \ell_-\rceil} \frac{1 -b^{\lceil \ell_+\rceil} }{ 1-  b^{\lceil \ell_-\rceil+\lceil \ell_+\rceil}}, \label{eq:pMP}
\end{eqnarray}  
where $\lceil \ell_-\rceil$ and $\lceil \ell_+\rceil$ denote the smallest integers that are greater than or equal to  $\ell_-$ and $\ell_+$, respectively.
The  generating function 
 \begin{equation}
g(y) = \langle e^{-y T_X(k_-+k_+)}  \rangle \label{eq:gJump}
 \end{equation} 
 is  for all $y>0$ given by (see  Appendix~\ref{App:Gen})
 \begin{eqnarray}
g(y) =  \left(\frac{2}{\zeta _+(y)}\right)^{\lceil \ell_+\rceil}  \frac{1-\left(\frac{\zeta _-(y)}{ \zeta _+(y) }\right)^{\lceil \ell_-\rceil}}{1-\left(\frac{\zeta_-(y)}{ \zeta _+(y)  } \right)^{^{\lceil \ell_-\rceil+\lceil \ell_+\rceil}}}   
\nonumber\\
  +  \left(\frac{\zeta _-(y)}{2}\right)^{\lceil \ell_-\rceil}   \frac{1- \left( \frac{\zeta_-(y)}{\zeta _+(y)}\right)^{\lceil \ell_+\rceil}}{1-  \left(\frac{\zeta _-(y)}{\zeta _+(y)} \right)^{\lceil \ell_-\rceil+\lceil \ell_+\rceil} } ,\label{eq:genFinalJump}
\end{eqnarray} 
 where 
\begin{equation}
\zeta_{\pm}(y) = \beta(y) \pm \sqrt{-4b +\beta^2(y)}\label{eq:zetapm}
\end{equation}  
and 
\begin{equation}
\beta(y) =(1+y) (1+b). \label{eq:betay}
\end{equation}

The moments of $T_X$ follow from 
\begin{equation} 
\langle T^n_X\rangle =  \left(\frac{-1}{k_-+k_+}\right)^n\left.\frac{{\rm d}^n}{({\rm d} y)^n} g(y)\right|_{y=0}, \label{eq:TMom}
\end{equation}    
where $n\in\mathbb{N}$.

\begin{figure}[t!]
\centering
\includegraphics[width=0.5\textwidth]{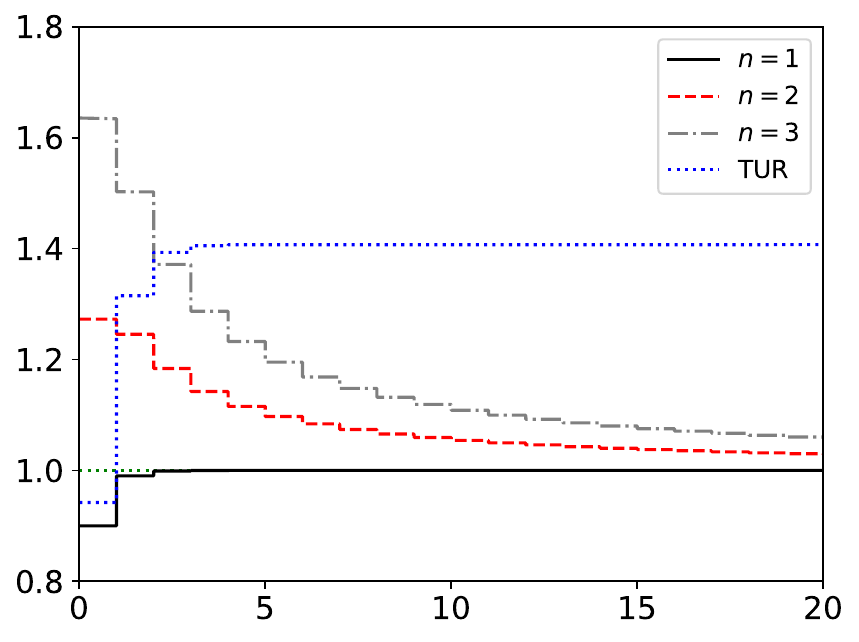}
 \put(-270,80){\Large$\dot{s}\frac{\langle T^n_X \rangle^{1/n}}{|\log p_-|}$}
  \put(-100,-10){\Large$\ell$} 
\caption{{\it Comparing the tightness of the  first-passage time bounds Eq.~(\ref{eq:2})  with the thermodynamic uncertainty   relation Eq.~(\ref{eq:T2}).} The ratio $\dot{s} \langle T^n_X \rangle^{1/n}/|\log p_-|$ for $n=1,2,3$ and the thermodynamic uncertainty (TUR) ratio   $\dot{s}(\langle T^2_X\rangle-\langle T_X\rangle^2)/(2\langle T_X\rangle)$ are plotted as a function of $\ell = \ell_-=\ell_+$  for a biased random walk process  $X$ described by Eq.~(\ref{eq:randomwalkX}) with $k_+ = 1$ and $b= 0.1$.     Note that the inequalities Eq.~(\ref{eq:2})   are tight for $\ell\rightarrow \infty$, while the uncertainty relation Eq.~(\ref{eq:T2}) is loose.   } \label{fig2M}
\end{figure}

Figure~\ref{fig2M} compares the  first-passage time bounds Eqs.~(\ref{eq:2}) with the thermodynamic uncertainty relation Eq.~(\ref{eq:T2}).  The plotted curves are obtained from the explicit analytical  expressions for $\dot{s}$ and $p_-$, given by Eqs.~(\ref{eq:sigmaxM}) and (\ref{eq:pMP}), respectively, and from explicit analytical expressions for $\langle T^n \rangle$ that we have obtained from the  Eqs.~(\ref{eq:gJump}-\ref{eq:TMom}) and can be found in the Appendix~\ref{App:E6}.  The figure shows that for large values of the  first-passage thresholds the bounds Eqs.~(\ref{eq:2})  are tight, as predicted by Eq.~(\ref{eq:eq}), while the thermodynamic uncertainty relation is loose.     
 
 In Fig.~\ref{fig2M}  we also observe that the first moment $\langle T\rangle$ converges fast to its asymptotic value, while higher order moments $\langle T^2\rangle$ and $\langle T^3\rangle$ converge slowly to their asymptotic values.      Using Eqs.~(\ref{eq:sigmaxM}), (\ref{eq:pMP}), and (\ref{eq:gJump}-\ref{eq:TMom}), we obtain the asymptotics  (see  Appendices~\ref{App:E8} and \ref{App:E9})   
  \begin{eqnarray}
\frac{\lceil \ell_+\rceil}{\lceil \ell_-\rceil}\frac{|\log p_-|  }{\langle T_X \rangle} &=&  \dot{s}  + O\left(b^{\lceil \ell_-\rceil}\right), \label{eq:T1Asympt}
  \end{eqnarray} 
  and  for $n>1$
\begin{equation}
\frac{\lceil \ell_+\rceil}{\lceil \ell_-\rceil}\frac{|\log p_-|  }{\left(\langle T^n_X \rangle\right)^{1/n}} = \dot{s} + O\left(\frac{1}{\lceil \ell_+\rceil}\right). \label{eq:T2Asympt}
\end{equation}  
Hence, the first moment converges exponentially fast to the entropy production rate $\dot{s}$, while the higher order moments converge as $1/\lceil \ell_+\rceil$ to their asymptotic value.  Consequently, in this example the first moment is more effective for the inference of the entropy production rate $\dot{s}$.    However, from Eq.~(\ref{eq:asymptT}) we can conclude that the exponentially fast convergence for the first moment is a model specific property.  
   
The asymptotic expression for the  thermodynamic uncertainty relation depends on the subleading $O\left(1/\lceil \ell_+\rceil\right)$ term in Eq.~(\ref{eq:T2Asympt}), and is given by 
\begin{equation}
\frac{ 2\langle T_X\rangle}{\langle T^2_X\rangle-\langle T_X\rangle^2} =  2(k_+-k_-) \tanh\left(\frac{a}{2\mathsf{T}_{\rm env}}\right)  + O\left(b^{\lceil \ell_-\rceil}\right).   \label{eq:T2Unc}
\end{equation}  
Since $\tanh(x)\leq x$,  the thermodynamic uncertainty relation Eq.~(\ref{eq:T2}) holds.    However,  contrary to Eqs.~(\ref{eq:T1Asympt}) and (\ref{eq:T2Asympt}), the thermodynamic uncertainty relation is not tight in the limit of large thresholds and the ratio Eq.~(\ref{eq:T2Unc})  depends on the affinity $a/\mathsf{T}_{\rm env}$ of the process.

Taken together, we can conclude that the   equality Eq.~(\ref{eq:eq}), and thus the tightness of the bound Eq.~(\ref{eq:2}) for $J=S$, follows from the universality of the leading order term in the  Eqs.~(\ref{eq:T1Asympt}) and (\ref{eq:T2Asympt}) for $\langle T^n_X\rangle$.  On the other hand, the looseness of the thermodynamic uncertainty relation  Eq.~(\ref{eq:T2}) for $S=J$ is a consequence of the nonuniversality of the subleading term of $\langle T^2_X \rangle$ in the Eqs.~(\ref{eq:T1Asympt}) and (\ref{eq:T2Asympt}) and therefore the right-hand side of Eq.~(\ref{eq:T2Unc}) depends on the affinity $a$ of the process.

 \section{Discussion}  \label{sec:discu}
Driving a system out of equilibrium  can speed up the rate  of a  chemical reaction.    However,   there exists a fundamental thermodynamic tradeoff between speed,  the fluctuations in the process, and the rate of dissipation.   The main contribution of this paper is the derivation of a  universal inequality, Eq.~(\ref{eq:2}), that expresses in   nonequilibrium stationary states a thermodynamic tradeoff  between speed, uncertainty,   and dissipation,   which are quantified in terms of   the mean first passage time $\langle T_J \rangle$,  the splitting probability $p_-$, and the dissipation rate $\dot{s}$, respectively.      The main advantage of the inequality (\ref{eq:2}) with respect to previously published tradeoff relations, such as the thermodynamic uncertainty relations \cite{barato2015thermodynamic, pietzonka2016universalx, gingrich2016dissipation, pietzonka2017finite, horowitz2017proof, proesmans2017discrete,  hasegawa2019fluctuation, shreshtha2019thermodynamic, dechant2020fluctuation, falasco2020unifying}, is that Eq.~(\ref{eq:2}) is an equality  when $J(t)=cS(t)$ with $c$ a time and trajectory independent constant, see  Eq.~(\ref{eq:eq}), and hence the bound is optimal in this case.

From a  physical and mathematical point of view, the Eqs.~(\ref{eq:2}) and  (\ref{eq:eq}) are interesting as they are  related to  thermodynamic uncertainty relations, the   Van't Hoff-Arrhenius law,  martingale theory,   and the theory of sequential hypothesis testing.    Indeed, both Eq.~(\ref{eq:2}) and the thermodynamic uncertainty relations   are a consequence of  the    large deviation function bound Eq.~(\ref{eq:JBoundx}).  On the other hand, the  equality Eq.~(\ref{eq:eq})  follows  from  martingale theory \cite{chetrite2011two, neri2017statistics}, in particular the integral fluctuation relation at stopping times \cite{neri2019integral}.    We have also recovered the  Van't Hoff-Arrhenius law Eq.~(\ref{eq:Arrh})  in the near equilibrium limit $\dot{s}\rightarrow 0$.    In addition,  we have also  derived Eqs.~(\ref{eq:2}) and  (\ref{eq:eq})  from the  theory of sequential hypothesis testing~\cite{lai1981asymptotic, tartakovsky2014sequential}, more specifically, the  asymptotic optimality of sequential probability ratio tests.  It is fascinating that all these different research areas are related to each other and certainly more fundamental insights about stochastic thermodynamics can be gained by exploring the links between these areas.  

The present paper derives the main result Eq.~(\ref{eq:2})  in the setup of currents $J$ in stationary Markov jump processes $X$; in addition, to identify  $\dot{s}$, as defined in Eq.~(\ref{eq:probRatio})-(\ref{eq:sdotDef}), with the mean rate of dissipation we require local detailed balance.    Nevertheless, we expect that (\ref{eq:2})  can be generalised.         In Sec.~\ref{sec:fp}, we have derived  the bound  (\ref{eq:2})  using large deviation theory, in particular, we have used the bound (\ref{eq:JBoundx}) on the large deviation function of the current.  Since  the bound    (\ref{eq:JBoundx}) has been derived for stationary Markov jump processes, see Ref.~\cite{gingrich2016dissipation}, also (\ref{eq:2})  applies to this setup.    Consequently,  (\ref{eq:2})   extends to  processes $X$ for which a bound on the large deviation function of the form (\ref{eq:JBoundx})  holds.   Notable examples worthwhile exploring are overdamped Langevin processes \cite{polettini2016tightening} and asymptotically stationary processes with time-dependent driving \cite{koyuk2020thermodynamic}.    Another possible  avenue of approach for generalising  (\ref{eq:2})   is based  on the  theory of sequential hypothesis testing,  as presented  in Sec.~\ref{sec:seq}.      In this approach, we  have derived the partial result (\ref{eq:2xxxx})  in a very general setup and with full mathematical rigour.  However, to get (\ref{eq:2}) we relied on  the additional result  (\ref{eq:pp}),  which has not been derived with the same level of mathematical rigour as (\ref{eq:2xxxx}).      It will be interesting to establish the conditions under which (\ref{eq:pp}) holds with full mathematical rigour,  as this will pave the way for extensions beyond the setup of stationary Markov jump processes.

 The equality (\ref{eq:eq})  has been derived in  a more general  setup than the bound (\ref{eq:2}).  In Sec.~\ref{sec:seq},   we have presented a rigorous derivation of   (\ref{eq:2})  based on  the   martingale property of $e^{-S}$ and the $r$-quick convergence of $S(t)/t$ to a deterministic limit.      The martingale property of  $e^{-S}$ holds as long as it can be written in the form \cite{chetrite2011two, neri2017statistics, neri2019integral, chetrite2019martingale, neri2020second, PhysRevLett.126.080603} 
 \begin{equation}
e^{-S(t)} =\frac{\tilde{p}(\vec{X}^t_0) }{p(\vec{X}^t_0)}, \label{eq:probRatiox}
\end{equation}  
 with $\tilde{p}$ a probability distribution characterising the statistics of trajectories in the time-reversed process, whereas the $r$-quick convergence is a mild condition  on the fluctuations of $S(t)/t$ in the limit of large $t$.   In  Langevin processes, including   nonstationary   processes,   $e^{-S}$ can be written in the form (\ref{eq:probRatiox}), see e.g.~Refs.~\cite{seifert2012stochastic, neri2020second}, and hence the equality  (\ref{eq:eq})  should also apply to continuous stochastic processes.

We end the paper with a brief discussion of potential applications for the Eqs.~(\ref{eq:2}) and (\ref{eq:eq}).  The inequality Eq.~(\ref{eq:2}) could be used to  infer dissipation rates from   the measurements of first-passage times of stochastic currents.  It is  difficult to measure the entropy production rate directly as it is related to the heat exchanged with the environment \cite{ciliberto2017experiments}.    However, since the mean first-passage time $\langle T_J\rangle$ and the splitting probability $p_-$  are directly measurable quantities,  Eq.~(\ref{eq:2}) can be used to  bound  the entropy production rate from below.      When compared with other methods that infer entropy production rates  from the measurements of stochastic currents, see e.g.~\cite{pietzonka2016universal, gingrich2017inferring, seifert2019stochastic, van2020entropy,  manikandan2020inferring},   the present inequalities may turn out to perform better as they are optimal  when $J=S$, although this requires further study as the inequality Eq.~(\ref{eq:2}) has also some drawbacks.   In particular, the probability $p_-$ decreases exponentially with $\ell_-$, which raises the question how  $p_-$ can be estimated at large values of $\ell_-$.    A second interesting application is in the study of first-passage  problems of nonequilibrium processes, such as those of  self-propelled particles \cite{angelani2014first, malakar2018steady, dhar2019run, biswas2020first, walter2021first}.    The inequality (\ref{eq:2})  and the equality (\ref{eq:eq}) are generic results with a clear physical meaning, and hence, when used to  bound the statistics of first-passage problems in nonequilibrium processes,  can provide  further physical understanding of mathematical results.   A third interesting application is  in the use of the bound Eq.~(\ref{eq:2}) to determine how far molecular systems operate from what is physically nonpermissible.   Notable examples are molecular motors  that are responsible for copying genetic information in biological cells, such as ribosomes or polymerases.  These motors are known to attain a reliability  that is larger than what is possible in equilibrium  through kinetic proof reading  \cite{hopfield1974kinetic, murugan2012speed, mallory2020we}, but it is not known how close  to the physically nonpermissible limits these motors operate.   Other examples are  transistors that are small enough so that they are prone to noise \cite{gao2021principles}.    Bounds of the form  Eq.~(\ref{eq:2}) could be used to  understand  thermodynamic limitations  on computing  that are based on the tradeoff between dissipation, speed, and uncertainty in nonequilibrium processes.

\section*{Acknowledgements}

The author thanks Andre Barato, Patrick Pietzonka, and Benjamin Walter for insightful discussions.

\appendix

  \section{Martingales}\label{app:A}
  In this appendix, we  state the definition of a martingale and one of its key properties that we use repeatedly in this paper, namely,  Doob's optional stopping theorem.  
  
  \subsection{Definition of a martingale} \label{App:A1} 
  Let $\Omega$ be the set of all realisations of a physical process $X$, which is endowed with a $\sigma$-algebra $\mathscr{F}$.   Let $\mathbb{P}$ be a probability measure that determines the probabilities $\mathbb{P}[\Phi]$  of  events $\Phi\in\mathscr{F}$.      We denote averages with respect to $\mathbb{P}$ by $\langle \cdot\rangle$.  Let  $\left\{\mathscr{F}(t)\right\}_{t\geq 0}$ be the filtration generated by $X$, i.e.,  a sequence of sub-$\sigma$-algebras $\mathscr{F}(t)$ that is generated by the trajectories $X^t_0$ of the process $X$.      
  
  A martingale $M(t)$ with respect to a filtration $\left\{\mathscr{F}(t)\right\}_{t\geq 0}$ is a stochastic process for which (i) the process $M(t)$ is $\mathscr{F}(t)$-measurable (ii) $\langle |M(t)|\rangle<\infty$ (iii) $\langle M(t)|\mathscr{F}(s)\rangle = M(s)$ for all $s\leq t$~\cite{doob1953stochastic, liptser2013statistics}.   The latter condition implies that the martingale $M$ is a driftless process.
 
  \subsection{Doob's optional stopping theorem}\label{App:A2}
   A stopping time $T$ is a random time $T:\Omega\rightarrow \mathbb{R}^+\cup \left\{\infty\right\}$ such that $\left\{T\leq t\right\}\in \mathscr{F}(t)$ for all values of $t\in\mathbb{R}^+$.      This means that $T$ stops the process $X$ based on a stopping rule  that does not anticipate the future or use side information.

  One of the key properties of martingales that we use in this paper is described by    Doob's optional stopping theorem \cite{liptser2013statistics}.     
    \begin{theorem}[Doob's optional stopping theorem] \label{TheDoob} Let $(\Omega,\mathscr{F},\mathbb{P})$  be a probability space with sample space $\Omega$, $\sigma$-algebra $\mathscr{F}$, and probability measure $\mathbb{P}$.   Let $X(t)$ with $t\geq 0$ be a $\mathscr{F}$-measurable stochastic process  and let $\left\{\mathscr{F}(t)\right\}_{t\geq 0}$ be the filtration generated by $X$.   
    Let $M$ be a martingale process with respect to the filtration $\left\{\mathscr{F}(t)\right\}_{t\geq 0}$ and let $T$ be a stopping time relative to the filtration  $\left\{\mathscr{F}(t)\right\}_{t\geq 0}$.    It holds then that 
    \begin{equation}
    \langle M(T\wedge t) \rangle = \langle M(0) \rangle
    \end{equation} 
    where $T\wedge t = {\rm min}\left\{T,t\right\}$.
\end{theorem}

\section{Mean first-passage time for an overdamped Brownian particle   in a generic periodic potential and in  a uniform force field} \label{eq:Period} 
In this appendix, we analyse the first-passage problem  for  a Brownian motion   in a generic periodic potential $u$ and a uniform force field $f$, as described by Eq.~(\ref{eq:nonequilb}).    In particular, we derive analytical expressions for the mean first-passage time $\langle T_X\rangle$, the splitting probability $p_-$, and the mean entropy production rate $\dot{s}$, where $T_X$ is defined as in Eq.~(\ref{eq:def2}).   In the limit of large thresholds $\ell_-=\ell_+=\ell\gg 1$, we show that the main result Eq.~(\ref{eq:eq}) holds.   In addition, in the near-equilibrium limit and at low temperatures, we show that Eq.~(\ref{eq:eq}) is a Van't Hoff-Arrhenius law.

 \subsection{Stationary distribution and current}\label{appC1}
 We derive Eq.~(\ref{eq:jssx}) in the main text for the stationary current $j_{\rm ss}$.
 
The stationary distribution of $X\in \mathbb{R}$ does not exist.  However, we can define the process on a ring with  periodic boundary conditions such that $X(t) = X(t)+\delta$.   The stationary state $p_{\rm ss}$ of the equivalent process defined on a ring exists, and we can  use this  process on a ring to determine the stationary current $j_{\rm ss}$.     
 
The stationary distribution $p_{\rm ss}$ solves the equation  \cite{gardiner1985handbook, seifert2012stochastic}
\begin{equation}
\partial_x j_{\rm ss}(x) = 0 \label{eq:jDiff}
\end{equation}
with periodic boundary conditions $p_{\rm ss}(x) = p_{\rm ss}(x+\delta)$,
where 
\begin{equation}
j_{\rm ss}(x) = \mu (f-\partial_x u(x)) p_{\rm ss}(x) - \frac{\mathsf{T}_{\rm env}}{\gamma}\partial_x p_{\rm ss}(x). \label{eq:jsp}
\end{equation}
The solution to Eq.~(\ref{eq:jDiff}) is given by  \cite{risken1996fokker, neri2019integral}
\begin{equation}
p_{\rm ss}(x) =  \frac{w(x) \left(\int^{x+\delta}_{x}{\rm d}x' \frac{1}{w(x')}\right)}{\int^\delta_0 {\rm d}y\: w(y) \left(\int^{y+\delta}_{y}{\rm d}x' \frac{1}{w(x')}\right)} \label{eq:pss}
\end{equation}
with $x\in[0,\delta]$, and 
where 
\begin{eqnarray}
w(x)   = e^{-\frac{u(x)-fx}{\mathsf{T}_{\rm env}}} \label{eq:w}.
\end{eqnarray}
The expression Eq.~(\ref{eq:jssx})  for the  stationary current $j_{\rm ss}$  follows readily from the Eqs.~(\ref{eq:jsp}) and (\ref{eq:pss}).

 \subsection{Entropy production}\label{appC2}
 We derive Eqs.~(\ref{eq:sdot}) and (\ref{eq:SBrown}) in the main text for the entropy production rate $\dot{s}$ and the stochastic entropy production $S$, respectively.  We will again use the equivalent process defined on a ring with periodic boundary conditions.

The stochastic entropy production $S$ of $X$, as defined in Eq.~(\ref{eq:probRatio}), is determined by the stochastic differential equation \cite{seifert2005entropy, pigolotti2017generic}
\begin{equation}
{\rm d}S  = v_S(X)  \: {\rm d}t+ \sqrt{2v_S(X)} \: {\rm d} W(t),
\end{equation}
where  
\begin{equation}
v_S(x)=  \frac{\gamma}{\mathsf{T}_{\rm env}}\frac{j^{2}_{\rm ss}}{p^2_{\rm ss}(x)}=   \frac{\mathsf{T}_{\rm env}}{\gamma}   \frac{\left(1 - e^{\frac{-f\delta}{\mathsf{T}_{\rm env}}}\right)^2}{ w^2(x) \left(\int^{x+\delta}_{x}{\rm d}x' \frac{1}{w(x')}\right)^2} . 
\end{equation}      
Alternatively, we can write 
\begin{equation}
S(t)  = \frac{fX(t)- u(X(t)) + u(X(0))}{\mathsf{T}_{\rm env}} +  \log \frac{ p_{\rm ss}(X(0))}{p_{\rm ss}(X(t))} .
\end{equation}
The latter formula  implies  for large $t\gg 1$  that 
\begin{equation}
S(t) =  \frac{f X(t)}{\mathsf{T}_{\rm env}} + o(t), 
\end{equation}
which is Eq.~(\ref{eq:SBrown}) in the main text.

The average stationary entropy production  rate is  given by 
\begin{eqnarray}
\dot{s} &=& \frac{ \langle S(t)\rangle}{t} =  \langle v_S\rangle  =  \frac{\gamma j^2_{\rm ss} }{\mathsf{T}_{\rm env}}\int^\delta_0 \frac{{\rm d}x}{p_{\rm ss}(x)}  .
\end{eqnarray}  
Since the stationary distribution $p_{\rm ss}$ is given by Eq.~(\ref{eq:pss}) and $u(x)$ is a periodic function, we can express this also as 
\begin{eqnarray}
\dot{s} &=&  j_{\rm ss}   \left(1- e^{\frac{-f\delta}{\mathsf{T}_{\rm env}}} \right)   \int^\delta_0  {\rm d}x    \frac{1}{w(x) \left(\int^{\delta}_{0}{\rm d}x' \frac{1}{w(x')} - (1-e^{-\frac{f\delta}{\mathsf{T}_{\rm env}}})\int^{x}_{0}{\rm d}x'\frac{1}{w(x')}\right) }. 
\end{eqnarray}  
Introducing the function 
\begin{equation}
\int^{x}_{0}{\rm d}x'\frac{1}{w(x')} = W(x),
\end{equation}
we find that 
\begin{eqnarray}
\dot{s} 
&=& j_{\rm ss}   \left(1- e^{\frac{-f\delta}{\mathsf{T}_{\rm env}}} \right)   \int^{W(\delta)}_0  {\rm d} u   \frac{1}{W(\delta) - (1-e^{-\frac{f\delta}{\mathsf{T}_{\rm env}}})u }  .
\end{eqnarray}  
Integrating yields the expression for $\dot{s}$ given by Eq.~(\ref{eq:sdot}) in the main text.

\subsection{Splitting probabilities}\label{App:C3}
We use the martingale property of $e^{-S(t)}$, see Refs.~\cite{neri2017statistics, neri2019integral} or Appendix~\ref{app:A},  to  determine the splitting probabilities $p_-$ and $p_+$.    Doob's optional stopping theorem for martingales  implies the following integral fluctuation relation at stopping times
\begin{equation}
\langle e^{-S(T_X)} |X(0)=0 \rangle = e^{-S(0)} = 1,
\end{equation}
and  since $S(t)$ is  continuous as a function of $t$ this implies that, see Refs.~\cite{neri2017statistics, neri2019integral},  
\begin{equation}
p_-= e^{-s_-}\frac{1-e^{-s_+}}{1-e^{-s_--s_+}}, \quad  {\rm and} \quad p_+  =  \frac{1-e^{-s_-}}{1-e^{-s_--s_+}}, \label{eq:pMpP}
\end{equation}
where 
\begin{equation}
s_- = -\frac{-f\ell - u(-\ell) + u(0)}{\mathsf{T}_{\rm env}} - \log \frac{ p_{\rm ss}(0)}{p_{\rm ss}(-\ell)}  ,\quad {\rm and} \quad  s_+ =  \frac{f\ell- u(\ell) + u(0)}{\mathsf{T}_{\rm env}} +  \log \frac{ p_{\rm ss}(0)}{p_{\rm ss}(\ell)}.   \label{eq:SM}
\end{equation} 
Notice that we have used a slight abuse of notation in the sense that  $u(x)$ and $p_{\rm ss}(x)$ are here defined on $x\in\mathbb{R}$ using $u(x)=u(x\pm \delta)$ and $p_{\rm ss}(x) = p_{\rm ss}(x\pm \delta)$.

   \subsection{Mean first-passage time}   \label{App:C4}
   Consider the backward Fokker-Planck equation
   \begin{equation}
   \mu \left(f - \partial_x u(x)\right) \partial_x t(x) +\frac{\mathsf{T}_{\rm env}}{\gamma}\partial^2_x t(x) = -1 \label{eq:TBack}
   \end{equation} 
    with boundary conditions $t(-\ell) = t(\ell) = 0$.  It then holds that, see Ref.~\cite{grebenkov2014first},
        \begin{equation}
   \langle T_X |X(0) = x\rangle = t(x).
   \end{equation}
    
  The solution of $t(x)$ to Eq.~(\ref{eq:TBack}) with boundary conditions $t(-\ell) = t(\ell) = 0$ is  given by
  \begin{equation}
  t(x) =  \frac{\gamma}{\mathsf{T}_{\rm env}}\left( \int^\ell_{-\ell}{\rm d}y\frac{1}{w(y)} \int^{y}_0{\rm d}x' w(x')\right)    \left(\frac{\int^{x}_{-\ell}{\rm d}y\frac{1}{w(y)}}{\int^\ell_{-\ell}{\rm d}y\frac{1}{w(y)}}  -\frac{\int^{x}_{-\ell}{\rm d}y\frac{1}{w(y)}  \int^{y}_{0}{\rm d}x'w(x')}{\int^\ell_{-\ell}{\rm d}y\frac{1}{w(y)} \int^{y}_0{\rm d}x' w(x')} \right),
  \end{equation}
   and therefore 
  \begin{equation}
  \langle T_X\rangle = \frac{\gamma}{\mathsf{T}_{\rm env}} \left(\int^\ell_{-\ell}{\rm d}y\frac{1}{w(y)} \int^{y}_0{\rm d}x' w(x') \right)   \left(\frac{\int^{0}_{-\ell}{\rm d}y\frac{1}{w(y)}}{\int^\ell_{-\ell}{\rm d}y\frac{1}{w(y)}}  -\frac{\int^{0}_{-\ell}{\rm d}y\frac{1}{w(y)}  \int^{y}_{0}{\rm d}x'w(x')}{\int^\ell_{-\ell}{\rm d}y\frac{1}{w(y)} \int^{y}_0{\rm d}x' w(x')} \right). \label{eq:Tmean}
    \end{equation}

    In order to better understand the structure of the expression  Eq.~(\ref{eq:Tmean})  for the mean-first passage time, it is useful to  express the integrals in Eq.~(\ref{eq:Tmean}) that run over the intervals $[-\ell,\ell]$ and $[-\ell,0]$  in terms of integrals that run over  the interval $[0,\delta]$.    Let $n = \lfloor \ell/\delta \rfloor$ be the largest integer  smaller than or equal to $\ell/\delta$, then we can write
\begin{equation}
 \ell =  n  \delta  + z,
\end{equation}
with $z\in[0,\delta]$.   Using this decomposition for $\ell$, we obtain that 

\begin{eqnarray}
\int^{0}_{-n\delta - z}{\rm d}y\frac{1}{w(y)} &=& e^{n\frac{f \delta }{\mathsf{T}_{\rm env}}}\left\{\left(  \frac{1-e^{-n\frac{f\delta }{\mathsf{T}_{\rm env}}}}{1-e^{-\frac{f\delta }{\mathsf{T}_{\rm env}}}} \right)  \int^{\delta}_0\frac{{\rm d}x}{w(x)}  +   e^{\frac{f \delta   }{\mathsf{T}_{\rm env}}}   \int^{\delta}_{\delta-z}\frac{{\rm d}x}{w(x)}\right\} \label{eq:Integral1}
\end{eqnarray}
and 
\begin{eqnarray}
\lefteqn{\int^{n\delta+z}_{-n\delta - z}{\rm d}y\frac{1}{w(y)} } && \nonumber\\
&=& e^{n\frac{f \delta }{\mathsf{T}_{\rm env}}} \left\{ \left(\frac{1-e^{-2n\frac{f\delta }{\mathsf{T}_{\rm env}}}}{1-e^{-\frac{f\delta }{\mathsf{T}_{\rm env}}}}  \right) \int^{\delta}_0\frac{{\rm d}x}{w(x)}   +  e^{\frac{f \delta   }{\mathsf{T}_{\rm env}}}   \int^{\delta}_{\delta-z}\frac{{\rm d}x}{w(x)}  +  e^{-2n\frac{f \delta   }{\mathsf{T}_{\rm env}}}   \int^{z}_{0}\frac{{\rm d}x}{w(x)} \right\}. \nonumber\\\label{eq:Integral2}
\end{eqnarray}

In addition, 
\begin{eqnarray}
\lefteqn{\int^{n\delta +z}_{0}{\rm d}y\frac{1}{w(y)}  \int^{y}_{0}{\rm d}x'w(x') } && \nonumber\\ 
&=& n\left\{\frac{e^{-\frac{f\delta}{\mathsf{T}_{\rm env}}}}{1-e^{-\frac{f\delta}{\mathsf{T}_{\rm env}}}}\int^{\delta}_0 {\rm d}x \frac{1}{w(x)}  \int^{\delta}_0 {\rm d}x \, w(x) +   \int^{\delta}_{0} {\rm d}y \frac{1}{w(y)} \int^{y}_{0} w(x) {\rm d}x \right\}  \nonumber\\
&&
- \frac{e^{-\frac{f\delta}{\mathsf{T}_{\rm env}}}\left(1-e^{-n\frac{f\delta }{\mathsf{T}_{\rm env}}}\right)}{\left(1-e^{-\frac{f\delta}{\mathsf{T}_{\rm env}}}\right)^2} \int^{\delta}_0 {\rm d}x \, w(x)\int^{\delta}_0 {\rm d}x \frac{1}{w(x)}  \nonumber\\
&& +  e^{-\frac{f\delta}{\mathsf{T}_{\rm env}}}   \frac{1- e^{-n\frac{f \delta }{\mathsf{T}_{\rm env}}}}{1-e^{-\frac{f \delta }{\mathsf{T}_{\rm env}}}} \int^{z}_{0} {\rm d}y \frac{1}{w(y)}  \int^{\delta}_0 {\rm d}x \,  w(x) + \int^{ z}_{0} {\rm d}y \frac{1}{w(y)}  \int^{y}_{0} {\rm d}x \, w(x) ,\label{eq:Integral3}
\end{eqnarray}
and 
\begin{eqnarray}
\lefteqn{-\int^{0}_{-n\delta - z}{\rm d}y\frac{1}{w(y)}  \int^{y}_{0}{\rm d}x'\,w(x') }&&  \nonumber\\
&=&\frac{1-e^{n \frac{f \delta }{\mathsf{T}_{\rm env}}}}{(1-e^{-\frac{f \delta }{\mathsf{T}_{\rm env}}})(1-e^{\frac{f \delta }{\mathsf{T}_{\rm env}}})}  \left(\int^{\delta}_0 {\rm d}x \, w(x)\right) \left(\int^{\delta}_0 {\rm d}x \frac{1}{w(x)}\right)  
\nonumber\\ 
&& 
+ n  \left\{ \int^{\delta}_0 {\rm d}y\frac{1}{w(y)} \int^{\delta}_{y} {\rm d}x\,w(x)   - \frac{1}{1-e^{-\frac{f \delta }{\mathsf{T}_{\rm env}}}} \left(\int^{\delta}_0 {\rm d}x \, w(x)\right) \left(\int^{\delta}_0 {\rm d}x \frac{1}{w(x)}\right)    \right\}   
\nonumber\\ 
&& +    \frac{e^{n \frac{f \delta }{\mathsf{T}_{\rm env}}}- 1}{1-e^{-\frac{f \delta }{\mathsf{T}_{\rm env}}}}\int^{\delta}_{\delta-z} {\rm d}x\frac{1}{w(x)}  \int^{\delta}_{0} {\rm d}x \, w(x)  +   \int^{\delta}_{\delta-z} {\rm d}y\frac{1}{w(y)}  \int^{\delta}_{y} {\rm d}x\,w(x).\label{eq:Integral4}
\end{eqnarray}  
Using the Eqs.~(\ref{eq:Integral1}), (\ref{eq:Integral2}), (\ref{eq:Integral3}), and (\ref{eq:Integral4})  in Eq.~(\ref{eq:Tmean}), we obtain an expression for   $\langle T_X\rangle$ that depends only on integrals over the  interval  $[0,\delta]$.

   \subsection{Limit of large thresholds}  \label{App:C5}
   We derive the Eq.~(\ref{eq:asymptT}) that holds in the limit of large $\ell$.   The derivation goes in three steps.   First, in Sec.~\ref{subsec:v1} we derive an asymptotic expression for $p_-$, second in Sec.~\ref{subsec:v2} we derive an asymptotic expression for $\langle T_X\rangle$, lastly in Sec.~\ref{subsec:v3} we combine these two results to determine the ratio $|\log p_-|/\langle T_X\rangle$.
      
  \subsubsection{Splitting probabilities}     \label{subsec:v1}
  In the limit of large thresholds,  the linear term in $\ell$ dominates the Eqs.~(\ref{eq:SM})  and therefore
   \begin{equation}
s_- = \frac{f\ell}{\mathsf{T}_{\rm env}} + O_\ell(1),\quad {\rm and} \quad s_+ = \frac{f\ell}{\mathsf{T}_{\rm env}} + O_\ell(1).  \label{eq:SMAsympt}
\end{equation}
Using Eq.~(\ref{eq:SMAsympt}) in the  Eqs.~(\ref{eq:pMpP}) for $p_-$ and $p_+$, we obtain that 
   \begin{equation}
\log p_-= -\frac{f\ell}{\mathsf{T}_{\rm env}} + O_\ell(1), \quad {\rm and}  \quad    \log p_+=  O_\ell(1). \label{eq:pMAsympto}
\end{equation} 

  \subsubsection{Mean first-passage time}    \label{subsec:v2}
We recall that 
$n = \lfloor\frac{\ell}{\delta} \rfloor$, 
where as before $\lfloor \frac{\ell}{\delta}\rfloor$ denotes the largest integer that is smaller than or equal to $\frac{\ell}{\delta}$.  

Taking the asymptotic limit of large $\ell$ in Eqs.~(\ref{eq:Integral1}) and (\ref{eq:Integral2}), we obtain that 
   \begin{eqnarray}
   \frac{\int^{0}_{-\ell}{\rm d}y\frac{1}{w(y)}}{\int^{\ell}_{-\ell}{\rm d}y\frac{1}{w(y)} } 
   &=& 1  -  e^{-\lfloor\frac{\ell}{\delta}\rfloor\frac{f \delta }{\mathsf{T}_{\rm env}}} \frac{\int^{\delta}_0\frac{{\rm d}x}{w(x)}  }{\int^{\delta}_0\frac{{\rm d}x}{w(x)} +  (e^{\frac{f \delta   }{\mathsf{T}_{\rm env}}} -1)  \int^{\delta}_{\delta-z}\frac{{\rm d}x}{w(x)}  }+ O\left(e^{-2\lfloor\frac{\ell}{\delta}\rfloor\frac{f \delta }{\mathsf{T}_{\rm env}}}\right) . \label{eq:asympt1}
   \end{eqnarray}
 The asymptotic limit of Eq.~(\ref{eq:Integral3}) is
      \begin{eqnarray}
 \lefteqn{  \int^{\ell}_{0}{\rm d}y\frac{1}{w(y)}  \int^{y}_{0}{\rm d}x'w(x') } && \nonumber\\ 
 &=& \Big\lfloor\frac{\ell}{\delta}\Big\rfloor\left\{\frac{e^{-\frac{f\delta}{\mathsf{T}_{\rm env}}}}{1-e^{-\frac{f\delta}{\mathsf{T}_{\rm env}}}}\int^{\delta}_0 {\rm d}x \frac{1}{w(x)}  \int^{\delta}_0 {\rm d}x w(x) +   \int^{\delta}_{0} {\rm d}y \frac{1}{w(y)} \int^{y}_{0} w(x) {\rm d}x \right\}  + O_{\ell}(1),  \nonumber\\\label{eq:asympt2}
      \end{eqnarray}
      and from Eqs.~(\ref{eq:Integral3}) and (\ref{eq:Integral4}) it follows that 
      \begin{eqnarray}
\lefteqn{-\int^{\ell}_{-\ell}{\rm d}y\frac{1}{w(y)}  \int^{y}_{0}{\rm d}x'w(x') } && \nonumber\\
 &=&e^{\lfloor \frac{\ell}{\delta}\rfloor\frac{f \delta }{\mathsf{T}_{\rm env}}}  \left\{ \frac{ \int^{\delta}_0 {\rm d}x \, w(x)\int^{\delta}_0 {\rm d}x \frac{1}{w(x)}  }{(1-e^{-\frac{f \delta }{\mathsf{T}_{\rm env}}})(e^{\frac{f \delta }{\mathsf{T}_{\rm env}}}-1)}   +  \frac{\int^{\delta}_{\delta-z} {\rm d}x\frac{1}{w(x)}  \int^{\delta}_{0} {\rm d}x w(x)  }{1-e^{-\frac{f \delta }{\mathsf{T}_{\rm env}}}} \right\} \nonumber\\
&& + \Big \lfloor \frac{\ell}{\delta}\Big\rfloor\left\{ \int^{\delta}_0 {\rm d}y\frac{1}{w(y)} \int^{\delta}_{y} {\rm d}x\,w(x)  -  \frac{1}{\tanh\left(\frac{f\delta}{2\mathsf{T}_{\rm env}} \right)}
\int^{\delta}_0 {\rm d}x \frac{1}{w(x)}  \int^{\delta}_0 {\rm d}x \, w(x) \right.
\nonumber\\ 
&&  \left. -   \int^{\delta}_{0} {\rm d}y \frac{1}{w(y)} \int^{y}_{0}  {\rm d}x  \, w(x)
 \right\}  + O_{\ell}(1) .\label{eq:asympt3}
\end{eqnarray}
 The Eqs.~(\ref{eq:asympt2}) and  (\ref{eq:asympt3}) imply that the ratio 
\begin{eqnarray} 
\lefteqn{\frac{\int^{0}_{-\ell}{\rm d}y\frac{1}{w(y)}  \int^{y}_{0}{\rm d}x'w(x') }{\int^{\ell}_{-\ell}{\rm d}y\frac{1}{w(y)}  \int^{y}_{0}{\rm d}x'w(x') }}\nonumber\\ 
 &=& 1 +\Big\lfloor\frac{\ell}{\delta}\Big\rfloor e^{-\lfloor \frac{\ell}{\delta}\rfloor\frac{f \delta }{\mathsf{T}_{\rm env}}}  \left\{\frac{  e^{-\frac{f \delta }{\mathsf{T}_{\rm env}}} \left(\int^{\delta}_0 {\rm d}x \, w(x)\right) \left(\int^{\delta}_0 {\rm d}x \frac{1}{w(x)}\right) +  \left(1-e^{-\frac{f \delta }{\mathsf{T}_{\rm env}}}\right)  \int^{\delta}_{0} {\rm d}y \frac{1}{w(y)} \int^{y}_{0} w(x) {\rm d}x }{\frac{ \int^{\delta}_0 {\rm d}x w(x)\int^{\delta}_0 {\rm d}x \frac{1}{w(x)}  }{e^{\frac{f \delta }{\mathsf{T}_{\rm env}}}-1}   +  \int^{\delta}_{\delta-z} {\rm d}x\frac{1}{w(x)}  \int^{\delta}_{0} {\rm d}x \,  w(x)   } \right\} 
\nonumber \\ 
 &&+ O\left(e^{-\lfloor \frac{\ell}{\delta}\rfloor\frac{f \delta }{\mathsf{T}_{\rm env}}}\right). \label{eq:asympt4}
\end{eqnarray} 
Using Eqs.~(\ref{eq:asympt1})-(\ref{eq:asympt4}) in  Eq.~(\ref{eq:Tmean}) yields  for the mean first-passage time the asymptotic expression
   \begin{eqnarray}
  \langle T_X\rangle = \frac{\gamma}{\mathsf{T}_{\rm env}} \Big\lfloor \frac{\ell}{\delta}\Big\rfloor\left[ \frac{e^{-\frac{f \delta }{\mathsf{T}_{\rm env}}}}{1-e^{-\frac{f \delta }{\mathsf{T}_{\rm env}}}} \left(\int^{\delta}_0 {\rm d}x w(x)\right) \left(\int^{\delta}_0 {\rm d}x \frac{1}{w(x)}\right)  +   \int^{\delta}_{0} {\rm d}y \frac{1}{w(y)} \int^{y}_{0} w(x) {\rm d}x \right] + O_{\ell}(1). \nonumber\\\label{eq:TMeanAsympto}
   \end{eqnarray}  
   
     \subsubsection{The ratio $|\log p_-|/\langle T_X\rangle$}        \label{subsec:v3}
It follows from the asymptotic relations for  $\langle T_X\rangle$ and $|\log p_-|$, given by Eqs.~(\ref{eq:TMeanAsympto})  and (\ref{eq:pMAsympto}), respectively, that the ratio 
   \begin{eqnarray}
   \frac{|\log p_-|}{  \langle T_X\rangle }
    &=&  \frac{f\delta }{\gamma}\frac{1- e^{\frac{-f\delta}{\mathsf{T}_{\rm env}}}}{\int^\delta_0 {\rm d}y\: w(y) \left(\int^{y+\delta}_{y}{\rm d}x' \frac{1}{w(x')}\right)} + O(1/\ell).
   \end{eqnarray}
   Using Eqs.~(\ref{eq:sdot}) and  (\ref{eq:jssx}) for $\dot{s}$ and $j_{\rm ss}$, respectively, together with the identities
   \begin{equation}
 \int^{\delta}_{0} {\rm d}y \frac{1}{w(y)} \int^{y}_{0} {\rm d}x \, w(x)  =  \int^{\delta}_{0} {\rm d}y\,w(y)  \int^{\delta}_{y}  \frac{1}{w(x)} {\rm d}x 
\end{equation}
and 
\begin{equation}
e^{-\frac{f \delta }{\mathsf{T}_{\rm env}}} \int^{\delta}_0 {\rm d}x \, w(x)\int^{y}_0 {\rm d}x \frac{1}{w(x)}   =\int^{\delta}_0 {\rm d}x \, w(x)\int^{y+\delta}_
{\delta} {\rm d}x \frac{1}{w(x)}  ,
\end{equation}
we readily obtain Eq.~(\ref{eq:asymptT}),
       which is what we were meant to show.

\subsection{Van't Hoff-Arrhenius law near equilibrium} \label{app:C6}
We show that Eq.~(\ref{eq:asymptT})  yields the Van't Hoff-Arrhenius law Eq.~(\ref{eq:Arrhx}).

  Indeed, if $\ell$ is large enough, then  Eq.~(\ref{eq:asymptT}) together with Eq.~(\ref{eq:pMAsympto}) yields
   \begin{equation}
 \langle T_X\rangle = \frac{f\ell}{\mathsf{T}_{\rm env}} \frac{1}{\dot{s}} + O_\ell(1) \label{eq:TMeanSdot}
 \end{equation}  
 where the mean entropy production rate $\dot{s}$ is given by Eq.~(\ref{eq:sdot}).    Since the mean entropy production rate is proportional to the stationary current, given by Eq.~(\ref{eq:jssx}), we can use  saddle point integrals to evaluate the mean current in the limit $\mathsf{T}_{\rm env}\rightarrow 0$ and  to obtain the  Van't Hoff-Arrhenius law.  
 
 Let us therefore first revisit the saddle point method in Sec.~\ref{secsaddle}, and then apply it to the mean current to obtain the  Van't Hoff-Arrhenius law in Sec.~\ref{secAsykpt}.

 \subsubsection{Saddle point integrals in the limit of $\mathsf{T}_{\rm env}\rightarrow 0$} \label{secsaddle}
We first  revisit briefly the saddle point method.  
 
Let $v(x)$  be a function defined on the interval $[0,\delta]$.   We consider integrals of the form 
  \begin{equation}
 \int^\delta_0 {\rm d}x \: e^{\frac{v(x)}{\mathsf{T}_{\rm env}}} f(x) 
  \end{equation} 
in the limiting case of small $\mathsf{T}_{\rm env}$.    In this limiting case, 
  \begin{equation}
 \int^\delta_0 {\rm d}x \: e^{\frac{v(x)}{\mathsf{T}_{\rm env}}} f(x) = \kappa f(x_{\rm max})e^{\frac{v_{\rm max}}{\mathsf{T}_{\rm env}}} + O\left(\frac{\mathsf{T}_{\rm env}}{v_{\rm max}}\right)
  \end{equation} 
  where the prefactor $\kappa$  depends on the properties of the function $v$ at the maximum.   Below, we consider four relevant cases for $\kappa$.    Note that we use the following notation: if $x_{\rm max} = {\rm argmax}\: v(x)$, then $v_{\rm max} = v(x_{\rm max})$, $v'_{\rm max} = v'(x_{\rm max})$, and  $v''_{\rm max} = v''(x_{\rm max})$.  The four cases are the following:
\begin{itemize}
\item $v'_{\rm max} = 0$ and $x_{\rm max}\in(0,\delta)$:
 \begin{equation}
\kappa = \sqrt{\frac{2\pi \mathsf{T}_{\rm env}}{-v''_{\rm max}}};
 \end{equation}  
 \item $v'_{\rm max}$ does not exist (maximum is a cusp) and $x_{\rm max}\in(0,\delta)$:
  \begin{equation}
\kappa =  \mathsf{T}_{\rm env}  \left( \frac{1}{v^+_{\rm max}} - \frac{1}{v^-_{\rm max}}\right)
 \end{equation} 
 where 
   \begin{equation}
   v^+_{\rm max} =  \lim_{\epsilon\rightarrow 0} \frac{v(x_{\rm max})-v(x_{\rm max}-\epsilon)}{\epsilon}, \quad {\rm and} \quad     v^-_{\rm max} =  \lim_{\epsilon\rightarrow 0} \frac{v(x_{\rm max}+\epsilon)-v(x_{\rm max})}{\epsilon};
    \end{equation} 
\item $x_{\rm max}=0$:
  \begin{equation}
\kappa = -    \frac{\mathsf{T}_{\rm env}}{v^-_{\rm max}};
 \end{equation} 
\item $x_{\rm max}=\delta$:
  \begin{equation}
\kappa =   \frac{\mathsf{T}_{\rm env}  }{v^+_{\rm max}} .
 \end{equation} 
\end{itemize}

  \subsubsection{The mean first-passage time in the low temperature limit and the linear response limit}\label{secAsykpt}
  To derive the Arrhenius law, we take two limits, viz., the    near equilibrium limit  $f\delta/\mathsf{T}_{\rm env}\approx 0$ and the    low temperature limit  $\mathsf{T}_{\rm env}\approx 0$.    Note that we have already taken the large threshold limit in Eq.~(\ref{eq:TMeanSdot}).    Hence, the order of the limits is such that we  first take the large threshold limit, then  the near equilibrium limit, and lastly the low temperature limit. 
  
Taking the linear response limit with $f\delta/\mathsf{T}_{\rm env}\approx 0$, we obtain
\begin{equation}
w(x) = e^{-\frac{u(x)}{\mathsf{T}_{\rm env}}}\left(1 + \frac{f x}{\mathsf{T}_{\rm env}} + O\left(\left(\frac{f\delta}{\mathsf{T}_{\rm env}} \right)^2\right)\right),
\end{equation}
and 
\begin{equation}
\frac{1}{w(x)} =  e^{\frac{u(x)}{\mathsf{T}_{\rm env}}}\left(1 - \frac{f x}{\mathsf{T}_{\rm env}} + O\left(\left(\frac{f\delta}{\mathsf{T}_{\rm env}} \right)^2\right)\right),
\end{equation}
such that 
\begin{eqnarray}
j_{\rm ss} 
&=& \frac{f\delta }{\gamma}\frac{1}{\int^{\delta}_0 {\rm d}y e^{-\frac{u(y)}{\mathsf{T}_{\rm env}}} \int^{\delta}_0 {\rm d}x e^{\frac{u(x)}{\mathsf{T}_{\rm env}}}    }    + O\left(\left(\frac{f\delta}{\mathsf{T}_{\rm env}} \right)^2\right) . \label{jsss}
\end{eqnarray}   

Second, we take the low temperature limit with $\mathsf{T}_{\rm env}\approx0$.   
Using the saddle point method, we obtain that 
\begin{eqnarray}
j_{\rm ss} &=&  \frac{f\delta}{\gamma} \kappa_1 \kappa_2  e^{-\frac{E_{\rm b}}{\mathsf{T}_{\rm env}}} + O\left(\left(\frac{f\delta}{\mathsf{T}_{\rm env}} \right)^2\right)
\end{eqnarray}  
where $\kappa_1$ and $\kappa_2$ are two prefactors due to the two saddle point integrals in Eq.~(\ref{jsss}).    The entropy production rate follows from Eq.~(\ref{eq:sdot}) and is given by 
\begin{equation}
\dot{s} =   \frac{(f\delta)^2}{\gamma \mathsf{T}_{\rm env}}\kappa_1 \kappa_2 e^{-\frac{E_{\rm b}}{\mathsf{T}_{\rm env}}} + O\left(\left(\frac{f\delta}{\mathsf{T}_{\rm env}} \right)^3\right). 
\end{equation}
Lastly, using Eq.~(\ref{eq:TMeanSdot})  we obtain  the Van't Hoff-Arrhenius law for the 
mean-first passage time 
\begin{equation}
 \langle T_X\rangle =  \frac{\ell}{\delta} \frac{\gamma }{f\delta}  \frac{1}{\kappa_1\kappa_2} e^{\frac{E_{\rm b}}{\mathsf{T}_{\rm env}}} \left(1 + O\left(\frac{f\delta}{\mathsf{T}_{\rm env}}\right)\right) +  O_\ell(1). \label{eq:TMeanXX}
\end{equation}

We discuss two relevant cases: 
\begin{itemize}
\item $u'_{\rm max}= u'_{\rm min} = 0$ and $x_{\rm max},x_{\rm min}\in(0,\delta)$:
\begin{equation}
\kappa_1\kappa_2 = \frac{\sqrt{-u''_{\rm min}u''_{\rm max}}}{2\pi \mathsf{T}_{\rm env}};
\end{equation}
\item  $u'_{\rm max}\neq 0$ and  $u'_{\rm min} \neq 0$:
\begin{equation}
\kappa_1\kappa_2 =  \Bigg|\left(\frac{1}{u^+_{\rm max}}-\frac{1}{u^-_{\rm max}}\right)^{-1} \left(\frac{1}{u^+_{\rm min}}-\frac{1}{u^-_{\rm min}}\right)^{-1}  \Bigg| \frac{1}{ \mathsf{T}^2_{\rm env}}.
\end{equation}

 \end{itemize}

\section{Mean first-passage time for an overdamped Brownian particle   in a periodic potential that is triangular and in a uniform force field} \label{eq:Period2} 
We derive a number of explicit formulas that have been used to generate the  curves in the Figs.~\ref{fig1Mx}-\ref{fig:Meantime2}. Similar to the previous appendix, we consider  a  Brownian motion in a uniform force field $f$ and a periodic potential $u$, for which dynamics of the position variable $X$ is described by  the overdamped Langevin Eq.~(\ref{eq:nonequilb}).   However, in this appendix we specify the potential of the process, viz., we consider the triangular potential given by Eq.~(\ref{eq:modelU}), which allows us to derive explicit results.

      \subsection{Stationary distribution}
    For the triangular potential Eq.~(\ref{eq:modelU}),    the stationary probability distribution given by Eq.~(\ref{eq:pss}) reads
\cite{pigolotti2017generic}
\begin{equation}
p_{\rm ss}(x)  = \left\{ \begin{array}{ccc} a_1 + a_2 e^{\frac{x f_+  }{\mathsf{T}_{\rm env}}} &{\rm if}&  x\in [0,x^\ast],\\  a_3 + a_4 e^{ \frac{x f_-}{ \mathsf{T}_{\rm env}}}  &{\rm if}& x\in [x^\ast,\delta],\end{array}\right. \label{eq:pssT}
\end{equation} 
 where 
\begin{equation}
 f_+ = f  - \frac{u_0}{x^\ast} ,  \quad {\rm and} \quad  f_- = f  +  \frac{u_0}{\delta-x^\ast},
\end{equation}
and 
\begin{eqnarray}
a_1 &=& f_+f^2_-\frac{ e^{\frac{f_-x^\ast}{\mathsf{T}_{\rm env}}} -e^{\frac{f_-\delta + f_+x^\ast}{\mathsf{T}_{\rm env}}} }{\mathcal{N}}  , \label{eq:a1}\\ 
a_2 &=& f_+f_-(f_--f_+) \frac{e^{\frac{f_-\delta}{\mathsf{T}_{\rm env}}} - e^{\frac{f_-x^\ast}{\mathsf{T}_{\rm env}}}}{\mathcal{N}}, \\ 
a_3 &=&f^2_+f_-  \frac{ e^{\frac{f_-x^\ast}{\mathsf{T}_{\rm env}}} -e^{\frac{f_-\delta + f_+x^\ast}{\mathsf{T}_{\rm env}}} }{\mathcal{N}}  , \\ 
a_4 &=&f_+f_-(f_+-f_-)   \frac{e^{\frac{f_+x^\ast}{\mathsf{T}_{\rm env}}}-1}{\mathcal{N}},
\end{eqnarray} 
and where  the normalisation constant
\begin{eqnarray}
\lefteqn{\mathcal{N} = \mathsf{T}_{\rm env} (f_+-f_-)^2\left(e^{\frac{f_+x^\ast}{\mathsf{T}_{\rm env}}}-1\right)\left(e^{\frac{f_-\delta}{\mathsf{T}_{\rm env}}} - e^{\frac{f_-x^\ast}{\mathsf{T}_{\rm env}}}\right)} && \nonumber\\
 && + f_+f_- (f_+ \delta- f_+x^\ast + f_-x^\ast) \left(e^{\frac{f_-x^\ast}{\mathsf{T}_{\rm env}}} -e^{\frac{f_-\delta + f_+x^\ast}{\mathsf{T}_{\rm env}}} \right). \label{eq:N}
\end{eqnarray}

The stationary current is given by  the expression
\begin{equation}
 j_{\rm ss} = \frac{f_+a_1}{\gamma} = \frac{f_-a_3}{\gamma}. \label{eq:j}
\end{equation} 

In Fig.~\ref{fig2}, we plot the stationary distribution $p_{\rm ss}$ for various values of  the nonequilibrium driving $f\delta/\mathsf{T}_{\rm env}$.   Observe that the distribution concentrates around the values $x\approx0$ or $x\approx\delta$, and thus the process resembles a hopping process, as is also visible in  Fig.~\ref{fig1Mx}.

\subsection{Mean first-passage time} 
For the case of a triangular potential, we evaluate explicitly the integrals in Eqs.~(\ref{eq:Integral1}), (\ref{eq:Integral2}), (\ref{eq:Integral3}), and (\ref{eq:Integral4})  leading to  an explicit expression for the mean first-passage time $\langle T_X\rangle$ in Eq.~(\ref{eq:Tmean}).  In particular, we obtain  explicit expressions for the following integrals:   
\begin{equation}
 \int^{z}_{0} {\rm d}x\: w(x)= \left\{ \begin{array}{ccc}   \frac{\mathsf{T}_{\rm env}}{f_+} \left( e^{\frac{f_+z}{\mathsf{T}_{\rm env}}} -1\right)  &{\rm if}&  z<x^\ast,\\  \frac{\mathsf{T}_{\rm env}}{f_+} \left( e^{\frac{f_+x^\ast}{\mathsf{T}_{\rm env}}} -1\right)  + \frac{\mathsf{T}_{\rm env}}{f_-} e^{-\frac{u_0}{\mathsf{T}_{\rm env}}\frac{\delta}{\delta -x^\ast}}  \left( e^{\frac{f_-z}{\mathsf{T}_{\rm env}}} - e^{\frac{f_-x^\ast}{\mathsf{T}_{\rm env}}} \right)&{\rm if}& z>x^\ast,\end{array}\right. \label{eq:Int1}
\end{equation} 
 \begin{equation}
 \int^{z}_{0}\frac{{\rm d}x}{w(x)}= \left\{ \begin{array}{ccc}   \frac{\mathsf{T}_{\rm env}}{f_+} \left(1 - e^{\frac{-f_+z}{\mathsf{T}_{\rm env}}} \right)  &{\rm if}&  z<x^\ast,\\   \frac{\mathsf{T}_{\rm env}}{f_+} \left(1 - e^{\frac{-f_+x^\ast}{\mathsf{T}_{\rm env}}} \right) + \frac{\mathsf{T}_{\rm env}}{f_-} e^{\frac{u_0}{\mathsf{T}_{\rm env}}\frac{\delta}{\delta -x^\ast}}  \left( e^{\frac{-f_-x^\ast}{\mathsf{T}_{\rm env}}} - e^{\frac{-f_-z}{\mathsf{T}_{\rm env}}} \right)&{\rm if}& z>x^\ast, \end{array}\right.\label{eq:Int2}
\end{equation} 
and
 \begin{eqnarray}
\lefteqn{ \int^{\delta}_{\delta-z}\frac{{\rm d}x}{w(x)}} && \nonumber\\ 
&& = \left\{ \begin{array}{ccc}   \frac{\mathsf{T}_{\rm env}}{f_+} \left(e^{\frac{-f_+(\delta-z)}{\mathsf{T}_{\rm env}}} - e^{\frac{-f_+x^\ast}{\mathsf{T}_{\rm env}}} \right)  + \frac{\mathsf{T}_{\rm env}}{f_-} e^{\frac{u_0}{\mathsf{T}_{\rm env}}\frac{\delta}{\delta -x^\ast}}  \left( e^{\frac{-f_-x^\ast}{\mathsf{T}_{\rm env}}} - e^{\frac{-f_-\delta}{\mathsf{T}_{\rm env}}} \right)&{\rm if}&  \delta-z<x^\ast,\\   \frac{\mathsf{T}_{\rm env}}{f_-} e^{\frac{u_0}{\mathsf{T}_{\rm env}}\frac{\delta}{\delta -x^\ast}}  \left( e^{\frac{-f_- (\delta-z)}{\mathsf{T}_{\rm env}}} - e^{\frac{-f_-\delta}{\mathsf{T}_{\rm env}}} \right) &{\rm if}& \delta-z>x^\ast. \end{array}\right. \nonumber\\\label{eq:Int3}
\end{eqnarray} 
In addition, if  $z<x^\ast$, then 
\begin{equation}
\int^{z}_{0} {\rm d}y \frac{1}{w(y)} \int^{y}_{0} w(x) {\rm d}x = \frac{\mathsf{T}_{\rm env}}{f_+}   z -\left(\frac{\mathsf{T}_{\rm env}}{f_+}  \right)^2 \left(1-e^{-\frac{f_+z}{\mathsf{T}_{\rm env}}} \right) ,\label{eq:Int4}
\end{equation}
and if $z>x^\ast$, then 
\begin{eqnarray}
\lefteqn{\int^{z}_{0} {\rm d}y \frac{1}{w(y)} \int^{y}_{0} w(x) {\rm d}x} && 
\nonumber\\ 
 &=&\frac{\mathsf{T}_{\rm env}}{f_+}   x^\ast  +  \frac{\mathsf{T}_{\rm env}}{f_-}   (z-x^\ast) -\left(\frac{\mathsf{T}_{\rm env}}{f_+}  \right)^2 \left(1-e^{-\frac{f_+x^\ast}{\mathsf{T}_{\rm env}}} \right)-\left(\frac{\mathsf{T}_{\rm env}}{f_-} \right)^2   \left(1- e^{\frac{f_-(x^\ast-z)}{\mathsf{T}_{\rm env}}} \right) 
\nonumber\\ 
&&+  \frac{\mathsf{T}_{\rm env}}{f_-} e^{\frac{u_0}{\mathsf{T}_{\rm env}}\frac{\delta}{\delta -x^\ast}}  \left( e^{-\frac{f_- x^\ast}{\mathsf{T}_{\rm env}}} - e^{-\frac{f_-z}{\mathsf{T}_{\rm env}}} \right)  \frac{\mathsf{T}_{\rm env}}{f_+} \left( e^{\frac{f_+x^\ast}{\mathsf{T}_{\rm env}}}-1 \right). \label{eq:Int5}
\end{eqnarray}
Lastly, it holds that 
\begin{equation}
 \int^{\delta}_0 {\rm d}y\frac{1}{w(y)} \int^{\delta}_{y} {\rm d}xw(x)  =   \int^{\delta}_0 {\rm d}y\frac{1}{w(y)} \int^{\delta}_{0} {\rm d}xw(x)  -   \int^{\delta}_0 {\rm d}y\frac{1}{w(y)} \int^{y}_{0} {\rm d}xw(x)\label{eq:Int6}
 \end{equation}
and 
\begin{eqnarray}
 \int^{\delta}_{\delta-z} {\rm d}y\frac{1}{w(y)}  \int^{\delta}_{y} {\rm d}xw(x)    
 &=&   \int^{\delta}_0 {\rm d}y\frac{1}{w(y)} \int^{\delta}_{0} {\rm d}xw(x) -   \int^{\delta}_0 {\rm d}y\frac{1}{w(y)} \int^{y}_{0} {\rm d}xw(x)   
 \nonumber\\ 
 && - \int^{\delta-z}_0{\rm d}y\frac{1}{w(y)}  \int^{\delta}_{0} {\rm d}xw(x) + \int^{\delta-z}_0{\rm d}y\frac{1}{w(y)} \int^y_0 {\rm d}xw(x). \nonumber \\\label{eq:Int7}
 \end{eqnarray} 
 
 Substituting the above integrals, given by Eqs.~(\ref{eq:Int1})-(\ref{eq:Int7}),  into Eqs.~(\ref{eq:Integral1}), (\ref{eq:Integral2}), (\ref{eq:Integral3}), and (\ref{eq:Integral4}), and consequently using these in  Eq.~(\ref{eq:Tmean}) for  $\langle T_X\rangle$, we obtain a closed form expression for  $\langle T_X\rangle$.

 In  the Figs.~\ref{fig:Meantime} and \ref{fig:Meantime2} of the main text we have used this closed form expression of $\langle T_X\rangle$   to plot   $\langle T\rangle \dot{s}/|\log p_-|$ as a function of $\ell$ or $\langle T_X\rangle$ as a function of $\mathsf{T}_{\rm env}$.  
 
\subsection{Recovering the Van't Hoff-Arrhenius law}
The Eq.~(\ref{eq:TMeanXX})  in the particular case of a triangular potential leads to
 \begin{equation}
 \langle T_X\rangle =  \frac{\ell \gamma}{f}\frac{\mathsf{T}^2_{\rm env}}{u^2_0} e^{\frac{u_0}{\mathsf{T}_{\rm env}}} \left(1 + O\left(\frac{f\delta}{\mathsf{T}_{\rm env}}\right)\right).
\end{equation}
The green dotted line in  Fig.~\ref{fig:Meantime2} of the main text plots this equation.

\begin{figure}[t!]
\centering
\includegraphics[width=0.5\textwidth]{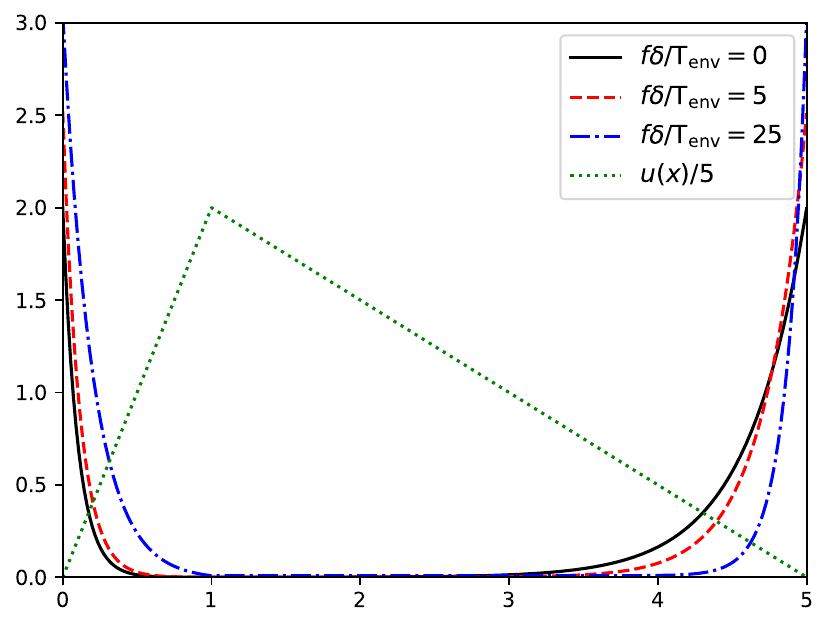}
 \put(-230,80){\Large$p_{\rm ss}$}
  \put(-110,-5){\Large$x$} 
\caption{Stationary distribution $p_{\rm ss}$ for the overdamped Langevin process (\ref{eq:nonequilb}) with triangular potential (\ref{eq:modelU})  as a function of $x$   for $\delta=5$, $x^\ast=1$,  $u_0=10$,  $\mathsf{T}_{\rm env}=1$ and for given values of $f$.  The value of $\gamma$ is immaterial.   Solid lines are plots of $p_{\rm ss}$ obtained from the Eqs.~(\ref{eq:pssT})-(\ref{eq:N}).   The green dotted line  plots the potential $u$ divided by $5$. } \label{fig2}
\end{figure}

  \section{Biased hopping process}   \label{App:hopp}

We determine the splitting probabilities and the moments of the first-passage time $T_X$, defined in Eq.~(\ref{eq:def3}), of the biased hopping process $X$, determined by Eq.~(\ref{eq:randomwalkX}).   We  make use of the decision variable
  \begin{equation}
  D =\left\{ \begin{array}{ccc}1 &{\rm if}&  X(T_X)\geq \ell_+,\\ -1&{\rm if}&  X(T_X)\leq -\ell_- . \end{array}\right.
  \end{equation}

\subsection{Martingales in the biased hopping processes}

The processes
 \begin{equation}
 Z(t) = e^{zX(t) + [ (1-e^{z})k_+ + (1-e^{-z})k_-
]t } 
 \end{equation}
are martingales for all values of $z\in\mathbb{R}$ (see Appendix~\ref{App:A1} for the definition of a martingale).    Indeed, using It\^{o}'s formula for jump processes \cite{protter2005stochastic}, we obtain
   \begin{equation}
  {\rm d}Z(t) =   (e^{z}-1)Z(t)  \left[ {\rm d}N_+(t)  - k_+ {\rm d}t\right] + (e^{-z}-1)Z(t)  \left[  {\rm d}N_-(t) - k_- {\rm d}t \right],
  \end{equation}   
  which is a martingale process as both ${\rm d}N_+(t)  - k_+ {\rm d}t$ and ${\rm d}N_-(t) - k_- {\rm d}t$ are martingales.    In the special case of $z = \log \frac{k_-}{k_+}$, we obtain that $Z(t) = e^{-S(t)}$ is the exponentiated negative entropy production, which is an example of martingale process \cite{neri2019integral}.

\begin{proposition}[A martingale equality]\label{propMartEq}
If  $k_+>k_-$, then for all $z\in \mathbb{R}\setminus [\log \frac{k_-}{k_+}, 0]$ it holds that  
   \begin{eqnarray}
1  = \Big \langle 1_{T_X<\infty}1_{D=1} e^{z\lceil \ell_+\rceil + T_X f(z)}   +  1_{T_X<\infty}1_{D=-1} e^{-z\lceil \ell_-\rceil +T_X f(z)}   \Big\rangle,  \label{eq:important2xJump}
\end{eqnarray}
where 
\begin{equation}
f(z) = (1-e^{z})k_+ + (1-e^{-z})k_-.
\end{equation}
\end{proposition} 
\begin{proof}
Since   $Z(t)$ is a martingale,  we can apply Theorem~\ref{TheDoob} to $Z(t\wedge T_X)$ yielding
\begin{eqnarray}
1 = \langle Z(t\wedge T_X)\rangle =   \Big \langle e^{zX(t\wedge T_X) + (t\wedge T_X)f(z)  }  \Big\rangle .
\end{eqnarray}   
Since  for $z\in \mathbb{R}\setminus [\log \frac{k_-}{k_+}, 0]$ it holds that  $f(z)<0$, the  process  $Z(t\wedge T_X)$ is bounded from above, viz.,
\begin{eqnarray}
e^{zX(t\wedge T_X) +(t\wedge T_X) f(z)  }  < {\rm max}\left\{e^{z \lceil \ell_+\rceil},e^{-z\lceil \ell_-\rceil}\right\}
\end{eqnarray} 
 Hence, the bounded convergence theorem applies, see e.g.~Ref.~\cite{tao2011introduction}, and therefore the  limit $t\rightarrow \infty$ can be taken under the expectation value yielding
\begin{eqnarray}
1 &=&  \Big\langle \lim_{t\rightarrow \infty} e^{zX(t\wedge T_X) +  (t\wedge T_X) f(z) } \Big\rangle \\ 
&=&\Big \langle 1_{T_X<\infty}1_{D=1} e^{z\lceil \ell_+\rceil + T_X f(z) }   +  1_{T_X<\infty}1_{D=-1} e^{-z\lceil \ell_-\rceil +T_X f(z) }   \Big\rangle, \nonumber\\
\end{eqnarray}    
which completes the proof of the equality (\ref{eq:important2xJump}).
\end{proof}

In what follows, we repeatedly use the martingale equality of Proposition~\ref{propMartEq} to derive various  properties $T_X$.
\subsection{ The first-passage time $T_X$ is with probability one finite}

\begin{proposition}\label{propTJump} It holds that $T_X$ is almost surely finite, i.e., 
\begin{eqnarray}
p_-+p_+ = 1. \label{eq:finiteJump}
\end{eqnarray} 
\end{proposition}

\begin{proof}
We take the  limit $z\rightarrow 0$  in Eq.~(\ref{eq:important2xJump}).    Since for $z\in[0,1]$ the  argument in the expectation value is bounded by $e^{\ell_+}$, the     bounded convergence theorem applies, see e.g.~Ref.~\cite{tao2011introduction}, and 
  \begin{eqnarray}
1  &=&  \lim_{z\rightarrow 0}  \Big\langle  1_{T_X<\infty}1_{D=1} e^{z\lceil \ell_+\rceil + T_X f(z) }   +  1_{T_X<\infty}1_{D=-1} e^{-z\lceil \ell_-\rceil +T_X f(z) }  \Big\rangle  \nonumber\\  
 &=& \langle  1_{T_X<\infty}1_{D=1}+  1_{T_X<\infty}1_{D=-1} \rangle =  \langle 1_{T_X<\infty} \rangle = \mathbb{P}[T_X<\infty], \nonumber
\end{eqnarray}
where we have used that $f(0) = 0$.
\end{proof}

 \subsection{Splitting probabilities} \label{App:SplitE}
 
  \begin{proposition}  \label{prop:Aa}
  It holds  that 
  \begin{eqnarray}
p_+ = \frac{1 -e^{-\lceil \ell_-\rceil \log\frac{k_+}{k_-}} }{ 1- e^{-(\lceil \ell_+\rceil + \lceil \ell_-\rceil) \log\frac{k_+}{k_-}}}   ,\quad  {\rm and}\quad 
 p_- =e^{-\lceil \ell_-\rceil \log\frac{k_+}{k_-}}  \frac{1 -e^{-\lceil \ell_+\rceil \log\frac{k_+}{k_-}} }{ 1- e^{-(\lceil \ell_+\rceil + \lceil \ell_-\rceil) \log\frac{k_+}{k_-}}} \label{eq:P-1App}. 
\end{eqnarray}   

\end{proposition} 
\begin{proof}
We apply Theorem~\ref{TheDoob} to the martingale
\begin{equation}
e^{-S(t)} =   e^{X(t)\log\frac{k_-}{k_+} },
\end{equation}
yielding, for every $t\geq 0$,
\begin{eqnarray}
\Big \langle e^{X(t\wedge T_X)\log\frac{k_-}{k_+} } \Big\rangle = 1.
\end{eqnarray}
Up to the stopping time the process is confined to the interval
$X(t\wedge T_X)\in[-\lceil \ell_-\rceil,\lceil \ell_+\rceil]$, and therefore the random variable
$e^{X(t\wedge T_X)\log\frac{k_-}{k_+}}$ is bounded, uniformly in $t$, by
$e^{-\lceil \ell_-\rceil\log\frac{k_-}{k_+}}$ (where we have used that $k_-<k_+$).   Moreover, $T_X<\infty$ almost surely by
Proposition~\ref{propTJump}, and hence $X(t\wedge T_X)\rightarrow X(T_X)$
almost surely as $t\rightarrow\infty$.   The bounded convergence theorem,
see e.g.~Ref.~\cite{tao2011introduction}, thus permits taking the limit
under the expectation value,
\begin{eqnarray}
1 = \lim_{t\rightarrow\infty}\Big \langle e^{X(t\wedge T_X)\log\frac{k_-}{k_+} } \Big\rangle
= \Big \langle e^{X(T_X)\log\frac{k_-}{k_+} } \Big\rangle .
\end{eqnarray}
Since $X(T_X) = \lceil \ell_+\rceil$ when $D=1$ and $X(T_X) = -\lceil \ell_-\rceil$ when $D=-1$,
we obtain
\begin{eqnarray}
 p_+ e^{\lceil \ell_+\rceil \log\frac{k_-}{k_+}} + p_- e^{-\lceil \ell_-\rceil \log\frac{k_-}{k_+}}  = 1. \label{eq:xx}
\end{eqnarray}
The solutions to the Eqs.~(\ref{eq:finiteJump}) and (\ref{eq:xx}) are given by
Eqs.~(\ref{eq:P-1App}), which completes the proof.
\end{proof}

Using  $b = k_-/k_+$ in Eq.~(\ref{eq:P-1App}), we obtain the Eq.~(\ref{eq:pMP}) in the main text.

\subsection{Generating function}\label{App:Gen}
We derive the explicit formula given by  Eqs.~(\ref{eq:genFinalJump})-(\ref{eq:betay}) for the generating function $g(y)$,  as defined   in Eq.~(\ref{eq:gJump}), for $y>0$.

The generating function $g(y)$ can be written as 
\begin{equation}
g(y) = p_+ g_+(y) + p_- g_-(y)
\end{equation}
where $g_+$ and $g_-$ are the conditional generating functions 
\begin{equation}
g_+(y) =  \langle e^{-y T_X(k_-+k_+)} |D = 1 \rangle, \quad  {\rm and} \quad g_-(y) =  \langle e^{-y T_X(k_-+k_+)} |D = -1 \rangle.
\end{equation}

\begin{lemma}  
It holds that 
 \begin{eqnarray}
1  =  \left(\frac{1}{2} \left[(1+b)(1+ y)  + \sqrt{-4b + (1+b)^2 (1+y)^2}\right]\right)^{\lceil \ell_+\rceil}   p_+   g_+(y)   
\nonumber\\ 
+   \left(\frac{1}{2} \left[(1+b)(1+ y)  + \sqrt{-4b + (1+b)^2 (1+y)^2}\right]\right)^{-\lceil \ell_-\rceil}p_-  g_-(y),  \label{eq:important2xJumpTwoa}
\end{eqnarray}
and 
 \begin{eqnarray}
1  =  \left(\frac{1}{2} \left[(1+b)(1+ y)  - \sqrt{-4b + (1+b)^2 (1+y)^2}\right]\right)^{\lceil \ell_+\rceil}   p_+   g_+(y)   
\nonumber\\ 
+   \left(\frac{1}{2} \left[(1+b)(1+ y)  - \sqrt{-4b + (1+b)^2 (1+y)^2}\right]\right)^{-\lceil \ell_-\rceil}p_-  g_-(y).\label{eq:important2xJumpTwox}
\end{eqnarray}

\end{lemma} 
\begin{proof}
We rewrite the relation  (\ref{eq:important2xJump})     for $z\notin [\log \frac{k_-}{k_+},0]$ as 
 \begin{eqnarray}
1  = e^{z\lceil \ell_+\rceil}   p_+  \langle e^{f(z)T (k_-+k_+)} |D = 1\rangle   +  e^{-z\lceil \ell_-\rceil} p_-  \langle  e^{f(z)T (k_-+k_+)}  |D=-1\rangle.  \label{eq:important2xJumpTwo}
\end{eqnarray}
Setting 
\begin{equation}
y = -f(z) \label{eq:yJump}
\end{equation} 
and solving towards $z$, we obtain  two solutions.  

First, let us consider the   solution branch for $z\geq 0$, which is given by 
\begin{equation}
 z = \log \left(\frac{1}{2} \left[(1+b)(1+ y)  + \sqrt{-4b + (1+b)^2 (1+y)^2}\right]\right). \label{eq:zJumpJump}
\end{equation} 
Using Eqs.~(\ref{eq:yJump}) and (\ref{eq:zJumpJump}) in (\ref{eq:important2xJumpTwo}), we obtain Eq.~(\ref{eq:important2xJumpTwoa}).

Second, let us  consider the  solution branch for $z\leq \log b$, namely, 
\begin{equation}
 z = \log \left(\frac{1}{2} \left[(1+b)(1+ y)  - \sqrt{-4b + (1+b)^2 (1+y)^2}\right]\right).\label{eq:zJumpJumpx}
\end{equation} 
In this case, using Eqs.~(\ref{eq:yJump}) and (\ref{eq:zJumpJumpx}) in (\ref{eq:important2xJumpTwo}), we obtain the Eq.~(\ref{eq:important2xJumpTwox}).

\end{proof}

\begin{proposition} 
The generating function Eq.~(\ref{eq:gJump})  is given by Eqs.~(\ref{eq:genFinalJump})-(\ref{eq:betay}).
  \end{proposition}
  \begin{proof}  
  We find Eq.~(\ref{eq:genFinalJump}) readily by solving the Eqs.~(\ref{eq:important2xJumpTwoa})-(\ref{eq:important2xJumpTwox}).
  \end{proof}
  
  \subsection{Moments of the first-passage times $T_X$} 
  The moments of first passage times follow  from taking explicitly the  derivatives in Eq.~(\ref{eq:TMom}).
  
 The first moment is given by 
    \begin{eqnarray}
    \langle  T_X \rangle =   \frac{\lceil \ell_+\rceil p_+ - \lceil \ell_-\rceil p_-}{k_+-k_-}. \label{eq:TX}
\end{eqnarray} 

The second moment is given by 
   \begin{eqnarray}
(k_+-k_-)^2  \langle  T^2_X \rangle &=&   \frac{p_+}{1-b^{\lceil \ell_-\rceil+\lceil \ell_+\rceil}}  \left(\lceil \ell_+\rceil^2 + \lceil \ell_+\rceil  \tanh^{-1}\left(\frac{a}{2\mathsf{T}_{\rm env}} \right) \right)  - \lceil \ell_-\rceil^2 p_- \left(\frac{3 + b^{\lceil \ell_-\rceil+\lceil \ell_+\rceil}}{1-b^{\lceil \ell_-\rceil+\lceil \ell_+\rceil}}\right) \nonumber\\ 
    && +  \frac{p_+ b^{\lceil \ell_-\rceil+\lceil \ell_+\rceil}}{1-b^{\lceil \ell_-\rceil+\lceil \ell_+\rceil}}  \left(3 \lceil \ell_+\rceil^2 - \lceil \ell_+\rceil  \tanh^{-1}\left(\frac{a}{2\mathsf{T}_{\rm env}} \right)\right) \nonumber\\
    && 
    +\lceil \ell_-\rceil \tanh^{-1}\left(\frac{a}{2\mathsf{T}_{\rm env}} \right)\frac{b^{2\lceil \ell_-\rceil+\lceil \ell_+\rceil}(1-b^{\lceil \ell_+\rceil})}{(1-b^{\lceil \ell_-\rceil+\lceil \ell_+\rceil})^2}  - 4 \lceil \ell_+\rceil\lceil \ell_-\rceil  \frac{b^{2\lceil \ell_-\rceil+\lceil \ell_+\rceil}}{(1-b^{\lceil \ell_-\rceil+\lceil \ell_+\rceil})^2}  \nonumber\\
      &&  + \left(\lceil \ell_-\rceil \tanh^{-1}\left(\frac{a}{2\mathsf{T}_{\rm env}} \right)+8\lceil \ell_-\rceil\lceil \ell_+\rceil\right) \frac{b^{\lceil \ell_-\rceil}b^{\lceil \ell_+\rceil}}{(1-b^{\lceil \ell_-\rceil}b^{\lceil \ell_+\rceil})^2} \nonumber\\ 
  &&     -  \lceil \ell_-\rceil \left(\tanh^{-1}\left(\frac{a}{2\mathsf{T}_{\rm env}} \right)+4 \lceil \ell_+\rceil\right)\frac{b^{\lceil \ell_-\rceil}}{(1-b^{\lceil \ell_-\rceil}b^{\lceil \ell_+\rceil})^2}, \label{eq:TX2}
\end{eqnarray} 
where we have used the notation $\tanh^{-1}\left(\frac{a}{2\mathsf{T}_{\rm env}} \right) = 1/\tanh\left(\frac{a}{2\mathsf{T}_{\rm env}} \right)$.

We avoid writing down the expression for  $\langle T^3_X\rangle$, as it is  even lengthier    than $\langle T^2_X\rangle$.

  \subsection{Moments of the first-passage times $T_X$ in the case of symmetric thresholds}\label{App:E6} 
  We derive the formulae used to     plot the lines in the Fig.~\ref{fig2M} of the main text.
  
In the specific case where $\ell_+ = \ell_- = \ell$, we obtain the  simpler expression 
    \begin{eqnarray}
g(y) = \frac{ 2^{\lceil \ell\rceil} +2^{-\lceil \ell\rceil} \left(\beta(y) - \sqrt{-4 \frac{k_-}{k_+} +\beta^2(y)}\right)^{\lceil \ell\rceil}  \left(\beta(y)+ \sqrt{-4 \frac{k_-}{k_+}  +\beta^2(y)}\right)^{\lceil \ell\rceil}  }{\left(\beta(y) - \sqrt{-4 \frac{k_-}{k_+}  +\beta^2(y)}\right)^{\lceil \ell\rceil} + \left(\beta(y)+ \sqrt{-4  \frac{k_-}{k_+} + \beta^2(y)}\right)^{\lceil \ell\rceil} }\label{eq:genFinalJump2}
\end{eqnarray} 
for the generating function, where $\beta(y)$ is defined in Eq.~(\ref{eq:betay}).

In this case, the mean first-passage time is given by
    \begin{eqnarray}
    \langle T_X \rangle =  \frac{ \lceil \ell\rceil}{k_+-k_-}\frac{1-b^{\lceil \ell\rceil}}{1+b^{\lceil \ell\rceil}}. \label{eq:Tx1ab}
\end{eqnarray} 
its  second moment
  \begin{eqnarray}
    \langle T^2_X \rangle =  \lceil \ell\rceil \frac{\lceil \ell\rceil   + \frac{ k_+ +  k_- }{ k_+-k_-}  - 6 \lceil \ell\rceil b^{\lceil \ell\rceil}  + b^{2\lceil \ell\rceil} \left( \lceil \ell\rceil -\frac{ k_+ +  k_- }{ k_+-k_-}  \right)  }{\left(k_+-k_-\right)^2 \left(1 +b^{\lceil \ell\rceil} \right)^2} \label{eq:Tx2ab}
     \end{eqnarray}  
     and its third moment 
       \begin{eqnarray}
    \langle T^3_X \rangle &=& \frac{\lceil \ell\rceil}{k^3_+(1-b)^5(1 + b^{\lceil \ell\rceil})^3} \Big\{ 2 + 8 b + 2 b^2 + 3\lceil \ell\rceil (1-b^2) + \lceil \ell\rceil^2(1-b)^2   \nonumber\\ 
  &&  \left.  + b^{\lceil \ell\rceil} (2 + 2 b (4 + b) + 15 (-1 + b^2) \lceil \ell\rceil - 23 (-1 + b)^2 \lceil \ell\rceil^2)  \right. \nonumber\\ 
  && +    b^{2 \lceil \ell\rceil} (-2 - 2 b (4 + b) + 15 (-1 + b^2) \lceil \ell\rceil + 23 (-1 + b)^2 \lceil \ell\rceil^2) \nonumber\\ 
  &&  +  b^{3 \lceil \ell\rceil}  (2 + 2 b (4 + b) + 3 (-1 + b^2)\lceil \ell\rceil + (-1 + b)^2 \lceil \ell\rceil^2) \Big\}.\label{eq:Tx3ab}
          \end{eqnarray} 
Formulae (\ref{eq:Tx1ab})-(\ref{eq:Tx3ab}) are plotted in Fig.~\ref{fig2M} of the main text.
     
One readily verifies the thermodynamic uncertainty relation
    \begin{eqnarray}
   \lim_{\lceil \ell\rceil\rightarrow\infty}\frac{   \langle T^2_X \rangle  -     \langle T_X \rangle^2}{  \langle T_X \rangle } =  \frac{k_++k_-}{ (k_+-k_-)^2} \geq \frac{2}{ (k_+-k_-)\log\frac{k_+}{k_-} }
    \end{eqnarray} 
    where we used the fact that $\log(x) \geq 2\frac{x-1}{x+1}$ with $x = k_+/k_-$.

 \subsection{Asymptotic expressions for  large thresholds}\label{App:E8}
We determine the splitting probabilities and the first two  moments of $T_X$ in  the limit $\ell_{+},\ell_{-}\gg 1$ with the ratio $\ell_+/\ell_-$ fixed to a constant value.  In particular, we derive the Eq.~(\ref{eq:T1Asympt}) and the Eq.~(\ref{eq:T2Asympt}) for the specific cases of $n=1$ and $n=2$.

From Eqs.~(\ref{eq:pMP}), we obtain for the splitting probabilities that
\begin{equation}
p_-  =  b^{\lceil \ell_-\rceil} +O(b^{\lceil \ell_+\rceil+\lceil \ell_-\rceil}), \quad {\rm and} \quad p_+ =   1+O(b^{\lceil \ell_-\rceil}). \label{eq:pAsympt}
\end{equation}

Equation (\ref{eq:TX}) implies that the mean first-passage time  
\begin{equation}
\langle T_X\rangle = \frac{\lceil \ell_+\rceil}{k_+-k_-} \left(1 +O(b^{\lceil \ell_-\rceil}) \right),\label{eq:TAsympt}
\end{equation}  
and from Eq.~(\ref{eq:TX2}) it follows that the second moment 
\begin{equation}
\langle T^2_X\rangle = \frac{ \lceil \ell_+\rceil^2}{(k_+-k_-)^2}  \left(1 + \frac{1}{\lceil \ell_+\rceil\tanh\left(\frac{a}{2\mathsf{T}_{\rm env}} \right)} +O(b^{\lceil \ell_-\rceil})   \right). \label{eq:TAsympt2}
\end{equation}  

The Eqs.~(\ref{eq:pAsympt}) and (\ref{eq:TAsympt})  imply that  
\begin{equation}
\frac{\lceil \ell_+\rceil}{\lceil \ell_-\rceil}\frac{|\log p_-|  }{\langle T_X \rangle} = \frac{a}{\mathsf{T}_{\rm env}} (k_+-k_-) (1+O(b^{\lceil \ell_-\rceil})).
\end{equation}  
We recognise in the above formula the entropy production rate $\dot{s}$ given by Eq.~(\ref{eq:sigmaxM}),
and thus
\begin{equation}
\frac{\lceil \ell_+\rceil}{\lceil \ell_-\rceil}\frac{|\log p_-|  }{\langle T_X \rangle} = \dot{s}  +O(b^{\lceil \ell_-\rceil}).
\end{equation}   
Analogously, Eqs.~(\ref{eq:pAsympt}) and (\ref{eq:TAsympt2})  imply that  
\begin{equation}
\frac{\lceil \ell_+\rceil}{\lceil \ell_-\rceil}\frac{|\log p_-|  }{\sqrt{\langle T^2_X \rangle}} = \dot{s} + O\left(\frac{1}{\lceil \ell_+\rceil}\right). \label{eq:asymptT22}
\end{equation}  

The thermodynamic uncertainty relation is governed by the subleading $O\left(1/\lceil \ell_+\rceil\right)$ term in Eq.~(\ref{eq:asymptT22}).   Using  Eqs.~(\ref{eq:pAsympt}) and (\ref{eq:TAsympt2}), we obtain the Eq.~(\ref{eq:T2Unc}) in the main text.
Since, 
\begin{equation}
\frac{1}{\tanh(x/2)} \geq \frac{2}{x}
\end{equation}  
the thermodynamic uncertainty relation  \cite{gringich2017bis}
\begin{equation}
\frac{ 2\langle T_X\rangle}{\langle T^2_X\rangle-\langle T_X\rangle^2} \leq \dot{s}
\end{equation}  
holds.  

In order to find asymptotic expressions for the higher order moments, we analyze in the next subsection the probability distribution of $T_X$ in the limit of large thresholds $\ell_-$ and $\ell_+$.

\subsection{Probability distribution in the asymptotic limit $\ell_{\pm}\rightarrow \infty$} \label{App:E9}
In the present appendix, we derive  the asymptotic formula (\ref{eq:T2Asympt}) for the moments of $T_X$.
In order to derive asymptotic expressions for the moments $\langle T^n_X\rangle$ with $n>2$, we determine the probability distribution in this limit.

Using that $\zeta_-<\zeta_+$, where $\zeta_{\pm}$ are defined in Eq.~(\ref{eq:zetapm}), we obtain  in the limit  $\ell_{\rm min}\rightarrow\infty$  for the generating function given by Eq.~(\ref{eq:genFinalJump}) the formula
  \begin{eqnarray}
g(y) =  \left(\frac{2}{\zeta _+(y)}\right)^{\lceil \ell_+\rceil} \left(1 +O\left(\left(\frac{\zeta _-(y)}{ \zeta _+(y) }\right)^{\lceil \ell_-\rceil}\right)\right) 
  +  \left(\frac{\zeta _-(y)}{2}\right)^{\lceil \ell_-\rceil}   \left(1 +O\left(\left(\frac{\zeta _-(y)}{ \zeta _+(y) }\right)^{\lceil \ell_-\rceil}\right)\right),\nonumber\\  \label{eq:genFinalJumpxx}
\end{eqnarray}    
which can be further simplified into
  \begin{eqnarray}
g(y) =  \left(\frac{2}{\zeta _+(y)}\right)^{\lceil \ell_+\rceil}  + O(b^{\lceil \ell_-\rceil}).
\end{eqnarray}

Considering that $T$ will be large when  both $\lceil \ell_+\rceil$ and $\lceil \ell_-\rceil$ are large, we use that $y\sim \frac{1}{\lceil \ell_{\rm min}\rceil}$.   Therefore, 
\begin{equation}
\zeta_{+}(y) = 2 +  2 \frac{1+b}{1-b} y + O(y^2).
\end{equation} 
Taking the inverse Laplace transform, we obtain up to leading order
\begin{eqnarray}
p_{T_X}(t) =(k_++k_-)\frac{( (k_++k_-)  t)^{\lceil \ell_+\rceil-1}}{\Gamma(\lceil \ell_+\rceil)} \left( \frac{1-b}{1+b}\right)^{\lceil \ell_+\rceil}e^{-t (k_++k_-) \frac{1-b}{1+b} }  + O(b^{\lceil \ell_-\rceil}),
\end{eqnarray}
which is the Gamma distribution with shape parameter $\lceil \ell_+\rceil$ and rate $(k_++k_-)(1-b)/(1+b)$.

If we introduce a new variable,
 \begin{eqnarray}
\tau  =\frac{( k_++k_-)t}{\lceil \ell_+\rceil}, 
\end{eqnarray}
then we  get 
\begin{eqnarray}
p_{\frac{( k_-+k_+)T_X}{\lceil \ell_+\rceil}}(\tau) \sim    \exp\left(-\lceil \ell_+\rceil I(\tau)   + O_{\lceil \ell_+\rceil}(1)\right)  + O(b^{\lceil \ell_-\rceil})
\end{eqnarray}
with the  large deviation function  
\begin{equation}
I(\tau) = \frac{1-b}{1+b} \tau - \log \left(  \tau\right) - \log \frac{1-b}{1+b}-1.
\end{equation}
The minimum of $I$ is reached when
\begin{equation}
\tau^\ast =   \frac{1+b}{1-b} 
\end{equation}
in which case $I(\tau^\ast) = 0$.  Expanding $I(\tau)$ around $\tau^\ast$ we obtain 
\begin{equation}
I(\tau) = \frac{\left(  \tau- \frac{1+b}{1-b}\right)^2}{2 \left(\frac{1+b}{1-b}\right)^2 } + O((\tau-\tau^\ast)^3). \label{eq:I}
\end{equation}

Hence, the distribution of $T_X$ is
\begin{eqnarray}
p_{\frac{( k_++k_-)T_X}{\lceil \ell_+\rceil}}(\tau) =\sqrt{\frac{\lceil \ell_+\rceil}{2\pi (\tau^\ast)^2}}      \exp\left(-\lceil \ell_+\rceil\frac{(\tau-\tau^\ast)^2}{2  (\tau^\ast)^2} + O((\tau-\tau^\ast)^2)\right) + O(b^{\lceil \ell_-\rceil}). \label{eq:Txx}
\end{eqnarray}
For large $\lceil \ell_+\rceil$, the distribution  $p_{\frac{( k_++k_-)T_X}{\lceil \ell_+\rceil}}(\tau)$  is centered around $ \tau=\tau^\ast$, and therefore  $\frac{( k_++k_-)T_X}{\lceil \ell_+\rceil}$ is a deterministic variable in this limit.   The moments of $T$ are thus  given by
\begin{equation}
\langle T^n_X\rangle =\lceil \ell_+\rceil^n \frac{(\tau^\ast)^n}{(k_++k_-)^n} + O(\lceil \ell_+\rceil^{n-1})=  \frac{\lceil \ell_+\rceil^n}{(k_+-k_-)^n} + O(\lceil \ell_+\rceil^{n-1}). 
\end{equation} 
Using  the formula for $p_-$, given by Eq.~(\ref{eq:pAsympt}), and the expression for $\dot{s}$ in Eq.~(\ref{eq:sigmaxM}), we recover   
\begin{equation}
\frac{\lceil \ell_+\rceil}{\lceil \ell_-\rceil}\frac{|\log p_-|  }{\left(\langle T^n_X \rangle\right)^{1/n}} = \dot{s} + O\left(\frac{1}{\lceil \ell_+\rceil}\right),\label{eq:universal}
\end{equation}
which is also the formula (\ref{eq:T2Asympt}) that we were meant to derive.

Note that obtaining an explicit expression for the  $1/\lceil \ell_+\rceil$ correction terms in Eq.~(\ref{eq:T2Asympt})  is more complicated as we need to determine the subleading order terms in Eq.~(\ref{eq:I}), which depend on $b$ and are process dependent.     Hence, the moments $\langle T^n_X\rangle$ converge for large thresholds to the universal limit given by Eq.~(\ref{eq:universal}) since they are governed by the leading order term in the asymptotic behaviour of $T_X$.  On the other hand, the  Fano factor  
 \begin{equation}
\frac{ \langle T^2_X\rangle-\langle T_X\rangle^2}{\langle T_X\rangle}
\end{equation} 
characterising uncertainty 
depends on the subleading terms and will therefore not converge to a universal limit when the thresholds diverge.   This clarifies on an example why  the first-passage time relations in the present paper are equalities  for $J=S$, while  this is not the case for thermodynamic uncertainty relations.

  \section{Splitting probabilities of currents $J$ in Markov jump processes}   \label{derivP}   
  We derive the equality (\ref{eq:pp}) for the splitting probabilities of currents $J$ in Markov jump processes $X$ taking values in a  finite set $\mathcal{X}$ in the limit when the thresholds $-\ell_-$ and $\ell_+$ are large.  Equation~(\ref{eq:pp})  contains the splitting probability
  \begin{equation}
  p_- = \mathbb{P} \left[J(T_J(-\ell_-,\ell_+))<-\ell_-\right] 
  \end{equation}
  to hit the negative boundary first, and the splitting probability 
    \begin{equation}
  p^\dagger_+ = (\mathbb{P}\circ\Theta) \left[J(T_J(-\ell_-,\ell_+))>\ell_+\right] 
  \end{equation}
  to hit the positive boundary first in the time-reversed dynamics.   We have written explicitly the dependencies on the thresholds in $T_J(-\ell_-,\ell_+)$, as this will be relevant when considering time-reversal arguments.   In particular, since the   map $\Theta$ changes the sign of $J$, it holds that 
      \begin{equation}
  p^\dagger_+ = \mathbb{P} \left[J(T_J(-\ell_+,\ell_-))<-\ell_+\right]  .
  \end{equation}

\subsection{Splitting probabilities from a random walk process on the real line  }
To determine the splitting probabilities $p_-$ and $p^\dagger_+$, we define the discrete-time process
\begin{equation}
J_n = J(n\Delta t) \label{eq:discrete}
\end{equation}
with $n\in\mathbb{N}$.  When $\Delta t$ is large enough, then the increments $J_n-J_{n-1}$ are independent and identically distributed random variables drawn from a probability distribution $p_{\Delta_J}(\Delta_j)$, and for Markov jump processes the distribution  $p_{\Delta_J}$ is supported on a   discrete subset of the real line.     Hence, for large threshold values, the splitting probability $p_-$ of $J$ is equal to the splitting probability of the following random walk process defined on the real line, 
\begin{equation}
{\rm d}K = \sum_{j\in \mathbb{Z}} \Delta_j {\rm d}N_j,
\end{equation}
where ${\rm d}N_j$ is a Poisson process with   rate $k_j = p_{\Delta J}(j)$, where $\Delta_j = -\Delta_{-j}$, and $\Delta_0 = 0$.    We can thus write
 \begin{equation}
  p_- = \mathbb{P}\left[ K(T_K(-\ell_-,\ell_+))<-\ell_-\right]. 
  \end{equation} 
  Since, by definition, the current $J$ changes sign under time reversal, the time-reversed process is  
  \begin{equation}
{\rm d}K^\dagger = -\sum_{j\in \mathbb{Z}} \Delta_{j} {\rm d}N_{j}.
\end{equation}
  and 
     \begin{equation}
  p^\dagger_+ = \mathbb{P}\left[K(T_K(-\ell_+,\ell_-))<-\ell_+\right]  . \label{eq:PDaggerP}
  \end{equation}

In the remaining part of this appendix, we determine, using an approach similar to the one presented in Appendix~\ref{App:hopp},  the splitting probabilities $p_-$ and $p^\dagger_+$  of the  first-passage time $T_K$ in the random walk process $K$, which  for large threshold values $-\ell_-$ and $\ell_+$ are identical to those of $J$.     Consequently, we use the obtained expressions for the splitting probabilities to  demonstrate that the equality (\ref{eq:pp})  is valid in the limit of large thresholds.

  In the calculations we repeatedly use the decision variable
  \begin{equation}
  D_K = {\rm sign}\left(K(T_K)-K(0)\right).
  \end{equation}

\subsection{Martingales related to $K$}

The processes
 \begin{equation}
 Z(t) = e^{zK(t) +  tf(z)
 }  \label{eq:martingalezK}
 \end{equation}
 with 
 \begin{equation}
 f(z)=\sum_{j\in \mathbb{Z}} (1-e^{z\Delta_j})k_j 
 \end{equation}
are martingales for all values of $z\in\mathbb{R}$.    Indeed, applying It\^{o}'s formula for jump processes \cite{protter2005stochastic} to the latter equation, we obtain
   \begin{equation}
  {\rm d}Z(t) =  \sum_{j\in \mathbb{Z}} (e^{z \Delta_j}-1)Z(t)  \left[ {\rm d}N_j(t)  - k_j {\rm d}t\right] ,
 \end{equation}   
  which is a martingale process as the ${\rm d}N_j(t)  - k_j{\rm d}t$  are martingales.   \\

  We also define the time-reversed processes 
    \begin{equation}
  {\rm d}Z^\dagger(t) =  \sum_{j\in\mathbb{Z}} (e^{-z \Delta_j}-1)Z^\dagger(t)   \left[ {\rm d}N_j(t)  - k_j {\rm d}t\right]
  \end{equation}   
  that run backwards ($\Delta_j \rightarrow -\Delta_j$).   Note that the time-reversed process $Z^\dagger$ is related to $Z$  by $z\leftrightarrow -z$.

\begin{proposition}[A martingale equality] 
For all values $z\in \mathbb{R}$ for which $f(z) < 0$,
   \begin{eqnarray}
1  = \Big \langle 1_{T_K<\infty}1_{D_K=1} e^{z\ell_+ (1+o_{\ell_{\rm min}}(1)) + T_K f(z)}   +  1_{T_K<\infty}1_{D_K=-1} e^{-z\ell_-(1+o_{\ell_{\rm min}}(1))  +T_K f(z)}   \Big\rangle .  \label{eq:important2xJumpa}
\end{eqnarray}
\end{proposition} 
\begin{proof}
Since   $Z(t)$ is a martingale,  we can apply Theorem~\ref{TheDoob} to $Z(t\wedge T_K)$ yielding
\begin{eqnarray}
1 = \langle Z(t\wedge T_K)\rangle =   \Big \langle e^{zK(t\wedge T_K) + (t\wedge T_K) f(z)  }  \Big\rangle .
\end{eqnarray}   
Since  $f(z)<0$,
\begin{eqnarray}
e^{zK(t\wedge T_K) + (t\wedge T_K)  f(z) }  < {\rm max}\left\{e^{z \ell_+ (1+o_{\ell_{\rm min}}(1))}, e^{-z \ell_- (1+o_{\ell_{\rm min}}(1))}\right\}.
\end{eqnarray} 
 Hence, the bounded convergence theorem applies, see e.g.~Ref.~\cite{tao2011introduction}, and we can take the limit $t\rightarrow \infty$ under the expectation value to obtain
\begin{eqnarray}
1 &=&  \langle \lim_{t\rightarrow \infty} e^{zK(t\wedge T_K) + (t\wedge T_K)  f(z) } \rangle \\ 
&=&\Big \langle 1_{T_K<\infty}1_{D_K=1} e^{z K(T_K) + T_K f(z) }   +  1_{T_K<\infty}1_{D_K=-1} e^{zK(T_K)  +T_Kf(z) }   \Big\rangle.  \nonumber\\ \label{eq:DoobK}
\end{eqnarray}    
For large threshold values $-\ell_-$ and $\ell_+$, we have that 
\begin{eqnarray}
K(T_K)  = -\ell_- (1+ o_{\ell_{\rm min}}(1))&{\rm if}&  D_K=-1 \label{K:1}
\end{eqnarray}
and 
\begin{eqnarray}
K(T_K)  = \ell_+  (1+ o_{\ell_{\rm min}}(1))&{\rm if}&  D_K=1.  \label{K:2}
\end{eqnarray}
Substitution of Eqs.~(\ref{K:1}) and (\ref{K:2}) in Eq.~(\ref{eq:DoobK})  gives readily the equality (\ref{eq:important2xJumpa}), which completes the proof.
\end{proof}

In what follows, we use the martingale equality (\ref{eq:important2xJumpa}) to  determine the splitting probabilities $p_-$ and $p^\dagger_+$ of $T_J$.
\subsection{ The first-passage time $T_K$ is with probability one finite}

\begin{proposition}\label{propTJumpa} It holds that $T_K$ is almost surely finite, i.e., 
\begin{eqnarray}
p_-+p_+ = 1. \label{eq:finiteJumpa}
\end{eqnarray} 
\end{proposition}
\begin{proof}
We take the limit $z\rightarrow 0^+$ in Eq.~(\ref{eq:important2xJumpa}).   Since
$f(0)=0$ and we assumed that the process has positive drift,  $f'(0) = -\sum_{j\in\mathbb{Z}}\Delta_j k_j <0$, there exists a
$z_0>0$ such that $f(z)\leq 0$ for all $z\in[0,z_0]$.   Hence, for $z\in[0,z_0]$
the argument in the expectation value is bounded from above by
$e^{\ell_+(1+o_{\ell_{\rm min}}(1))}$, so that the bounded convergence theorem
applies, see e.g.~Ref.~\cite{tao2011introduction}, and
  \begin{eqnarray}
1  &=&  \lim_{z\rightarrow 0^+}  \Big\langle  1_{T_K<\infty}1_{D_K=1} e^{z\ell_+(1+o_{\ell_{\rm min}}(1)) + f(z) T_K }   +  1_{T_K<\infty}1_{D_K=-1} e^{-z\ell_-(1+o_{\ell_{\rm min}}(1)) +f(z) T_K }  \Big\rangle \nonumber\\
 &=& \langle  1_{T_K<\infty}1_{D_K=1}+  1_{T_K<\infty}1_{D_K=-1} \rangle  \nonumber\\ &=&  \langle 1_{T_K<\infty} \rangle  = \mathbb{P}[T_K<\infty], \nonumber
\end{eqnarray}
where we have used  $f(0) = 0$.

\end{proof}

 \subsection{Derivation of the Eq.~(\ref{eq:pp})  for the splitting probabilities} \label{App:SplitEa}

  \begin{proposition}   \label{anotherProp}
  
  If $z^\ast$ is a nonzero solution to the equation 
  \begin{equation}
  f(z^\ast)= \sum_{j\in\mathbb{Z}} (1-e^{z^\ast \Delta_j})k_j = 0, \label{eq:fzAsta}
  \end{equation}
  then 
  \begin{eqnarray}
p_+ = \frac{1 -e^{\ell_- z^\ast (1+o_{\ell_{\rm min}}(1))} }{ 1- e^{(\ell_+ + \ell_-) z^\ast (1+o_{\ell_{\rm min}}(1))}}   \label{eq:P-1Appax}
\end{eqnarray}   
and
  \begin{eqnarray}
 p_- =e^{\ell_- z^\ast(1+o_{\ell_{\rm min}}(1))}  \frac{1 -e^{\ell_+z^\ast (1+o_{\ell_{\rm min}}(1))} }{ 1- e^{(\ell_+ + \ell_-)z^\ast (1+o_{\ell_{\rm min}}(1))}}   . \label{eq:P-1Appa}
\end{eqnarray}   
\end{proposition} 
\begin{proof}  
The proof is similar to the one of Proposition~\ref{prop:Aa}.     The process $e^{K(t)z^\ast}$ is a martingale as it is of the form Eq.~(\ref{eq:martingalezK}).
Applying Theorem~\ref{TheDoob} to the martingale $e^{K(t)z^\ast}$
yields, for every $t\geq 0$,
\begin{eqnarray}
\Big \langle e^{K(t\wedge T_K)z^\ast} \Big\rangle = 1.
\end{eqnarray}  
Up to the stopping time the process obeys $K(t\wedge T_K)\in[-\ell_-,\ell_+]$, and
therefore the random variable $e^{K(t\wedge T_K)z^\ast}$ is bounded, uniformly in
$t$, by ${\rm max}\left\{e^{\ell_+ z^\ast}, e^{-\ell_- z^\ast}\right\}$.
Moreover, $T_K<\infty$ almost surely by Proposition~\ref{propTJumpa}, and hence
$K(t\wedge T_K)\rightarrow K(T_K)$ almost surely as $t\rightarrow\infty$.   The
bounded convergence theorem, see e.g.~Ref.~\cite{tao2011introduction}, thus
permits taking the limit under the expectation value,
\begin{eqnarray}
1 = \lim_{t\rightarrow\infty}\Big \langle e^{K(t\wedge T_K)z^\ast} \Big\rangle
= \Big \langle e^{K(T_K)z^\ast} \Big\rangle .
\end{eqnarray}
Using the Eqs.~(\ref{K:1}) and (\ref{K:2}) to replace $K$ at the stopping time
with either $-\ell_-$ or $\ell_+$, we obtain
\begin{eqnarray}
p_- e^{-\ell_- z^\ast (1+o_{\ell_{\rm min}}(1))} + p_+ e^{\ell_+ z^\ast (1+o_{\ell_{\rm min}}(1))} = 1. \label{eq:xxa}
\end{eqnarray}
The solutions to the Eqs.~(\ref{eq:finiteJumpa}) and (\ref{eq:xxa}) are given by Eqs.~(\ref{eq:P-1Appax}) and (\ref{eq:P-1Appa}), which completes the proof.
\end{proof}

\begin{proposition}  
  If $z^\ast$ is a nonzero solution to the equation 
  \begin{equation}
  f(z^\ast)= \sum_{j\in\mathbb{Z}} (1-e^{z^\ast \Delta_j})k_j = 0, \label{eq:fzAsta2}
  \end{equation}
  then 
  \begin{eqnarray}
p^\dagger_+ = e^{\ell_+ z^\ast (1+o_{\ell_{\rm min}}(1))  }  \frac{1 -e^{\ell_- z^\ast (1+o_{\ell_{\rm min}}(1))  } }{ 1- e^{(\ell_+ + \ell_-) z^\ast (1+o_{\ell_{\rm min}}(1))  }} \label{eq:P-2Appax}
\end{eqnarray}   
and 
  \begin{eqnarray}
 p^\dagger_- = \frac{1 -e^{\ell_+z^\ast (1+o_{\ell_{\rm min}}(1))  } }{ 1- e^{(\ell_+ + \ell_-)z^\ast (1+o_{\ell_{\rm min}}(1))  }}   .\label{eq:P-2Appa}
\end{eqnarray}   
\end{proposition}
\begin{proof}  
Applying the Proposition~\ref{anotherProp} to the  Eq.~(\ref{eq:PDaggerP}),  and using the fact that $z^\ast$ is independent of the threshold values $\ell_-$ and $\ell_+$, we readily obtain the equalities~(\ref{eq:P-2Appax}) and (\ref{eq:P-2Appa}).
\end{proof}

  \begin{proposition}   
The splitting probabilities $p_-$ and $p^\dagger_+$ obey the equality Eq.~(\ref{eq:pp}).
  \end{proposition}   
\begin{proof}  
Eq.~(\ref{eq:pp}) follows readily from Eqs.~(\ref{eq:P-1Appa}) and (\ref{eq:P-2Appax}). 

\end{proof}

\bibliography{biblio}

\begin{thebibliography}{10}
\providecommand{\url}[1]{\texttt{#1}}
\providecommand{\urlprefix}{URL }
\expandafter\ifx\csname urlstyle\endcsname\relax
  \providecommand{\doi}[1]{doi:\discretionary{}{}{}#1}\else
  \providecommand{\doi}{doi:\discretionary{}{}{}\begingroup
  \urlstyle{rm}\Url}\fi
\providecommand{\eprint}[2][]{\url{#2}}

\bibitem{hanggi1990reaction}
P.~H\"anggi, P.~Talkner and M.~Borkovec,
\newblock \emph{Reaction-rate theory: fifty years after kramers},
\newblock Rev. Mod. Phys. \textbf{62}, 251 (1990),
\newblock \doi{10.1103/RevModPhys.62.251}.

\bibitem{mccann1999thermally}
L.~I. McCann, M.~Dykman and B.~Golding,
\newblock \emph{Thermally activated transitions in a bistable three-dimensional
  optical trap},
\newblock Nature \textbf{402}(6763), 785 (1999),
\newblock \doi{10.1038/45492}.

\bibitem{kramers1940brownian}
H.~A. Kramers,
\newblock \emph{Brownian motion in a field of force and the diffusion model of
  chemical reactions},
\newblock Physica \textbf{7}(4), 284 (1940),
\newblock \doi{10.1016/S0031-8914(40)90098-2}.

\bibitem{angelani2014first}
L.~Angelani, R.~Di~Leonardo and M.~Paoluzzi,
\newblock \emph{First-passage time of run-and-tumble particles},
\newblock The European Physical Journal E \textbf{37}(7), 1 (2014),
\newblock \doi{10.1140/epje/i2014-14059-4}.

\bibitem{malakar2018steady}
K.~Malakar, V.~Jemseena, A.~Kundu, K.~V. Kumar, S.~Sabhapandit, S.~N. Majumdar,
  S.~Redner and A.~Dhar,
\newblock \emph{Steady state, relaxation and first-passage properties of a
  run-and-tumble particle in one-dimension},
\newblock Journal of Statistical Mechanics: Theory and Experiment
  \textbf{2018}(4), 043215 (2018),
\newblock \doi{10.1088/1742-5468/aab84f}.

\bibitem{dhar2019run}
A.~Dhar, A.~Kundu, S.~N. Majumdar, S.~Sabhapandit and G.~Schehr,
\newblock \emph{Run-and-tumble particle in one-dimensional confining
  potentials: Steady-state, relaxation, and first-passage properties},
\newblock Phys. Rev. E \textbf{99}, 032132 (2019),
\newblock \doi{10.1103/PhysRevE.99.032132}.

\bibitem{biswas2020first}
A.~Biswas, J.~Cruz, P.~Parmananda and D.~Das,
\newblock \emph{First passage of an active particle in the presence of passive
  crowders},
\newblock Soft Matter \textbf{16}(26), 6138 (2020),
\newblock \doi{10.1039/D0SM00350F}.

\bibitem{walter2021first}
B.~Walter, G.~Pruessner and G.~Salbreux,
\newblock \emph{First passage time distribution of active thermal particles in
  potentials},
\newblock Phys. Rev. Research \textbf{3}, 013075 (2021),
\newblock \doi{10.1103/PhysRevResearch.3.013075}.

\bibitem{PhysRevE.73.061109}
D.~Ryvkine and M.~I. Dykman,
\newblock \emph{Pathways of activated escape in periodically modulated
  systems},
\newblock Phys. Rev. E \textbf{73}, 061109 (2006),
\newblock \doi{10.1103/PhysRevE.73.061109}.

\bibitem{godec2016active}
A.~Godec and R.~Metzler,
\newblock \emph{Active transport improves the precision of linear long distance
  molecular signalling},
\newblock Journal of Physics A: Mathematical and Theoretical \textbf{49}(36),
  364001 (2016),
\newblock \doi{10.1088/1751-8113/49/36/364001}.

\bibitem{loverdo2008enhanced}
C.~Loverdo, O.~B{\'e}nichou, M.~Moreau and R.~Voituriez,
\newblock \emph{Enhanced reaction kinetics in biological cells},
\newblock Nature physics \textbf{4}(2), 134 (2008),
\newblock \doi{10.1038/nphys830}.

\bibitem{siggia2013decisions}
E.~D. Siggia and M.~Vergassola,
\newblock \emph{Decisions on the fly in cellular sensory systems},
\newblock Proceedings of the National Academy of Sciences \textbf{110}(39),
  E3704 (2013),
\newblock \doi{10.1073/pnas.1314081110}.

\bibitem{desponds2020mechanism}
J.~Desponds, M.~Vergassola and A.~M. Walczak,
\newblock \emph{A mechanism for hunchback promoters to readout morphogenetic
  positional information in less than a minute},
\newblock Elife \textbf{9}, e49758 (2020),
\newblock \doi{10.7554/eLife.49758}.

\bibitem{biswas2021first}
K.~Biswas and A.~Ghosh,
\newblock \emph{First passage time in post-transcriptional regulation by
  multiple small rnas},
\newblock The European Physical Journal E \textbf{44}(2), 1 (2021),
\newblock \doi{10.1140/epje/s10189-021-00028-7}.

\bibitem{roldan2015decision}
E.~Rold\'an, I.~Neri, M.~D\"orpinghaus, H.~Meyr and F.~J\"ulicher,
\newblock \emph{Decision making in the arrow of time},
\newblock Phys. Rev. Lett. \textbf{115}, 250602 (2015),
\newblock \doi{10.1103/PhysRevLett.115.250602}.

\bibitem{gringich2017bis}
T.~R. Gingrich and J.~M. Horowitz,
\newblock \emph{Fundamental bounds on first passage time fluctuations for
  currents},
\newblock Phys. Rev. Lett. \textbf{119}, 170601 (2017),
\newblock \doi{10.1103/PhysRevLett.119.170601}.

\bibitem{maes2003origin}
C.~Maes,
\newblock \emph{On the origin and the use of fluctuation relations for the
  entropy},
\newblock S{\'e}minaire Poincar{\'e} \textbf{2}, 29 (2003),
\newblock \doi{10.1007/978-3-0348-7932-3_8}.

\bibitem{jarzynski2011equalities}
C.~Jarzynski,
\newblock \emph{Equalities and inequalities: Irreversibility and the second law
  of thermodynamics at the nanoscale},
\newblock Annu. Rev. Condens. Matter Phys. \textbf{2}(1), 329 (2011),
\newblock \doi{10.1146/annurev-conmatphys-062910-140506}.

\bibitem{seifert2012stochastic}
U.~Seifert,
\newblock \emph{Stochastic thermodynamics, fluctuation theorems and molecular
  machines},
\newblock Reports on Progress in Physics \textbf{75}(12), 126001 (2012),
\newblock \doi{10.1088/0034-4885/75/12/126001}.

\bibitem{chetrite2011two}
R.~Chetrite and S.~Gupta,
\newblock \emph{Two refreshing views of fluctuation theorems through kinematics
  elements and exponential martingale},
\newblock Journal of Statistical Physics \textbf{143}(3), 543 (2011),
\newblock \doi{10.1007/s10955-011-0184-0}.

\bibitem{neri2017statistics}
I.~Neri, E.~Rold\'an and F.~J\"ulicher,
\newblock \emph{Statistics of infima and stopping times of entropy production
  and applications to active molecular processes},
\newblock Phys. Rev. X \textbf{7}, 011019 (2017),
\newblock \doi{10.1103/PhysRevX.7.011019}.

\bibitem{pigolotti2017generic}
S.~Pigolotti, I.~Neri, E.~Rold\'an and F.~J\"ulicher,
\newblock \emph{Generic properties of stochastic entropy production},
\newblock Phys. Rev. Lett. \textbf{119}, 140604 (2017),
\newblock \doi{10.1103/PhysRevLett.119.140604}.

\bibitem{neri2019integral}
I.~Neri, {\'E}.~Rold{\'a}n, S.~Pigolotti and F.~J{\"u}licher,
\newblock \emph{Integral fluctuation relations for entropy production at
  stopping times},
\newblock Journal of Statistical Mechanics: Theory and Experiment
  \textbf{2019}(10), 104006 (2019),
\newblock \doi{10.1088/1742-5468/ab40a0}.

\bibitem{schnakenberg1976network}
J.~Schnakenberg,
\newblock \emph{Network theory of microscopic and macroscopic behavior of
  master equation systems},
\newblock Rev. Mod. Phys. \textbf{48}, 571 (1976),
\newblock \doi{10.1103/RevModPhys.48.571}.

\bibitem{maes2000definition}
C.~Maes, F.~Redig and A.~V. Moffaert,
\newblock \emph{On the definition of entropy production, via examples},
\newblock Journal of mathematical physics \textbf{41}(3), 1528 (2000),
\newblock \doi{10.1063/1.533195}.

\bibitem{yang2020unified}
Y.-J. Yang and H.~Qian,
\newblock \emph{Unified formalism for entropy production and fluctuation
  relations},
\newblock Phys. Rev. E \textbf{101}, 022129 (2020),
\newblock \doi{10.1103/PhysRevE.101.022129}.

\bibitem{peliti2021stochastic}
L.~Peliti and S.~Pigolotti,
\newblock \emph{Stochastic Thermodynamics: An Introduction},
\newblock Princeton University Press,
\newblock ISBN 9780691201771 (2021).

\bibitem{maes2020local}
C.~Maes,
\newblock \emph{Local detailed balance},
\newblock SciPost Physics Lecture Notes p. 032 (2021),
\newblock \doi{10.21468/SciPostPhysLectNotes.32}.

\bibitem{hartich2021violation}
D.~Hartich and A.~Godec,
\newblock \emph{Violation of local detailed balance despite a clear time-scale
  separation},
\newblock arXiv:2111.14734  (2021).

\bibitem{touchette2009large}
H.~Touchette,
\newblock \emph{The large deviation approach to statistical mechanics},
\newblock Physics Reports \textbf{478}(1-3), 1 (2009),
\newblock \doi{10.1016/j.physrep.2009.05.002}.

\bibitem{barato2015formal}
A.~C. Barato and R.~Chetrite,
\newblock \emph{A formal view on level 2.5 large deviations and fluctuation
  relations},
\newblock Journal of Statistical Physics \textbf{160}(5), 1154 (2015),
\newblock \doi{10.1007/s10955-015-1283-0}.

\bibitem{pietzonka2016universalx}
P.~Pietzonka, A.~C. Barato and U.~Seifert,
\newblock \emph{Universal bounds on current fluctuations},
\newblock Phys. Rev. E \textbf{93}, 052145 (2016),
\newblock \doi{10.1103/PhysRevE.93.052145}.

\bibitem{gingrich2016dissipation}
T.~R. Gingrich, J.~M. Horowitz, N.~Perunov and J.~L. England,
\newblock \emph{Dissipation bounds all steady-state current fluctuations},
\newblock Phys. Rev. Lett. \textbf{116}, 120601 (2016),
\newblock \doi{10.1103/PhysRevLett.116.120601}.

\bibitem{pietzonka2016affinity}
P.~Pietzonka, A.~C. Barato and U.~Seifert,
\newblock \emph{Affinity-and topology-dependent bound on current fluctuations},
\newblock Journal of Physics A: Mathematical and Theoretical \textbf{49}(34),
  34LT01 (2016),
\newblock \doi{10.1088/1751-8113/49/34/34LT01}.

\bibitem{saito2016waiting}
K.~Saito and A.~Dhar,
\newblock \emph{Waiting for rare entropic fluctuations},
\newblock EPL (Europhysics Letters) \textbf{114}(5), 50004 (2016),
\newblock \doi{10.1209/0295-5075/114/50004}.

\bibitem{Melsa}
J.~L. Melsa and D.~L. Cohn,
\newblock \emph{Decision and Estimation Theory},
\newblock New York: McGraw-Hill,
\newblock ISBN 0070414688 (1978).

\bibitem{tartakovsky2014sequential}
A.~Tartakovsky, I.~Nikiforov and M.~Basseville,
\newblock \emph{Sequential analysis: Hypothesis testing and changepoint
  detection},
\newblock CRC Press (2014).

\bibitem{wald1945sequential}
A.~Wald,
\newblock \emph{Sequential tests of statistical hypotheses},
\newblock The annals of mathematical statistics \textbf{16}(2), 117 (1945),
\newblock \doi{10.1214/aoms/1177731118}.

\bibitem{wald1948optimum}
A.~Wald and J.~Wolfowitz,
\newblock \emph{Optimum character of the sequential probability ratio test},
\newblock The Annals of Mathematical Statistics pp. 326--339 (1948),
\newblock \doi{10.1214/aoms/1177730197}.

\bibitem{lai1981asymptotic}
T.~L. Lai,
\newblock \emph{Asymptotic optimality of invariant sequential probability ratio
  tests},
\newblock The Annals of Statistics pp. 318--333 (1981),
\newblock \doi{10.1214/aos/1176345398}.

\bibitem{PhysRevLett.125.120604}
G.~Falasco and M.~Esposito,
\newblock \emph{Dissipation-time uncertainty relation},
\newblock Phys. Rev. Lett. \textbf{125}, 120604 (2020),
\newblock \doi{10.1103/PhysRevLett.125.120604}.

\bibitem{PhysRevLett.128.050603}
L.-L. Yan, J.-W. Zhang, M.-R. Yun, J.-C. Li, G.-Y. Ding, J.-F. Wei, J.-T. Bu,
  B.~Wang, L.~Chen, S.-L. Su, F.~Zhou, Y.~Jia \emph{et~al.},
\newblock \emph{Experimental verification of dissipation-time uncertainty
  relation},
\newblock Phys. Rev. Lett. \textbf{128}, 050603 (2022),
\newblock \doi{10.1103/PhysRevLett.128.050603}.

\bibitem{barato2015thermodynamic}
A.~C. Barato and U.~Seifert,
\newblock \emph{Thermodynamic uncertainty relation for biomolecular processes},
\newblock Phys. Rev. Lett. \textbf{114}, 158101 (2015),
\newblock \doi{10.1103/PhysRevLett.114.158101}.

\bibitem{busiello2019hyperaccurate}
D.~M. Busiello and S.~Pigolotti,
\newblock \emph{Hyperaccurate currents in stochastic thermodynamics},
\newblock Phys. Rev. E \textbf{100}, 060102 (2019),
\newblock \doi{10.1103/PhysRevE.100.060102}.

\bibitem{lebowitz1999gallavotti}
J.~L. Lebowitz and H.~Spohn,
\newblock \emph{A gallavotti--cohen-type symmetry in the large deviation
  functional for stochastic dynamics},
\newblock Journal of Statistical Physics \textbf{95}(1-2), 333 (1999),
\newblock \doi{10.1023/A:1004589714161}.

\bibitem{barato2012gallavotti}
A.~C. Barato, R.~Chetrite, H.~Hinrichsen and D.~Mukamel,
\newblock \emph{A gallavotti-cohen-evans-morriss like symmetry for a class of
  markov jump processes},
\newblock Journal of Statistical Physics \textbf{146}(2), 294 (2012),
\newblock \doi{10.1007/s10955-011-0389-2}.

\bibitem{barato2012symmetry}
A.~Barato and R.~Chetrite,
\newblock \emph{On the symmetry of current probability distributions in jump
  processes},
\newblock Journal of Physics A: Mathematical and Theoretical \textbf{45}(48),
  485002 (2012),
\newblock \doi{10.1088/1751-8113/45/48/485002}.

\bibitem{gaspard2013multivariate}
P.~Gaspard,
\newblock \emph{Multivariate fluctuation relations for currents},
\newblock New Journal of Physics \textbf{15}(11), 115014 (2013),
\newblock \doi{10.1088/1367-2630/15/11/115014}.

\bibitem{polettini2019effective}
M.~Polettini and M.~Esposito,
\newblock \emph{Effective fluctuation and response theory},
\newblock Journal of Statistical Physics \textbf{176}(1), 94 (2019),
\newblock \doi{10.1007/s10955-019-02291-7}.

\bibitem{pietzonka2017finite}
P.~Pietzonka, F.~Ritort and U.~Seifert,
\newblock \emph{Finite-time generalization of the thermodynamic uncertainty
  relation},
\newblock Phys. Rev. E \textbf{96}, 012101 (2017),
\newblock \doi{10.1103/PhysRevE.96.012101}.

\bibitem{horowitz2017proof}
J.~M. Horowitz and T.~R. Gingrich,
\newblock \emph{Proof of the finite-time thermodynamic uncertainty relation for
  steady-state currents},
\newblock Phys. Rev. E \textbf{96}, 020103 (2017),
\newblock \doi{10.1103/PhysRevE.96.020103}.

\bibitem{proesmans2017discrete}
K.~Proesmans and C.~Van~den Broeck,
\newblock \emph{Discrete-time thermodynamic uncertainty relation},
\newblock EPL (Europhysics Letters) \textbf{119}(2), 20001 (2017),
\newblock \doi{10.1209/0295-5075/119/20001}.

\bibitem{hasegawa2019fluctuation}
Y.~Hasegawa and T.~Van~Vu,
\newblock \emph{Fluctuation theorem uncertainty relation},
\newblock Phys. Rev. Lett. \textbf{123}, 110602 (2019),
\newblock \doi{10.1103/PhysRevLett.123.110602}.

\bibitem{shreshtha2019thermodynamic}
M.~Shreshtha and R.~J. Harris,
\newblock \emph{Thermodynamic uncertainty for run-and-tumble--type processes},
\newblock EPL (Europhysics Letters) \textbf{126}(4), 40007 (2019),
\newblock \doi{/10.1209/0295-5075/126/40007}.

\bibitem{dechant2020fluctuation}
A.~Dechant and S.-i. Sasa,
\newblock \emph{Fluctuation--response inequality out of equilibrium},
\newblock Proceedings of the National Academy of Sciences \textbf{117}(12),
  6430 (2020),
\newblock \doi{10.1073/pnas.1918386117}.

\bibitem{falasco2020unifying}
G.~Falasco, M.~Esposito and J.-C. Delvenne,
\newblock \emph{Unifying thermodynamic uncertainty relations},
\newblock New Journal of Physics  (2020),
\newblock \doi{10.1088/1367-2630/ab8679}.

\bibitem{polettini2016tightening}
M.~Polettini, A.~Lazarescu and M.~Esposito,
\newblock \emph{Tightening the uncertainty principle for stochastic currents},
\newblock Physical Review E \textbf{94}(5), 052104 (2016),
\newblock \doi{10.1103/PhysRevE.94.052104}.

\bibitem{koyuk2020thermodynamic}
T.~Koyuk and U.~Seifert,
\newblock \emph{Thermodynamic uncertainty relation for time-dependent driving},
\newblock Physical Review Letters \textbf{125}(26), 260604 (2020),
\newblock \doi{10.1103/PhysRevLett.125.260604}.

\bibitem{chetrite2019martingale}
R.~Ch{\'e}trite, S.~Gupta, I.~Neri and {\'E}.~Rold{\'a}n,
\newblock \emph{Martingale theory for housekeeping heat},
\newblock EPL (Europhysics Letters) \textbf{124}(6), 60006 (2019),
\newblock \doi{10.1209/0295-5075/124/60006}.

\bibitem{neri2020second}
I.~Neri,
\newblock \emph{Second law of thermodynamics at stopping times},
\newblock Phys. Rev. Lett. \textbf{124}, 040601 (2020),
\newblock \doi{10.1103/PhysRevLett.124.040601}.

\bibitem{PhysRevLett.126.080603}
G.~Manzano, D.~Subero, O.~Maillet, R.~Fazio, J.~P. Pekola and E.~Rold\'an,
\newblock \emph{Thermodynamics of gambling demons},
\newblock Phys. Rev. Lett. \textbf{126}, 080603 (2021),
\newblock \doi{10.1103/PhysRevLett.126.080603}.

\bibitem{ciliberto2017experiments}
S.~Ciliberto,
\newblock \emph{Experiments in stochastic thermodynamics: Short history and
  perspectives},
\newblock Physical Review X \textbf{7}(2), 021051 (2017),
\newblock \doi{10.1103/PhysRevX.7.021051}.

\bibitem{pietzonka2016universal}
P.~Pietzonka, A.~C. Barato and U.~Seifert,
\newblock \emph{Universal bound on the efficiency of molecular motors},
\newblock Journal of Statistical Mechanics: Theory and Experiment
  \textbf{2016}(12), 124004 (2016),
\newblock \doi{0.1088/1742-5468/2016/12/124004}.

\bibitem{gingrich2017inferring}
T.~R. Gingrich, G.~M. Rotskoff and J.~M. Horowitz,
\newblock \emph{Inferring dissipation from current fluctuations},
\newblock Journal of Physics A: Mathematical and Theoretical \textbf{50}(18),
  184004 (2017),
\newblock \doi{10.1088/1751-8121/aa672f}.

\bibitem{seifert2019stochastic}
U.~Seifert,
\newblock \emph{From stochastic thermodynamics to thermodynamic inference},
\newblock Annual Review of Condensed Matter Physics \textbf{10}, 171 (2019),
\newblock \doi{10.1146/annurev-conmatphys-031218-013554}.

\bibitem{van2020entropy}
T.~Van~Vu, V.~T. Vo and Y.~Hasegawa,
\newblock \emph{Entropy production estimation with optimal current},
\newblock Phys. Rev. E \textbf{101}, 042138 (2020),
\newblock \doi{10.1103/PhysRevE.101.042138}.

\bibitem{manikandan2020inferring}
S.~K. Manikandan, D.~Gupta and S.~Krishnamurthy,
\newblock \emph{Inferring entropy production from short experiments},
\newblock Phys. Rev. Lett. \textbf{124}, 120603 (2020),
\newblock \doi{10.1103/PhysRevLett.124.120603}.

\bibitem{hopfield1974kinetic}
J.~J. Hopfield,
\newblock \emph{Kinetic proofreading: a new mechanism for reducing errors in
  biosynthetic processes requiring high specificity},
\newblock Proceedings of the National Academy of Sciences \textbf{71}(10), 4135
  (1974),
\newblock \doi{10.1073/pnas.71.10.4135}.

\bibitem{murugan2012speed}
A.~Murugan, D.~A. Huse and S.~Leibler,
\newblock \emph{Speed, dissipation, and error in kinetic proofreading},
\newblock Proceedings of the National Academy of Sciences \textbf{109}(30),
  12034 (2012),
\newblock \doi{10.1073/pnas.1119911109}.

\bibitem{mallory2020we}
J.~D. Mallory, O.~A. Igoshin and A.~B. Kolomeisky,
\newblock \emph{Do we understand the mechanisms used by biological systems to
  correct their errors?},
\newblock The Journal of Physical Chemistry B \textbf{124}(42), 9289 (2020),
\newblock \doi{10.1021/acs.jpcb.0c06180}.

\bibitem{gao2021principles}
C.~Y. Gao and D.~T. Limmer,
\newblock \emph{Principles of low dissipation computing from a stochastic
  circuit model},
\newblock Physical Review Research \textbf{3}(3), 033169 (2021),
\newblock \doi{10.1103/PhysRevResearch.3.033169}.

\bibitem{doob1953stochastic}
J.~L. Doob,
\newblock \emph{Stochastic processes}, vol. 101,
\newblock New York Wiley (1953).

\bibitem{liptser2013statistics}
R.~S. Liptser and A.~N. Shiryayev,
\newblock \emph{Statistics of random Processes I: General Theory}, vol.~5 of
  \emph{Stochastic Modelling and Applied Probability},
\newblock Springer-Verlag New York, Inc., 1st edn.,
\newblock Translated by A. B. Aries (1977).

\bibitem{gardiner1985handbook}
C.~W. Gardiner \emph{et~al.},
\newblock \emph{Handbook of stochastic methods}, vol.~3,
\newblock springer Berlin (1985).

\bibitem{risken1996fokker}
H.~Risken,
\newblock \emph{Fokker-planck equation},
\newblock Springer (1996).

\bibitem{seifert2005entropy}
U.~Seifert,
\newblock \emph{Entropy production along a stochastic trajectory and an
  integral fluctuation theorem},
\newblock Phys. Rev. Lett. \textbf{95}, 040602 (2005),
\newblock \doi{10.1103/PhysRevLett.95.040602}.

\bibitem{grebenkov2014first}
D.~S. Grebenkov,
\newblock \emph{First exit times of harmonically trapped particles: a didactic
  review},
\newblock Journal of Physics A: Mathematical and Theoretical \textbf{48}(1),
  013001 (2014),
\newblock \doi{10.1088/1751-8113/48/1/013001}.

\bibitem{protter2005stochastic}
P.~E. Protter,
\newblock \emph{Stochastic integration and differential equations},
\newblock Springer-Verlag (2004).

\bibitem{tao2011introduction}
T.~Tao,
\newblock \emph{An introduction to measure theory},
\newblock American Mathematical Society Providence, RI (2011).

\end{thebibliography}

\end{document}